\documentclass[12pt,letterpaper]{article}%
\usepackage{amsmath}
\usepackage{amsthm}
\usepackage{amssymb}
\usepackage{amsfonts}
\usepackage{mathpazo}
\usepackage{times}
\usepackage[T1]{fontenc}
\usepackage[authoryear]{natbib}
\usepackage{titlesec}
\usepackage{xcolor}
\usepackage{hyperref}
\usepackage[normalem]{ulem}
\usepackage{graphicx}
\usepackage[tight]{subfigure}
\usepackage{epstopdf}
\usepackage{caption}%
\providecommand{\U}[1]{\protect\rule{.1in}{.1in}}
\makeatletter
\long
\def\@makecaption#1#2{  \vskip\abovecaptionskip
\sbox\@tempboxa{{\captionfonts #1: #2}}  \ifdim \wd\@tempboxa >\hsize
{\captionfonts #1: #2\par}
\else
\hbox to\hsize{\hfil\box\@tempboxa\hfil}  \fi
\vskip\belowcaptionskip}
\makeatother
\makeatletter
\@addtoreset{equation}{section}
\makeatother
\renewcommand{\theequation}{{\thesection}.{\arabic{equation}}}
\titleformat{\subsection}{\normalfont\bfseries}{\thesubsection}{0.5em}{}
\titleformat{\subsubsection}{\normalfont\bfseries\itshape}{\thesubsubsection}{0.5em}{}
\titleformat{\paragraph}[runin]{\normalfont\bfseries}{\theparagraph}{0.5em}{}[.]
\renewcommand{\thesubsubsection}{\arabic{section}.\arabic{subsection}.\arabic{subsubsection}}
\renewcommand{\theparagraph}{(\arabic{paragraph})}
\newtheoremstyle{thm}{1.0em}{1.0em}{\itshape}{}{\bfseries}{.}{.5em}{}
\theoremstyle{thm}
\newtheorem{theorem}{Theorem}[section]
\newtheorem{corollary}{Corollary}[section]
\newtheorem{lemma}{Lemma}[section]
\newtheorem{proposition}{Proposition}[section]
\newtheoremstyle{rem}{1.0em}{1.0em}{\normalfont}{}{\bfseries}{.}{.5em}{}
\theoremstyle{rem}
\newtheorem{remark}{Remark}[section]

\newenvironment{proof}[1][Proof]{\par\vspace{1.0em}\noindent\textbf{#1. }}{\hfill $\Box$\par\vspace{1.0em}}
\hypersetup{
bookmarks = true,
bookmarksnumbered = true,
bookmarksopen = true,
bookmarksopenlevel = 3,
pdfstartview = FitH,
colorlinks = false,
citebordercolor = {white},
linkbordercolor = {red}
}
\newcommand{\captionfonts}{\small}

\renewcommand{\qed}{\hfill $\Box$}
\AtBeginDocument{
\DeclareSymbolFont{AMSb}{U}{msb}{m}{n}
\DeclareSymbolFontAlphabet{\mathbb}{AMSb}
}
\ifx\pdfoutput\relax\let\pdfoutput=\undefined\fi
\newcount\msipdfoutput
\ifx\pdfoutput\undefined\else
\ifcase\pdfoutput\else
\ifx\paperwidth\undefined\else
\ifdim\paperheight=0pt\relax\else\pdfpageheight\paperheight\fi
\ifdim\paperwidth=0pt\relax\else\pdfpagewidth\paperwidth\fi
\fi\fi\fi
\begin{document}
\hspace{13.9cm}1

\ \vspace{1mm}\newline

\noindent{\Large \textbf{Information-Theoretic Bounds and Approximations
\newline in Neural Population Coding}}

\ \newline\textbf{{\large Wentao Huang}}\newline\textit{whuang21@jhmi.edu}
\newline\textit{Department of Biomedical Engineering, Johns Hopkins University
School of Medicine, Baltimore, MD 21205, U.S.A., and Cognitive and Intelligent
Lab of China Electronics Technology Group Corporation, Beijing 100846, China.}
\medskip\newline\textbf{{\large Kechen Zhang}}\newline%
\textit{kzhang4@jhmi.edu\newline Department of Biomedical Engineering, Johns
Hopkins University School of Medicine, Baltimore, MD 21205, U.S.A.}\newline

\thispagestyle{empty} \markboth{}{Information-Theoretic Bounds} \ \vspace
{-20mm}\newline

\begin{center}
\textbf{Abstract}
\end{center}

\textbf{While Shannon's mutual information has wide spread applications in
many disciplines, for practical applications it is often difficult to
calculate its value accurately for high-dimensional variables because of the
curse of dimensionality. This paper is focused on effective approximation
methods for evaluating mutual information in the context of neural population
coding. For large but finite neural populations, we derive several
information-theoretic asymptotic bounds and approximation formulas that remain
valid in high-dimensional spaces. We prove that optimizing the population
density distribution based on these approximation formulas is a convex
optimization problem which allows efficient numerical solutions. Numerical
simulation results confirmed that our asymptotic formulas were highly accurate
for approximating mutual information for large neural populations. In special
cases, the approximation formulas are exactly equal to the true mutual
information. We also discuss techniques of variable transformation and
dimensionality reduction to facilitate computation of the approximations. }

\section{Introduction}

\label{Sec:1} Shannon's mutual information (MI) provides a quantitative
characterization of the association between two random variables by measuring
how much knowing one of the variables reduces uncertainty about the other
\citep{Shannon(1948-mathematical)}. Information theory has become a useful
tool for neuroscience research
\citep{Rieke-book,Borst(1999-information),Pouget(2000-information), Laughlin(2003-communication), Brown(2004-multiple),Quiroga(2009-extracting)},
with applications to various problems such as sensory coding problems in the
visual systems
\citep{Eckhorn(1975-rigorous), Optican(1987-temporal), Atick(1990-towards), McClurkin(1991-lateral), Atick(1992-understanding),Becker(1992-self),VanHateren(1992-real), Gawne(1993-how), Tovee(1993-information), Bell(1997-independent),Lewis(2006-are)}
and the auditory systems \citep{Chechik(2006-reduction),Gourevitch(2007-evaluating),chase2005limited}.

One major problem encountered in practical applications of information theory
is that the exact value of mutual information is often hard to compute in
high-dimensional spaces. For example, suppose we want to calculate the mutual
information between a random stimulus variable that requires many parameters
to specify and the elicited noisy responses of a large population of neurons.
In order to accurately evaluate the mutual information between the stimuli and
the responses, one has to average over all possible stimulus patterns and over
all possible response patterns of the whole population. This averaging quickly
leads to a combinatorial explosion as either the stimulus dimension or the
population size increases.
This problem occurs not only when one computes MI numerically for a given theoretical model 
but also when one estimates MI empirically from experimental data.

Even when the input and output dimensions 
are not that high,
MI estimate from experimental data
tends to  have
 a positive bias due to limited
 sample size
\citep{Miller(1955-note),Treves(1995-upward)}.
For example, a perfectly flat joint probability distribution implies zero MI, 
but an empirical joint distribution with fluctuations due to finite data size
appears to suggest a positive MI.
The error may get much worse as the input and output dimensions increase because a reliable estimate of MI may require exponentially more data points to fill the space
of the joint  distribution.
 Various asymptotic expansion methods have been proposed to reduce the bias in MI estimate
\citep{Miller(1955-note), Carlton(1969-bias), Treves(1995-upward), Victor(2000-asymptotic), Paninski(2003-estimation)}.
Other estimators of MI have also been studied, such as those based on \textit{k}-nearest neighbor
\citep{Kraskov(2004-estimating)} and minimal spanning trees
\citep{Khan(2007-relative)}. However, it is not easy  for these methods
to handle the general situation with high-dimensional inputs and high-dimensional
outputs.

For numerical computation of MI for a given theoretical model, one useful approach is Monte Carlo sampling, a convergent method that may potentially reaches arbitrary accuracy
\citep{Yarrow(2012-fisher)}. However, its stochastic and inefficient
computational scheme makes it unsuitable for many applications. For instance,
to optimize the distribution of a neural population for a given set of
stimuli, one may want to slightly alter the population parameters and see how
the perturbation affects the MI, but a tiny change of MI can be easily drowned
out by the inherent noise in the Monte Carlo method.

An alternative approach is to use information-theoretic bounds and
approximations to simplify calculations. For example, the Cram\'{e}r-Rao lower
bound \citep{Rao(1945-information)} tell us that the inverse of Fisher
information (FI) is a lower bound to the mean square decoding error of any
unbiased decoder. Fisher information is useful for many applications partly
because it is often much easier to calculate than MI \citep[see e.g.,][]{Zhang(1998-interpreting),Zhang(1999-neuronal),Abbott(1999-effect), Bethge(2002-optimal), Harper(2004-optimal), Toyoizumi(2006-fisher)}.

A link between MI and FI has been studied by several researchers
\citep{Clarke(1990-information),Rissanen(1996-fisher),Brunel(1998-mutual),Sompolinsky(2001-population)}.
\cite{Clarke(1990-information)} first derived an asymptotic formula between
the relative entropy and FI for parameter estimation from independent and
identically distributed (i.i.d.) observations with suitable smoothness
conditions. \cite{Rissanen(1996-fisher)} generalized it in the framework of
stochastic complexity for model selection. \cite{Brunel(1998-mutual)}
presented an asymptotic relationship between the MI and FI in the limit of a
large number of neurons. The method was extended to discrete inputs\ by
\cite{Kang(2001-mutual)}. More general discussions about this also appeared in
other papers \citep[e.g.][]{Ganguli(2014-efficient),Wei(2015-mutual)}.
However, for finite population size, the asymptotic formula may lead to large
errors, especially for high-dimensional inputs as detailed in sections
\ref{Sec:2.2} and \ref{Sec:3.1}.

In this paper, our main goal is to improve FI approximations to MI for finite
neural populations especially for high-dimensional inputs. Another goal is to
discuss how to use these approximations to optimize neural population coding.
We will present several information-theoretic bounds and approximation
formulas and discuss the conditions under which they are established\ in
section \ref{Sec:2}, with detailed proofs given in Appendix. We also discuss
how our approximation formulas are related to other statistical estimators and
information-theoretic bounds, such as Cram\'{e}r-Rao bound and van Trees'
Bayesian Cram\'{e}r-Rao bound (section \ref{Sec:3}). In order to better apply
the approximation formulas in high-dimensional input space, we propose some
useful techniques in section \ref{Sec:3.0}, including variable transformation
and dimensionality reduction, which may greatly reduce the computational
complexity for practical applications. Finally, in section \ref{Sec:4}, we
discuss how to use the approximation formulas for the optimization of
information transfer for neural population coding.

\section{Bounds and Approximations for Mutual Information in Neural Population
Coding}

\label{Sec:2}

\subsection{Mutual Information and Notations}

\label{Sec:2.1}Suppose the input $\mathbf{x}$ is a $K$-dimensional vector,
$\mathbf{x}=(x_{1}$,\thinspace$x_{2}$,\thinspace$\cdots$,\thinspace$x_{K}%
)^{T}$, the outputs of $N$ neurons are denoted by a vector, $\mathbf{r}%
=(r_{1}$,\thinspace$r_{2}$,\thinspace$\cdots$,\thinspace$r_{N})^{T}$. In this
paper we denote\ random variables by upper case letters, e.g., random
variables $X$ and $R$, in contrast to their vector values $\mathbf{x}$ and
$\mathbf{r}$. The MI $I\left(  X{\text{;\thinspace}}R\right)  $ (denoted as
$I$ below) between $X$ and $R$ is defined by \citep{Cover(2006-BK-elements)}%
\begin{equation}
I=\int_{{\mathcal{X}}}\int_{{\mathcal{R}}}p(\mathbf{r}|\mathbf{x}%
)p(\mathbf{x})\ln\frac{p(\mathbf{r}|\mathbf{x})}{p(\mathbf{r})}d\mathbf{r}%
d\mathbf{x}\text{,} \label{MI}%
\end{equation}
where $\mathbf{x}\in{{\mathcal{X}}}\subseteq%
\mathbb{R}
^{K}$, $\mathbf{r}\in\mathcal{R}\subseteq%
\mathbb{R}
^{N}$, $d\mathbf{x}=\prod_{k=1}^{K}dx_{k}$, $d\mathbf{r}=\prod_{n=1}^{N}%
dr_{n}$, and the integration symbol $%
{\textstyle\int}
$ is for the continuous variables and can be replaced by summation symbol $%
{\textstyle\sum}
$\ for discrete variables. The probability density function (p.d.f.)
of\textbf{\ }$\mathbf{r}$, $p(\mathbf{r})$, satisfies
\begin{equation}
p(\mathbf{r})=\int_{{{\mathcal{X}}}}p(\mathbf{r}|\mathbf{x})p(\mathbf{x}%
)d\mathbf{x}\text{.} \label{pr}%
\end{equation}
The MI $I$ in (\ref{MI}) may also be expressed equivalently as
\begin{equation}
I=H(X)-\left\langle \ln\frac{p(\mathbf{r})}{p(\mathbf{r}|\mathbf{x}%
)p(\mathbf{x})}\right\rangle _{\mathbf{r,\,x}}=H(X)-H(X|R)\text{,} \label{MI1}%
\end{equation}
where $H(X)$ is the entropy of random variable $X$:
\begin{equation}
H(X)=-\left\langle \ln p(\mathbf{x})\right\rangle _{\mathbf{x}}\text{,
}H(X|R)=-\left\langle \ln p(\mathbf{x}|\mathbf{r})\right\rangle _{\mathbf{r}%
\text{\textbf{,\thinspace}}\mathbf{x}}\text{,} \label{H(x)}%
\end{equation}
and $\left\langle \cdot\right\rangle $ denotes expectation:
\begin{align}
&  {\left\langle \cdot\right\rangle _{\mathbf{x}}}={\int_{{{\mathcal{X}}}%
}p(\mathbf{x})(\cdot)d}\mathbf{x}\text{, }\label{Expect1}\\
&  {\left\langle \cdot\right\rangle _{\mathbf{r}|\mathbf{x}}}={\int%
_{{\mathcal{R}}}p(\mathbf{r}|\mathbf{x})(\cdot)d}\mathbf{r}\text{,
}\label{Expect2}\\
&  {\left\langle \cdot\right\rangle _{\mathbf{r}\text{,\thinspace}\mathbf{x}}%
}={\int_{{{\mathcal{X}}}}\int_{{\mathcal{R}}}p(\mathbf{r}\text{,\thinspace
}\mathbf{x})(\cdot){d}\mathbf{r}d}\mathbf{x}\text{.} \label{Expect3}%
\end{align}
Next, we introduce the following notations,
\begin{align}
&  l\left(  \mathbf{r}|\mathbf{x}\right)  =\ln p\left(  \mathbf{r}%
|\mathbf{x}\right)  \text{, }\label{Notations1}\\
&  L\left(  \mathbf{r}|\mathbf{x}\right)  =\ln\left(  p\left(  \mathbf{r}%
|\mathbf{x}\right)  p\left(  \mathbf{x}\right)  \right)  \text{,}%
\label{Notations2}\\
&  {q}\left(  {\mathbf{x}}\right)  =\ln p\left(  \mathbf{x}\right)  \text{,}
\label{Notations3}%
\end{align}
and
\begin{align}
I_{F}  &  =\frac{1}{2}\left\langle \ln\left(  \det\left(  \frac{\mathbf{J}%
(\mathbf{x})}{2\pi e}\right)  \right)  \right\rangle _{\mathbf{x}%
}+H(X)\text{,}\label{IF}\\
I_{G}  &  =\frac{1}{2}\left\langle \ln\left(  \det\left(  \frac{\mathbf{G}%
(\mathbf{x})}{2\pi e}\right)  \right)  \right\rangle _{\mathbf{x}%
}+H(X)\text{,} \label{IG}%
\end{align}
where $\det\left(  \cdot\right)  $ denotes the matrix determinant, and
\begin{align}
&  \mathbf{J}(\mathbf{x})=\left\langle l^{\prime}(\mathbf{r}|\mathbf{x}%
)l^{\prime}(\mathbf{r}|\mathbf{x})^{T}\right\rangle _{\mathbf{r}|\mathbf{x}%
}\text{,}\label{Jx}\\
&  \mathbf{G}(\mathbf{x})=\mathbf{J}(\mathbf{x})+\mathbf{P}\left(
\mathbf{x}\right)  \text{,}\label{Gx}\\
&  \mathbf{P}(\mathbf{x})=-{{q}^{\prime\prime}(\mathbf{x})}\text{.} \label{Px}%
\end{align}
Here $\mathbf{J}(\mathbf{x})$ is FI matrix, which is symmetric and
positive-semidefinite, and $^{\prime}$ and ${^{\prime\prime}}$ denote the
first and second derivative for $\mathbf{x}$, respectively; that is,
$l^{\prime}(\mathbf{r}|\mathbf{x})=\partial l\left(  \mathbf{r}|\mathbf{x}%
\right)  /\partial\mathbf{x}$ and ${{q}^{\prime\prime}}(\mathbf{r}%
|\mathbf{x})=\partial^{2}\ln p\left(  \mathbf{x}\right)  /\partial
\mathbf{x}\partial\mathbf{x}^{T}$. If $p(\mathbf{r}|\mathbf{x})$ is twice
differentiable for $\mathbf{x}$, then
\begin{equation}
\mathbf{J}(\mathbf{x})=\left\langle l^{\prime}(\mathbf{r}|\mathbf{x}%
)l^{\prime}(\mathbf{r}|\mathbf{x})^{T}\right\rangle _{\mathbf{r}|\mathbf{x}%
}=-\left\langle l^{\prime\prime}(\mathbf{r}|\mathbf{x})\right\rangle
_{\mathbf{r}|\mathbf{x}}\text{.} \label{Jx1}%
\end{equation}
We denote the Kullback-Leibler (KL) divergence as
\begin{equation}
D\left(  \mathbf{x}||\mathbf{\hat{x}}\right)  =\int_{{\mathcal{R}}}p\left(
\mathbf{r}|\mathbf{x}\right)  \ln\frac{p\left(  \mathbf{r}|\mathbf{x}\right)
}{p\left(  \mathbf{r}|\mathbf{\hat{x}}\right)  }d\mathbf{r}\text{,}
\label{DKL}%
\end{equation}
and define
\begin{equation}
{{\mathcal{X}}}_{\omega}(\mathbf{x})=\left\{  \mathbf{\breve{x}}\in%
\mathbb{R}
^{K}:\left(  \mathbf{\breve{x}}-\mathbf{x}\right)  ^{T}\mathbf{G}%
(\mathbf{x})\left(  \mathbf{\breve{x}}-\mathbf{x}\right)  <N\omega
^{2}\right\}  \text{,} \label{X_neigb}%
\end{equation}
as the $\omega$ neighborhoods of $\mathbf{x}$, and its complementary set as
\begin{equation}
{{\mathcal{\bar{X}}}}_{\omega}(\mathbf{x})={{\mathcal{X}}}-{{\mathcal{X}}%
}_{\omega}(\mathbf{x})\text{,} \label{X_neigb-}%
\end{equation}
where $\omega$ is a positive number.

\subsection{Information-Theoretic Asymptotic Bounds and Approximations}

\label{Sec:2.2}In large $N$ limit, \cite{Brunel(1998-mutual)} proposed an
asymptotic relationship $I\sim I_{F}$ between MI and FI and gave a proof in
the case of one-dimensional input. Another proof is given by
\cite{Sompolinsky(2001-population)} although there appears to be an error in
their proof when replica trick is used (see Eq.~(B1) in their paper; their
Eq.~(B5) does not follow directly from the replica trick). For large but
finite $N$, $I\simeq I_{F}$ is usually a good approximation as long as the
inputs are low-dimensional. For the high-dimensional inputs, the approximation
may no longer be valid. For example, suppose $p(\mathbf{r}|\mathbf{x})$ is a
normal distribution with mean $\mathbf{A}^{T}\mathbf{x}$ and covariance matrix
$\mathbf{I}_{N}$ and $p(\mathbf{x})$ is a normal distribution with mean
${\boldsymbol{\mu}}$ and covariance matrix $\boldsymbol{\Sigma}$,
\begin{equation}
{p(\mathbf{r}|\mathbf{x})}=\mathcal{N}{\left(  \mathbf{A}^{T}\mathbf{x}%
\text{,\thinspace}\mathbf{I}_{N}\right)  }\text{, }\quad{p(\mathbf{x}%
)=\mathcal{N}\left(  {\boldsymbol{\mu}}\text{,\thinspace}\boldsymbol{\Sigma
}\right)  }\text{,} \label{Gau}%
\end{equation}
where $\mathbf{A}=\left[  \mathbf{a}_{1}\text{,\thinspace}\mathbf{a}%
_{2}\text{,\thinspace}\cdots\text{,\thinspace}\mathbf{a}_{N}\right]  $ is a
deterministic $K\times N$ matrix and $\mathbf{I}_{N}$\ is the $N\times N$
identity matrix. The MI $I$ is given by
\citep[see][for details]{Verdu(1986-IP-capacity),Guo(2005-mutual)}
\begin{equation}
I=\frac{1}{2}\ln\left(  \det\left(  \boldsymbol{\Sigma}^{1/2}\mathbf{AA}%
^{T}\boldsymbol{\Sigma}^{1/2}+\mathbf{I}_{K}\right)  \right)  \text{.}
\label{I_Gau}%
\end{equation}
If $\mathrm{rank}\left(  \mathbf{J}(\mathbf{x})\right)  <K$, then
$I_{F}=-\infty$. Notice that here $\mathbf{J}(\mathbf{x})=\mathbf{AA}^{T}$.
When $\mathbf{a}=\mathbf{a}_{1}=\cdots=\mathbf{a}_{N}$ and $\boldsymbol{\Sigma
}=\mathbf{I}_{K}$, then by (\ref{I_Gau}) and matrix determinant lemma, we
have
\begin{equation}
I=\frac{1}{2}\ln\left(  \det\left(  N\mathbf{aa}^{T}+\mathbf{I}_{K}\right)
\right)  =\frac{1}{2}\ln\left(  N\mathbf{a}^{T}\mathbf{a}+1\right)
\geq0\text{,} \label{I_Gau1}%
\end{equation}
and by (\ref{IF}),
\begin{equation}
I_{F}=\frac{1}{2}\ln\left(  \det\left(  N\mathbf{aa}^{T}\right)  \right)
=-\infty\text{,} \label{I_Gau2}%
\end{equation}
which is obviously incorrect as an approximation to $I$. For high-dimensional
inputs, the determinant $\det\left(  \mathbf{J}(\mathbf{x})\right)  $ may
become close to zero in practical applications. When the FI matrix
$\mathbf{J}(\mathbf{x})$ becomes degenerate, the regularity condition ensuring
the Cram\'{e}r-Rao paradigm of statistics is violated
\citep{Amari(2005-difficulty)}, in which case using $I_{F}$ as a proxy for $I$
incurs large errors.

In the following, we will show $I_{G}$ is a better approximation of $I$ for
high-dimensional inputs. For instance, for the above example, we can verify
that%
\begin{align}
I_{G}  &  =\frac{1}{2}\ln\left(  \det\left(  \frac{1}{2\pi e}\left(
\mathbf{AA}^{T}+\boldsymbol{\Sigma}^{-1}\right)  \right)  \right)  +\frac
{1}{2}\ln\left(  \det\left(  2\pi e\boldsymbol{\Sigma}\right)  \right)
\nonumber\\
&  =\frac{1}{2}\ln\left(  \det\left(  \boldsymbol{\Sigma}^{1/2}\mathbf{AA}%
^{T}\boldsymbol{\Sigma}^{1/2}+\mathbf{I}_{K}\right)  \right)  =I\text{,}
\label{I_Gau3}%
\end{align}
which is exactly equal to the MI $I$ given in (\ref{I_Gau}).

\subsubsection{Regularity Conditions}

\label{Conditions}First, we consider the following regularity conditions for
$p(\mathbf{x})$ and $p(\mathbf{r}|\mathbf{x})$:

\textbf{C1:} $p(\mathbf{x})$ and $p(\mathbf{r}|\mathbf{x})$ are twice
continuously differentiable for almost every $\mathbf{x}\in{{\mathcal{X}}}$,
where ${{\mathcal{X}}}$ is a convex set; $\mathbf{G}(\mathbf{x})$ is positive
definite and ${\left\Vert \mathbf{G}^{-1}\left(  \mathbf{x}\right)
\right\Vert }=O\left(  N^{-1}\right)  $, where ${\left\Vert {\mathbf{\cdot}%
}\right\Vert }$ denotes the Frobenius norm of a matrix; the following
conditions hold
\begin{subequations}
\begin{align}
&  {\left\Vert {q}{^{\prime}(\mathbf{x})}\right\Vert }<\infty\text{,}%
\label{C1.a1}\\
&  {\left\Vert { {q}^{\prime\prime}(\mathbf{x})}\right\Vert }<\infty
\text{,}\label{C1.a2}\\
&  {\left\langle \left(  N^{-1}l^{\prime}(\mathbf{r}|\mathbf{x})^{T}l^{\prime
}(\mathbf{r}|\mathbf{x})\right)  ^{2}\right\rangle _{\mathbf{r}|\mathbf{x}}%
}=O\left(  1\right)  \text{,}\label{C1.b1}\\
&  {\left\langle \left\Vert N^{-1}\left(  l^{\prime\prime}(\mathbf{r}%
|\mathbf{x})-\left\langle l^{\prime\prime}(\mathbf{r}|\mathbf{x})\right\rangle
_{\mathbf{r}|\mathbf{x}}\right)  \right\Vert ^{2}\right\rangle _{\mathbf{r}%
|\mathbf{x}}}=O\left(  N^{-1}\right)  \text{,} \label{C1.b2}%
\end{align}
and there exists an $\omega=\omega\left(  \mathbf{x}\right)  >0$\ for
$\forall\mathbf{\breve{x}}\in{{\mathcal{X}}}_{\omega}(\mathbf{x})$ such that%
\begin{equation}
N^{-1}{\left\Vert l^{\prime\prime}(\mathbf{r}|\mathbf{\breve{x}}%
)-l^{\prime\prime}(\mathbf{r}|\mathbf{x})\right\Vert }=O\left(  1\right)
\text{,} \label{C1.c}%
\end{equation}
where $O$\ indicates the big-O notation.

\textbf{C2:} The following condition is satisfied:%
\end{subequations}
\begin{subequations}
\begin{equation}
{\left\langle \left\Vert N^{-1}\left(  l^{\prime\prime}(\mathbf{r}%
|\mathbf{x})-\left\langle l^{\prime\prime}(\mathbf{r}|\mathbf{x})\right\rangle
_{\mathbf{r}|\mathbf{x}}\right)  \right\Vert ^{2\left(  m+1\right)
}\right\rangle _{\mathbf{r}|\mathbf{x}}}=O\left(  N^{-1}\right)  \text{,}
\label{C2.a}%
\end{equation}
for $m\in%
\mathbb{N}
$, and there exists $\eta>1$ such that
\begin{equation}
\mathbb{P}_{\mathbf{r}|\mathbf{x}}\left\{  {\det}\left(  \mathbf{G}\left(
\mathbf{x}\right)  \right)  ^{1/2}\int_{{{\mathcal{\bar{X}}}}_{\hat{\omega}%
}(\mathbf{x})}p(\mathbf{\hat{x}}|\mathbf{r})d\mathbf{\hat{x}}>\epsilon
p(\mathbf{x}|\mathbf{r})\right\}  =O\left(  N^{-\eta}\right)  \label{C2.b}%
\end{equation}
for all $\epsilon\in\left(  0\text{,\thinspace}1/2\right)  $, $\hat{\omega}%
\in\left(  0\text{,\thinspace}\omega\right)  $ and $\mathbf{x}\in
{{\mathcal{X}}}$ with $p(\mathbf{x})>0$, where $\mathbb{P}_{\mathbf{r}%
|\mathbf{x}}\left\{  \cdot\right\}  $ denotes the probability of $\mathbf{r}%
$\ given $\mathbf{x}$.

\bigskip

The regularity conditions \textbf{C1} and \textbf{C2} are needed to prove
theorems in later sections. They are expressed in mathematical forms that are
convenient for our proofs although their meanings may seem opaque at the first
glance. In the following, we will examine these conditions more closely. We
will use specific examples to make interpretations of these conditions more
transparent.

\end{subequations}

\begin{remark}
\label{Remark 1d}

In this paper we assume that the probability distributions $p(\mathbf{x})$
and $p(\mathbf{r|x})$ are piecewise twice continuously differentiable. This is
because we need to use Fisher information to approximate mutual information,
and Fisher information requires derivatives that make sense only for
continuous variables. Therefore, the methods developed in this paper apply
only to continuous input variables or stimulus variables. For discrete input
variables, we need alternative methods for approximating MI and we will
address this issue in a separate publication.

Conditions (\ref{C1.a1}) and (\ref{C1.a2}) state that the first and the second
derivatives of $q(\mathbf{x})=\ln p(\mathbf{x})$ have finite values for any
given $\mathbf{x}\in{{\mathcal{X}}}$. These two conditions are easily
satisfied by commonly encountered probability distributions because they only
require finite derivatives within ${{\mathcal{X}}}$, the set of allowable
inputs, and derivatives do not need to be finitely bounded.

\end{remark}

\begin{remark}
\label{Remark 1b}

Conditions (\ref{C1.b1})--(\ref{C2.a}) constrain how the first and the second
derivatives of $l(\mathbf{r}|\mathbf{x})=\ln p(\mathbf{r}|\mathbf{x})$ scale
with $N$, the number of neurons. These conditions are easily met when
$p(\mathbf{r}|\mathbf{x})$ is conditionally independent or when the noises of
different neurons are independent, i.e., $p(\mathbf{r}|\mathbf{x})=\prod
_{n=1}^{N} p(r_{n}|\mathbf{x})$.

We emphasize that it is possible to satisfy these conditions even when
$p(\mathbf{r}|\mathbf{x})$ is not independent or when the noises are
correlated, as shown later. Here we first examine these conditions closely 
assuming independence. For simplicity, our demonstration below is based on a one-dimensional
input variable ($K=1$). The conclusions  are readily generalizable to
higher dimensional inputs ($K>1$) because $K$ is fixed and does not affect the
scaling with $N$.

Assuming independence, we have $l(\mathbf{r}|x)=\sum_{n=1}^{N}l(r_{n}|{x})$
with $l(r_{n}|{x})=\ln p(r_{n}|{x})$, and the left-hand side of (\ref{C1.b1})
becomes
\begin{align}
&  N^{-2}{\left\langle l^{\prime}(\mathbf{r}|x)^{4}\right\rangle
_{\mathbf{r}|x}}\nonumber\\
&  =N^{-2}{\sum_{n_{1},\cdots,n_{4}=1}^{N}\left\langle l^{\prime}(r_{n_{1}%
}|x)l^{\prime}(r_{n_{2}}|x)l^{\prime}(r_{n_{3}}|x)l^{\prime}(r_{n_{4}%
}|x)\right\rangle _{r_{n_{1}},r_{n_{2}},r_{n_{3}},r_{n_{4}}|x}}\nonumber\\
&  =N^{-2}\left(  \sum_{n\neq m}{\left\langle l^{\prime}(r_{n}|x)^{2}%
\right\rangle _{r_{n}|x}\left\langle l^{\prime}(r_{m}|x)^{2}\right\rangle
_{r_{m}|x}}+\sum_{n=1}^{N}{\left\langle l^{\prime}(r_{n}|x)^{4}\right\rangle
_{r_{n}|x}}\right)  , \label{RemS1.1}%
\end{align}
where the final result contains only two terms with even numbers of duplicated
indices while all other terms in the expansion vanish because any unmatched or
lone index $k$ (from $n_{1},n_{2},n_{3},n_{4}$) should yield a vanishing
average:
\begin{equation}
\left\langle l^{\prime}(r_{k}|x)\right\rangle _{r_{k}|x}={\int_{{\mathcal{R}}%
}}p(r_{k}|x)l^{\prime}(r_{k}|x)dr_{k}=\frac{\partial}{\partial x}\left(
{\int_{{\mathcal{R}}}}p(r_{k}|x)dr_{k}\right)  =0. \label{RemS1.1a}%
\end{equation}
Thus, condition (\ref{C1.b1}) is satisfied as long as ${\left\langle
l^{\prime}(r_{n}|x)^{2}\right\rangle _{r_{n}|x}}$ and ${\left\langle
l^{\prime}(r_{n}|x)^{4}\right\rangle _{r_{n}|x}}$ are bounded by some finite
numbers, say, $a$ and $b$, respectively, 
because now (\ref{RemS1.1}) should scale as $N^{-2}\left(
aN(N-1)+bN\right)  =O(1)$. For instance, a Gaussian distribution always meets
this requirement because the averages of the second and fourth powers are
proportional to the second and fourth moments, which are both finite.
Note that the argument above works even if 
${\left\langle
l^{\prime}(r_{n}|x)^{4}\right\rangle _{r_{n}|x}}$
is not finitely bounded but scales as $O(N)$.

Similarly, under the assumption of independence, the left-hand side of
(\ref{C1.b2}) becomes
\begin{align}
&  N^{-2}{\left\langle \left(  l^{\prime\prime}(\mathbf{r}|x)-\left\langle
l^{\prime\prime}(\mathbf{r}|x)\right\rangle _{\mathbf{r}|x}\right)
^{2}\right\rangle _{\mathbf{r}|x}}\nonumber\\
&  =N^{-2}\sum_{n,m=1}^{N}{\left\langle \left(  l^{\prime\prime}%
(r_{n}|x)-\left\langle l^{\prime\prime}(r_{n}|x)\right\rangle _{r_{n}%
|x}\right)  \left(  l^{\prime\prime}(r_{m}|x)-\left\langle l^{\prime\prime
}(r_{m}|x)\right\rangle _{r_{m}|x}\right)  \right\rangle _{r_{n},r_{m}|x}%
}\nonumber\\
&  =N^{-2}\sum_{n=1}^{N}{\left\langle \left(  l^{\prime\prime}(r_{n}%
|x)-\left\langle l^{\prime\prime}(r_{n}|x)\right\rangle _{r_{n}|x}\right)
^{2}\right\rangle _{r_{n}|x}}\nonumber\\
&  =N^{-2}\sum_{n=1}^{N} \left(  {\left\langle l^{\prime\prime}(r_{n}%
|x)^{2}\right\rangle _{r_{n}|x}} -\left\langle l^{\prime\prime}(r_{n}%
|x)\right\rangle _{r_{n}|x}^{2} \right)  \text{,} \label{RemS1.2}%
\end{align}
where in the second step, the only remaining terms are the squares while all
other terms in the expansion with $n\neq m$ have vanished because
${\left\langle l^{\prime\prime}(r_{n}|x)-\left\langle l^{\prime\prime}%
(r_{n}|x)\right\rangle _{r_{n}|x}\right\rangle _{r_{n}|x}=0}$. Thus, condition
(\ref{C1.b2}) is satisfied as long as $\left\langle l^{\prime\prime}%
(r_{n}|x)\right\rangle _{r_{n}|x}$ and $\left\langle l^{\prime\prime}%
(r_{n}|x)^{2}\right\rangle _{r_{n}|x}$ are bounded so that (\ref{RemS1.2})
scales as $N^{-2}N=N^{-1}$.

Condition (\ref{C1.c}) is easily satisfied under the assumption of
independence. It is easy to show that this condition holds when $l^{\prime
\prime}(r_{n}|x)$ is bounded.

Condition (\ref{C2.a}) can be examined using similar arguments used for
(\ref{RemS1.1}) and (\ref{RemS1.2}). Assuming independence, we rewrite the
left-hand side of (\ref{C2.a}) as:
\begin{align}
&  N^{-z}{\left\langle \left(  l^{\prime\prime}(\mathbf{r}|x)-\left\langle
l^{\prime\prime}(\mathbf{r}|x)\right\rangle _{\mathbf{r}|x}\right)
^{z}\right\rangle _{\mathbf{r}|x}}\nonumber\\
&  =N^{-z}\sum_{n_{1},\cdots,n_{z}=1}^{N} \left\langle \left(  l^{\prime
\prime}({r}_{n_{1}}|x)-\left\langle l^{\prime\prime}({r}_{n_{1}}|x)
\right\rangle _{r_{n_{1}}|x} \right)  \cdots\left(  l^{\prime\prime}({r}%
_{1}|x)-\left\langle l^{\prime\prime}({r}_{n_{z}}|x) \right\rangle _{r_{n_{z}%
}|x} \right)  \right\rangle _{r_{n_{z}}|x}\nonumber\\
&  =N^{-z}\sum_{n_{1},\cdots,n_{m+1}=1}^{N} \left\langle \prod_{i=1}^{m+1}
\left(  l^{\prime\prime}({r}_{n_{i}}|x)-\left\langle l^{\prime\prime}%
({r}_{n_{i}}|x) \right\rangle _{r_{n_{i}}|x} \right)  ^{2} \right\rangle
_{r_{n_{i}}|x} +\cdots\label{RemS1.C2.a}%
\end{align}
where $z=2(m+1)\geq4$ is an even number. Any term in the expansion with an
unmatched index $n_{k}$ should vanish, as in the cases of (\ref{RemS1.1}) and
(\ref{RemS1.2}). When $\left\langle l^{\prime\prime}(r_{n}|x)\right\rangle
_{r_{n}|x}$ and $\left\langle l^{\prime\prime}(r_{n}|x)^{2}\right\rangle
_{r_{n}|x}$ are bounded, the leading term with respect to scaling with $N$ is
the product of squares as shown at the end of (\ref{RemS1.C2.a}) because all
the other non-vanishing terms increase more slowly with $N$. Thus
(\ref{RemS1.C2.a}) should scale as $N^{-z}N^{m+1}=N^{-m-1}$, which trivially
satisfies condition (\ref{C2.a}). \qed

In summary, conditions (\ref{C1.b1})--(\ref{C2.a}) are easy to meet when
$p(\mathbf{r}|\mathbf{x})$ is independent. It is sufficient to satisfy these
conditions when the averages of the first and second derivatives of
$l(\mathbf{r}|\mathbf{x})=\ln p(\mathbf{r}|\mathbf{x})$ as well as the
averages of their powers are bounded by finite numbers for all the neurons.

\end{remark}

\begin{remark}
\label{Remark 1c}

For neurons with correlated noises, if there exists an invertible
transformation that maps $\mathbf{r}$ to $\mathbf{\tilde{r}}$ such that
$p(\mathbf{\tilde{r}}|\mathbf{x})$ becomes conditionally independent, then
conditions \textbf{C1 }and \textbf{C2} are easily met in the space of the new
variables by the discussion in \textbf{Remark \ref{Remark 1b}}. This situation is
best illustrated by the familiar example of a population of neurons with
correlated noises that obey a multivariate Gaussian distribution:
\begin{equation}
{p(\mathbf{r}|x)=\dfrac{1}{\sqrt{\det\left(  2\pi\boldsymbol{\Sigma}\right)
}}\exp\left(  -\dfrac{1}{2}\left(  {\mathbf{r}}-\mathbf{g}\right)
^{T}\boldsymbol{\Sigma}^{-1}\left(  {\mathbf{r}}-\mathbf{g}\right)  \right)
}\text{,} \label{pGau}%
\end{equation}
where $\boldsymbol{\Sigma}$ is an $N\times N$ invertible covariance matrix and
$\mathbf{g} =\left(  g_{1}( x{\text{;\thinspace}}\boldsymbol{\theta}_{1}),
\cdots, g_{N}( x{\text{;\thinspace}}\boldsymbol{\theta}_{N}) \right)  $
describes the mean responses with $\boldsymbol{\theta}_{n}$ being the
parameter vector. Using the following transformation,
\begin{align}
\mathbf{\tilde{r}}  &  =\boldsymbol{\Sigma}^{-1/2}{\mathbf{r}}=\left(
\tilde{r}_{1}\text{,\thinspace}\tilde{r}_{2}\text{,\thinspace}\cdots
\text{,\thinspace}\tilde{r}_{N}\right)  ^{T}\text{,}\label{rg}\\
\mathbf{\tilde{g}}  &  =\boldsymbol{\Sigma}^{-1/2}\mathbf{g}=\left(  \tilde
{g}_{1}\text{,\thinspace}\tilde{g}_{2}\text{,\thinspace}\cdots
\text{,\thinspace}\tilde{g}_{N}\right)  ^{T}\text{,} \label{rg.b}%
\end{align}
we obtain the independent distribution:
\begin{equation}
{p(\mathbf{\tilde{r}}|x)={\prod_{n=1}^{N}}\dfrac{1}{\sqrt{2\pi}}\exp\left(
-\dfrac{1}{2}\left(  \tilde{r}_{n}-\tilde{g}_{n}\right)  ^{2}\right)
}\text{.} \label{pr-x}%
\end{equation}
In the special case when the correlation coefficient between any pair of
neurons is a constant $c$, $-1<c<1$, the noise covariance can be written as
\begin{equation}
\boldsymbol{\Sigma}=a\left(  (1-c)\mathbf{I}_{N}+c\mathbf{uu}^{T}\right)
\text{,} \label{Sigr}%
\end{equation}
where $a>0$ is a constant, $\mathbf{I}_{N}$ is the $N\times N$ identity
matrix, $\mathbf{u}=(1,1,\cdots,1)^{T}\in\mathbb{R}^{N\times1}$. The desired
transformation in (\ref{rg}) and (\ref{rg.b}) is given explicitly by
\begin{equation}
\boldsymbol{\Sigma}^{-1/2}{=b_{0}\left(  \mathbf{I}_{N}-{b_{1}}\mathbf{uu}%
^{T}\right)  }\text{,} \label{Sigr_2}%
\end{equation}
where%
\begin{equation}
b_{0}=\frac{1}{\sqrt{a(1-c)}}\text{,} \quad{b_{1}}=\frac{1}{N}\left(
1\pm\sqrt{\frac{1-c}{(N-1)c+1}}\right)  \text{.} \label{b0b1}%
\end{equation}
The new response variables defined in (\ref{rg}) and (\ref{rg.b}) now read:
\begin{align}
\tilde{r}_{n}  &  ={b_{0}}\left(  r_{n}-{b_{1}}\sum_{m=1}^{N}r_{m}\right)
\text{,}\label{rn1}\\
\tilde{g}_{n}  &  ={b_{0}}\left(  g_{n}- {b_{1}}\sum_{m=1}^{N}g_{m}\right)
\text{.} \label{gn1}%
\end{align}
Now we have the derivatives:
\begin{align}
&  l^{\prime}(\tilde{r}_{n}|x) =\left(  \tilde{r}_{n}-\tilde{g}_{n}\right)
\dfrac{\partial\tilde{g}_{n}}{\partial x},\label{l'rx}\\
&  l^{\prime\prime}(\tilde{r}_{n}|x)-\left\langle l^{\prime\prime}(\tilde
r_{n}|x)\right\rangle _{r_{n}|x}=\left(  \tilde{r}_{n}-\tilde{g}_{n}\right)
\dfrac{\partial^{2}\tilde{g}_{n}}{\partial x^{2}}\text{,} \label{l"rx}%
\end{align}
where $\partial\tilde{g}_{n}/\partial x$ and $\partial^{2}\tilde{g}%
_{n}/\partial x^{2}$ are finite as long as $\partial{g}_{n}/\partial x$ and
$\partial^{2} {g}_{n}/\partial x^{2}$ are finite. Conditions \textbf{C1
}and\textbf{ C2} are satisfied when the derivatives and their powers are
finitely bounded as shown before.

The example above shows explicitly that it is possible to meet conditions
\textbf{C1 }and \textbf{C2} even when the noises of different neurons are
correlated. More generally, if a nonlinear transformation exists that maps
correlated random variables into independent variables, then by similar
argument, conditions \textbf{C1 }and \textbf{C2} are satisfied when the
derivatives of the log likelihood functions and their powers in the new
variables are finitely bounded. Even when the desired transformation does not
exist or is unknown, it does not necessarily imply that conditions \textbf{C1
}and \textbf{C2} must be violated.

While the exact mathematical conditions for the existence of the desired transformation
are unclear, let us consider a specific example.
If a joint probability density function can be morphed smoothly and reversibly into a flat 
or constant density in a cube (hypercube), which is a special case of an independent distribution,  
then this morphing is the desired transformation.
Here we may replace the flat distribution by any known independent distribution
and the argument above should still work. 
So the desired transformation may exist under rather general conditions.

For correlated random variables, one may use algorithms such as independent component
analysis to find an invertible linear mapping that makes the new random
variables as independent as possible \citep{Bell(1997-independent)}, or use
neural networks to find related nonlinear mappings
\citep[]{Huang(2017-IC-information)}. These methods do not 
directly apply to the problem of testing conditions \textbf{C1 }and \textbf{C2} because
they work for a given network size $N$ and further development is needed to address the scaling behavior in the large
network limit $N\rightarrow\infty$.

Finally, we note that the value of the MI of the transformed independent
variables is the same as the MI of the original correlated variables because
of the invariance of MI under invertible transformation of marginal variables.
A related discussion is in Theorem 4.1 which involves a transformation of the
input variables rather than a transformation of the output variables as needed
here.

\end{remark}

\begin{remark}
\label{Remark 1a} Condition (\ref{C2.b}) is satisfied if a positive number
$\delta$ and a positive integer $m$ exist such that%
\begin{equation}
{\det}\left(  \mathbf{G}\left(  \mathbf{x}\right)  \right)  ^{1/2}%
\int_{{{\mathcal{\bar{X}}}}_{\hat{\omega}}(\mathbf{x})}\int_{{{\mathcal{B}}%
}_{m\text{,\thinspace}\delta}\left(  \mathbf{x}\right)  }p(\mathbf{r|\hat{x}%
})p(\mathbf{\hat{x}})d\mathbf{r}d\mathbf{\hat{x}}=O\left(  N^{-\eta}\right)
\text{,} \label{RemS1.C.1}%
\end{equation}
for all $\mathbf{\hat{x}}\in{{\mathcal{\bar{X}}}}_{\hat{\omega}}(\mathbf{x})$,
where%
\begin{equation}
{{\mathcal{B}}}_{m\text{,\thinspace}\delta}\left(  \mathbf{x}\right)
=\left\{  \mathbf{r}\in{\mathcal{R}}:-\delta N^{\frac{\eta-1}{2m}}%
\mathbf{G}(\mathbf{x})<l^{\prime\prime}(\mathbf{r}|\mathbf{x})-\left\langle
l^{\prime\prime}(\mathbf{r}|\mathbf{x})\right\rangle _{\mathbf{r}|\mathbf{x}%
}<\delta N^{\frac{\eta-1}{2m}}\mathbf{G}(\mathbf{x})\right\}
\label{RemS1.C.1a}%
\end{equation}
and $\mathbf{A}<\mathbf{B}$ means that the matrix $\mathbf{A}-\mathbf{B}$\ is
negative definite. A proof is as follows.

First note that in (\ref{RemS1.C.1a}) if $\eta\rightarrow1$\ or $m\rightarrow
\infty$, then $N^{\frac{\eta-1}{2m}}\rightarrow1$. Following Markov's
inequality, condition \textbf{C2} and (\ref{A.vv1<}) in the Appendix, for the
complementary set of ${{\mathcal{B}}}_{m\text{,\thinspace}\delta}\left(
\mathbf{x}\right)  $, ${{\mathcal{\bar{B}}}}_{m\text{,\thinspace}\delta
}\left(  \mathbf{x}\right)  $, we have
\begin{align}
\mathbb{P}_{\mathbf{r}|\mathbf{x}}\left\{  {{\mathcal{\bar{B}}}}%
_{m\text{,\thinspace}\delta}\left(  \mathbf{x}\right)  \right\}   &
\leq\mathbb{P}_{\mathbf{r}|\mathbf{x}}\left\{  \left\Vert \mathbf{B}%
_{0}\right\Vert ^{2}\geq\delta^{2}N^{\frac{\eta-1}{m}}\right\} \nonumber\\
&  \leq\delta^{-2m}N^{-\left(  \eta-1\right)  }\left\langle \left\Vert
\mathbf{B}_{0}\right\Vert ^{2m}\right\rangle _{\mathbf{r}|\mathbf{x}%
}\nonumber\\
&  =O\left(  N^{-\eta}\right)  \text{,} \label{RemS1.C.2}%
\end{align}
where%
\begin{equation}
\mathbf{B}_{0}=\mathbf{G}^{-1/2}(\mathbf{x})\left(  l^{\prime\prime
}(\mathbf{r}|\mathbf{x})-\left\langle l^{\prime\prime}(\mathbf{r}%
|\mathbf{x})\right\rangle _{\mathbf{r}|\mathbf{x}}\right)  \mathbf{G}%
^{-1/2}(\mathbf{x})\text{.} \label{RemS1.C.2a}%
\end{equation}
Define the set,%
\begin{equation}
{{\mathcal{A}}}_{\hat{\omega}}\left(  \mathbf{x}\right)  =\left\{
\mathbf{r}\in{\mathcal{R}}:\int_{{{\mathcal{\bar{X}}}}_{\hat{\omega}%
}(\mathbf{x})}\frac{p(\mathbf{\hat{x}}|\mathbf{r})}{p(\mathbf{x}|\mathbf{r}%
)}d\mathbf{\hat{x}}>{\det}\left(  \mathbf{G}\left(  \mathbf{x}\right)
\right)  ^{-1/2}\epsilon\right\}  \text{,} \label{RemS1.C.3}%
\end{equation}
then it follows from the Markov's inequality and (\ref{RemS1.C.1}) that%
\begin{align}
&  \mathbb{P}_{\mathbf{r}|\mathbf{x}}\left\{  {{\mathcal{A}}}_{\hat{\omega}%
}\left(  \mathbf{x}\right)  \cap{{\mathcal{B}}}_{m\text{,\thinspace}\delta
}\left(  \mathbf{x}\right)  \right\} \nonumber\\
&  \leq\epsilon^{-1}{\det}\left(  \mathbf{G}\left(  \mathbf{x}\right)
\right)  ^{1/2}\int_{{{\mathcal{B}}}_{m\text{,\thinspace}\delta}\left(
\mathbf{x}\right)  }\int_{{{\mathcal{\bar{X}}}}_{\hat{\omega}}(\mathbf{x}%
)}\frac{p(\mathbf{r|\hat{x}})p(\mathbf{\hat{x}})}{p(\mathbf{x})}%
d\mathbf{\hat{x}}d\mathbf{r}\nonumber\\
&  =O\left(  N^{-\eta}\right)  \text{.} \label{RemS1.C.4}%
\end{align}
Hence, we get%
\[
\mathbb{P}_{\mathbf{r}|\mathbf{x}}\left\{  {{\mathcal{A}}}_{\hat{\omega}%
}\left(  \mathbf{x}\right)  \right\}  \leq\mathbb{P}_{\mathbf{r}|\mathbf{x}%
}\left\{  {{\mathcal{A}}}_{\hat{\omega}}\left(  \mathbf{x}\right)
\cap{{\mathcal{B}}}_{m\text{,\thinspace}\delta}\left(  \mathbf{x}\right)
\right\}  +\mathbb{P}_{\mathbf{r}|\mathbf{x}}\left\{  {{\mathcal{\bar{B}}}%
}_{m\text{,\thinspace}\delta}\left(  \mathbf{x}\right)  \right\}  =O\left(
N^{-\eta}\right)  \text{,}%
\]
which yields the condition (\ref{C2.b}).

Condition (\ref{RemS1.C.1}) is satisfied if there exists a positive number
$\varsigma$\ such that
\begin{equation}
\ln\frac{p(\mathbf{r|x})}{p(\mathbf{r|\hat{x}})}\geq N\varsigma
\label{RemS2.D.1}%
\end{equation}
for all $\mathbf{\hat{x}}\in{{\mathcal{\bar{X}}}}_{\hat{\omega}}(\mathbf{x})$
and $\mathbf{r}\in{{\mathcal{B}}}_{m\text{,\thinspace}\delta}\left(
\mathbf{x}\right)  $. This is because%
\begin{align}
&  {\det}\left(  \mathbf{G}\left(  \mathbf{x}\right)  \right)  ^{1/2}%
\int_{{{\mathcal{\bar{X}}}}_{\hat{\omega}}(\mathbf{x})}\int_{{{\mathcal{B}}%
}_{m\text{,\thinspace}\delta}\left(  \mathbf{x}\right)  }p(\mathbf{r|\hat{x}%
})p(\mathbf{\hat{x}})d\mathbf{r}d\mathbf{\hat{x}}\nonumber\\
&  ={\det}\left(  \mathbf{G}\left(  \mathbf{x}\right)  \right)  ^{1/2}%
\int_{{{\mathcal{\bar{X}}}}_{\hat{\omega}}(\mathbf{x})}p(\mathbf{\hat{x}}%
)\int_{{{\mathcal{B}}}_{m\text{,\thinspace}\delta}\left(  \mathbf{x}\right)
}p(\mathbf{r|x})\exp\left(  -\ln\tfrac{p(\mathbf{r|x})}{p(\mathbf{r|\hat{x}}%
)}\right)  d\mathbf{r}d\mathbf{\hat{x}}\nonumber\\
&  \leq{\det}\left(  \mathbf{G}\left(  \mathbf{x}\right)  \right)  ^{1/2}%
\exp\left(  -N\varsigma\right)  =O\left(  N^{K/2}e^{-N\varsigma}\right)
\text{.} \label{RemS2.D.2}%
\end{align}
Here notice that ${\det}\left(  \mathbf{G}\left(  \mathbf{x}\right)  \right)
^{1/2}=O\left(  N^{K/2}\right)  $ (see Eq. \ref{A.detG}).

Inequality (\ref{RemS2.D.1}) holds if $p(\mathbf{r}|\mathbf{x})$ is
conditionally independent, namely, $p(\mathbf{r}|\mathbf{x})={\prod_{n=1}%
^{N}p(}r{_{n}|}\mathbf{x}{)}$, with
\begin{equation}
\ln\frac{p(r{_{n}}\mathbf{|x})}{p(r{_{n}}\mathbf{|\hat{x}})}\geq
\varsigma\text{,\ }\forall n=1\text{,\thinspace}2\text{,\thinspace}%
\cdots\text{,\thinspace}N\text{,} \label{RemS1.D.3}%
\end{equation}
for all $\mathbf{\hat{x}}\in{{\mathcal{\bar{X}}}}_{\hat{\omega}}(\mathbf{x})$
and $\mathbf{r}\in{{\mathcal{B}}}_{m\text{,\thinspace}\delta}\left(
\mathbf{x}\right)  $. Consider the inequality $\left\langle \ln p(r{_{n}%
}\mathbf{|x})/p(r{_{n}}\mathbf{|\hat{x}})\right\rangle _{r{_{n}}\mathbf{|x}%
}\geq0$ where the equality holds when $\mathbf{x}=\mathbf{\hat{x}}$. If there
is only one extreme point at $\mathbf{\hat{x}}=\mathbf{x}$\ for $\mathbf{\hat
{x}}\in{{\mathcal{X}}}_{\omega}\left(  {\mathbf{x}}\right)  $, then generally
it is easy to find a set ${{\mathcal{B}}}_{m\text{,\thinspace}\delta}\left(
\mathbf{x}\right)  $\ that satisfies (\ref{RemS1.D.3}), so that (\ref{C2.b})
holds.\qed

\end{remark}

\subsubsection{Asymptotic Bounds and Approximations for Mutual Information}

Let
\begin{equation}
\xi={N^{-1}}\left\langle \left\Vert \left(  {{l^{\prime\prime}(\mathbf{r}%
|\mathbf{x})}-\left\langle {l^{\prime\prime}(\mathbf{r}|\mathbf{x}%
)}\right\rangle _{{\mathbf{r}|\mathbf{x}}}}\right)  \mathbf{G}{^{-1}\left(
\mathbf{x}\right)  l^{\prime}(\mathbf{r}|\mathbf{x})}\right\Vert
^{2}\right\rangle _{{_{\mathbf{r}|\mathbf{x}}}}\text{,} \label{XiDef}%
\end{equation}
and it follows from conditions \textbf{C1} and \textbf{C2} that%
\begin{align}
\xi &  \leq\left\Vert N\mathbf{G}{^{-1}\left(  \mathbf{x}\right)  }\right\Vert
^{2}\left\langle \left\Vert {N^{-1}}\left(  {{l^{\prime\prime}(\mathbf{r}%
|\mathbf{x})}-\left\langle {l^{\prime\prime}(\mathbf{r}|\mathbf{x}%
)}\right\rangle _{{\mathbf{r}|\mathbf{x}}}}\right)  \right\Vert ^{4}%
\right\rangle _{{_{\mathbf{r}|\mathbf{x}}}}^{1/2}\nonumber\\
&  \times\left\langle \left(  N^{-1}l^{\prime}(\mathbf{r}|\mathbf{x}%
)^{T}l^{\prime}(\mathbf{r}|\mathbf{x})\right)  ^{2}\right\rangle
_{{_{\mathbf{r}|\mathbf{x}}}}^{1/2}\nonumber\\
&  =O\left(  N^{-1/2}\right)  \text{.} \label{Xi<}%
\end{align}
Moreover, if $p(\mathbf{r}|\mathbf{x})$ is conditionally independent, then by
an argument similar to the discussion in \textbf{Remark} \ref{Remark 1b}, we
can verify that the condition $\xi=O\left(  N^{-1}\right)  $ is easily met.

In the following we state several conclusions about the MI, and their proofs
are given in Appendix.

\begin{lemma}
\label{Lemma 1}If condition \textbf{C1 }holds, then the MI $I$ has an
asymptotic upper bound for integer $N$,
\begin{equation}
I\leq I_{G}+O\left(  N^{-1}\right)  \text{.} \label{Lma1}%
\end{equation}
Moreover, if Eqs. (\ref{C1.b1}) and (\ref{C1.b2}) are replaced by%
\begin{subequations}
\begin{align}
&  {\left\langle \left\vert N^{-1}l^{\prime}(\mathbf{r}|\mathbf{x}%
)^{T}l^{\prime}(\mathbf{r}|\mathbf{x})\right\vert ^{1+\tau}\right\rangle
_{\mathbf{r}|\mathbf{x}}}=O\left(  1\right)  \text{,}\label{Lma1.1a}\\
&  {\left\langle \left\Vert N^{-1}\left(  l^{\prime\prime}(\mathbf{r}%
|\mathbf{x})-\left\langle l^{\prime\prime}(\mathbf{r}|\mathbf{x})\right\rangle
_{\mathbf{r}|\mathbf{x}}\right)  \right\Vert ^{2}\right\rangle _{\mathbf{r}%
|\mathbf{x}}}=o\left(  1\right)  \text{,} \label{Lma1.1b}%
\end{align}
for some $\tau\in\left(  0\text{,\thinspace}1\right)  $, where $o$\ indicates
the Little-O notation, then the MI has the following asymptotic upper bound
for integer $N$,
\end{subequations}
\begin{equation}
I\leq I_{G}+o\left(  1\right)  \text{.} \label{Lma1.2}%
\end{equation}

\end{lemma}

\begin{lemma}
\label{Lemma 2} If conditions \textbf{C1 }and\textbf{ C2 }hold, $\xi=O\left(
N^{-1}\right)  $, then the MI has an asymptotic lower bound for integer $N$,
\begin{equation}
I\geq I_{G}+O\left(  N^{-1}\right)  \text{.} \label{Lma2}%
\end{equation}
Moreover, if condition \textbf{C1 } holds but Eqs. (\ref{C1.b1}) and
(\ref{C1.b2}) are replaced by (\ref{Lma1.1a}) and (\ref{Lma1.1b}), and
inequality (\ref{C2.b}) in \textbf{C2} also holds for $\eta>0$, then the MI
has the following asymptotic lower bound for integer $N$,
\begin{equation}
I\geq I_{G}+o\left(  1\right)  \text{.} \label{Lma2a}%
\end{equation}

\end{lemma}

\begin{theorem}
\label{Theorem 1}If conditions \textbf{C1 }and\textbf{ C2} hold, $\xi=O\left(
N^{-1}\right)  $, then the MI has the following asymptotic equality for
integer $N$,
\begin{equation}
{I=I_{G}}+O\left(  N^{-1}\right)  \text{.} \label{Thm1}%
\end{equation}
For more relaxed conditions, suppose condition \textbf{C1 }holds\textbf{\ }but
Eqs. (\ref{C1.b1}) and (\ref{C1.b2}) are replaced by (\ref{Lma1.1a}) and
(\ref{Lma1.1b}), and inequality (\ref{C2.b}) in \textbf{C2} also holds for
$\eta>0$, then the MI has an asymptotic equality for integer $N$,
\begin{equation}
I=I_{G}+o\left(  1\right)  \text{.} \label{Thm1.1}%
\end{equation}

\end{theorem}

\begin{theorem}
\label{Theorem 1a}Suppose $\mathbf{J}(\mathbf{x})$ and $\mathbf{G}%
(\mathbf{x})$\ are symmetric and positive-definite. Let
\begin{align}
&  \varsigma=\left\langle {\mathrm{Tr}}\left(  \boldsymbol{\Psi}%
(\mathbf{x})\right)  \right\rangle _{\mathbf{x}}\text{,}\label{Thm1a.1}\\
&  \boldsymbol{\Psi}(\mathbf{x})=\mathbf{J}^{-1/2}(\mathbf{x})\mathbf{P}%
(\mathbf{x})\mathbf{J}^{-1/2}(\mathbf{x})\text{,} \label{Thm1a.1a}%
\end{align}
then%
\begin{equation}
I_{G}\leq I_{F}+\frac{\varsigma}{2}\text{,} \label{Thm1a.2}%
\end{equation}
where ${\mathrm{Tr}}\left(  \cdot\right)  $\ indicating matrix trace;
moreover, if $\mathbf{P}(\mathbf{x})$ is positive-semidefinite, then
\begin{equation}
0\leq I_{G}-I_{F}\leq\frac{\varsigma}{2}\text{.} \label{Thm1a.3}%
\end{equation}
On the other hand, if
\begin{equation}
\varsigma_{1}=\left\langle \left\Vert \boldsymbol{\Psi}(\mathbf{x})\right\Vert
\right\rangle _{\mathbf{x}}=O(N^{-\beta}) \label{Thm1a.3a}%
\end{equation}
for some $\beta>0$, then
\begin{equation}
I_{G}=I_{F}+O(N^{-\beta})\text{.} \label{Thm1a.4}%
\end{equation}

\end{theorem}

\begin{remark}
\label{Remark 2}In general, we only need to assume that $p(\mathbf{x})$ and
$p(\mathbf{r}|\mathbf{x})$ are piecewise twice continuously differentiable for
$\mathbf{x}\in{{\mathcal{X}}}$. In this case, \textbf{Lemma \ref{Lemma 1}%
}\textit{,\ }\textbf{Lemma \ref{Lemma 2}} and \textbf{Theorem \ref{Theorem 1}
}can still be established. For more general cases, such as discrete or
continuous inputs, we have also derived a general approximation formula for MI
from which we can easily derive formula for $I_{G}$ and which will be
discussed in separate paper.\qed
\end{remark}

\subsection{Approximations of Mutual Information in Neural Populations with
Finite Size}

\label{Sec:2.2a}In the preceding section we have provided several bounds,
including both lower and upper bounds, and asymptotic relationships for the
true MI in the large $N$ (network size) limit. In the following, we will
discuss effective approximations to the true MI in the case of finite $N$.
Here we only consider the case of continuous inputs and will discuss the case
of discrete inputs in another paper.

\textbf{Theorem \ref{Theorem 1}} tells us that under suitable conditions, we
can use $I_{G}$\ to approximate $I$ for a large but finite $N$ (e.g. $N\gg
K$); that is
\begin{equation}
I\simeq I_{G}\text{.} \label{I=IG}%
\end{equation}
Moreover, by \textbf{Theorem \ref{Theorem 1a},} we know that if $\varsigma
\approx0$ with positive-semidefinite $\mathbf{P}(\mathbf{x})$ or
$\varsigma_{1}\approx0$\ holds (see Eqs. \ref{Thm1a.1} and \ref{Thm1a.3a}),
then by (\ref{Thm1a.3}), (\ref{Thm1a.4}) and (\ref{I=IG}) we have
\begin{equation}
I\simeq I_{G}\simeq I_{F}\text{.} \label{I=IF}%
\end{equation}

Define%
\begin{align}
&  \mathbf{\tilde{G}}(\mathbf{x})=\mathbf{J}(\mathbf{x})+\mathbf{P}\left(
\mathbf{x}\right)  +\mathbf{Q}\left(  \mathbf{x}\right)  \text{,}\label{GQ}\\
&  \tilde{I}_{G}=\dfrac{1}{2}\left\langle \ln\left(  \det\left(
\dfrac{\mathbf{\tilde{G}}(\mathbf{x})}{2\pi e}\right)  \right)  \right\rangle
_{\mathbf{x}}+H(X)\text{,} \label{IG1}%
\end{align}
where $\mathbf{\tilde{G}}(\mathbf{x})$ is positive-definite, $\mathbf{Q}%
\left(  \mathbf{x}\right)  $\ is a symmetric matrix depending on $\mathbf{x}$
and ${\left\Vert \mathbf{Q}\left(  \mathbf{x}\right)  \right\Vert }=O(1)$.
Suppose ${\left\Vert \mathbf{\tilde{G}}^{-1}\left(  \mathbf{x}\right)
\right\Vert }=O\left(  N^{-1}\right)  $, if we replace $I_{G}$\ by $\tilde
{I}_{G}$ in \textbf{Theorem \ref{Theorem 1}},\ then we can prove equations
(\ref{Thm1}) and (\ref{Thm1.1}) in a manner similar to the proof of
\textbf{Theorem \ref{Theorem 1}}. Considering a special case where $\left\Vert
\mathbf{P}(\mathbf{x})\right\Vert \rightarrow{0}$, $\det\left(  \mathbf{J}%
(\mathbf{x})\right)  ={O}\left(  1\right)  $ (e.g. $\mathrm{rank}\left(
\mathbf{J}(\mathbf{x})\right)  <K$) and ${\left\Vert \mathbf{G}^{-1}\left(
\mathbf{x}\right)  \right\Vert }\neq O\left(  N^{-1}\right)  $, then we can no
longer use the asymptotic formulas in \textbf{Theorem \ref{Theorem 1}}.
However, if we substitute $\mathbf{\tilde{G}}(\mathbf{x})$\ for $\mathbf{G}%
(\mathbf{x})$ by choosing an appropriate $\mathbf{Q}\left(  \mathbf{x}\right)
$ such that $\mathbf{\tilde{G}}(\mathbf{x})$ is positive-definite and
${\left\Vert \mathbf{\tilde{G}}^{-1}\left(  \mathbf{x}\right)  \right\Vert
}=O\left(  N^{-1}\right)  $, then we can use (\ref{Thm1}) or (\ref{Thm1.1}) as
the asymptotic formulas.

If we assume $\mathbf{G}(\mathbf{x})$\ and $\mathbf{\tilde{G}}(\mathbf{x}%
)$\ are positive-definite and
\begin{equation}
\zeta=\left\langle \left\Vert \mathbf{Q}(\mathbf{x})\mathbf{\tilde{G}}%
^{-1}\left(  \mathbf{x}\right)  \right\Vert \right\rangle _{\mathbf{x}%
}=O(N^{-\beta})\text{, }\beta>0\text{,} \label{GErr_0}%
\end{equation}
then similar to the proof of \textbf{Theorem \ref{Theorem 1a},} we have
\begin{align}
&  {\left\langle \ln\left(  \det\left(  \mathbf{G}(\mathbf{x})\right)
\right)  \right\rangle _{\mathbf{x}}}\nonumber\\
&  ={\left\langle \ln\left(  \det\left(  \mathbf{\tilde{G}}(\mathbf{x}%
)\right)  \right)  \right\rangle _{\mathbf{x}}+\left\langle \ln\left(
\det\left(  \mathbf{I}_{K}-\mathbf{Q}(\mathbf{x})\mathbf{\tilde{G}}%
^{-1}\left(  \mathbf{x}\right)  \right)  \right)  \right\rangle _{\mathbf{x}}%
}\nonumber\\
&  ={\left\langle \ln\left(  \det\left(  \mathbf{\tilde{G}}(\mathbf{x}%
)\right)  \right)  \right\rangle _{\mathbf{x}}}+O(N^{-\beta}) \label{GErr}%
\end{align}
and%
\[
\tilde{I}_{G}={I_{G}}+O(N^{-\beta})\text{.}%
\]
For large $N$, we usually have $\tilde{I}_{G}\simeq I_{G}$.

It is more convenient to redefine the following quantities:
\begin{align}
&  \mathbf{Q}\left(  \mathbf{x}\right)  ={\mathbf{P}}_{+}-\mathbf{P}\left(
\mathbf{x}\right)  \text{,}\label{Q_x}\\
&  {\mathbf{P}}_{+}={\left\langle \dfrac{\partial\ln p(\mathbf{x})}%
{\partial\mathbf{x}}\dfrac{\partial\ln p(\mathbf{x})}{\partial\mathbf{x}^{T}%
}\right\rangle _{\mathbf{x}}}\text{,}\label{P0}\\
&  \mathbf{G}_{+}\left(  {\mathbf{x}}\right)  =\mathbf{\tilde{G}}%
(\mathbf{x})=\mathbf{J}(\mathbf{x})+{\mathbf{P}}_{+}\text{,} \label{G0}%
\end{align}
and%
\begin{equation}
I_{G_{+}}=\tilde{I}_{G}=\dfrac{1}{2}\left\langle \ln\left(  \det\left(
\dfrac{\mathbf{G}_{+}(\mathbf{x})}{2\pi e}\right)  \right)  \right\rangle
_{\mathbf{x}}+H(X)\text{.} \label{IG0}%
\end{equation}
Notice that if $p(\mathbf{x})$ is twice differentiable for $\mathbf{x}$ and
\begin{equation}
\int_{{\mathcal{X}}}\frac{\partial^{2}p(\mathbf{x})}{\partial\mathbf{x}%
\partial\mathbf{x}^{T}}d\mathbf{x}=\mathbf{0}\text{,} \label{ipxx}%
\end{equation}
then
\begin{equation}
\mathbf{P}_{+}=\left\langle \mathbf{P}\left(  \mathbf{x}\right)  \right\rangle
_{\mathbf{x}}=\left\langle \frac{1}{p(\mathbf{x})}\frac{\partial
^{2}p(\mathbf{x})}{\partial\mathbf{x}\partial\mathbf{x}^{T}}\right\rangle
_{\mathbf{x}}-\left\langle \frac{\partial^{2}\ln p(\mathbf{x})}{\partial
\mathbf{x}\partial\mathbf{x}^{T}}\right\rangle _{\mathbf{x}}\text{.}
\label{PPx}%
\end{equation}
For example, if $p(\mathbf{x})$ is a normal distribution, $p(\mathbf{x}%
)={\mathcal{N}\left(  {\boldsymbol{\mu}}\text{,\thinspace}\boldsymbol{\Sigma
}\right)  }$, then%
\begin{equation}
\mathbf{P}\left(  \mathbf{x}\right)  ={\mathbf{P}}_{+}=\boldsymbol{\Sigma
}^{-1}\text{.} \label{PN}%
\end{equation}
Similar to the proof of \textbf{Theorem \ref{Theorem 1a}}, we can prove that
\begin{equation}
0\leq I_{G_{+}}-I_{F}\leq\frac{\varsigma_{+}}{2}\text{,} \label{IG+-IF}%
\end{equation}
where%
\begin{equation}
\varsigma_{+}=\left\langle {\mathrm{Tr}}\left(  \mathbf{P_{+}J}^{-1}%
(\mathbf{x})\right)  \right\rangle _{\mathbf{x}}\text{.} \label{var+}%
\end{equation}

We find that $I_{G}$\ is often a good approximation of MI $I$ even for
relatively small $N$. However, we cannot guarantee that $\mathbf{P}%
(\mathbf{x})$ is always positive-semidefinite in Eqs. (\ref{Gx}), and as a
consequence, it may happen that $\det\left(  \mathbf{G}(\mathbf{x})\right)
$\ is very small for small $N$, $\mathbf{G}(\mathbf{x})$ is not
positive-definite and $\ln\left(  \det\left(  \mathbf{G}(\mathbf{x})\right)
\right)  $ is not a real number. In this case, $I_{G}$\ is not a good
approximation to $I$ but $I_{G_{+}}$\ is still a good approximation.
Generally, if $\mathbf{P}(\mathbf{x})$ is always positive-semidefinite, then
$I_{G}$ or $I_{G_{+}}$\ is a better approximation\ than $I_{F}$, especially
when $p(\mathbf{x})$ be close to a normal distribution.

In the following we will give an example of 1-D inputs. High-dimensional
inputs will be discussed in section \ref{Sec:3.1}.

\subsubsection{A Numerical Comparison for 1-D Stimuli}

Considering the Poisson neuron model (see Eq. \ref{PoissNeuron}\ in section
\ref{Sec:4.1.2} for details), the tuning curve of the \textit{n}-th neuron,
$f\left(  x\text{\textrm{;\thinspace}}\theta_{n}\right)  $, takes the form of
circular normal or von Mises distribution
\begin{equation}
f\left(  x\text{\textrm{;\thinspace}}\theta_{n}\right)  =A\exp\left(  -\left(
\tfrac{T}{2\pi\sigma_{f}}\right)  ^{2}\left(  1-\cos\left(  \tfrac{2\pi}%
{T}\left(  x-\theta_{n}\right)  \right)  \right)  \right)  \text{,}%
\label{Exm_fx}%
\end{equation}
where $x\in\left[  -T/2\text{,\thinspace}T/2\right)  $, $\theta_{n}\in\left[
-T_{\theta}/2\text{,\thinspace}T_{\theta}/2\right]  $, $n\in\left\{
1\text{,\thinspace}2\text{,\thinspace}\cdots\text{,\thinspace}N\right\}  $,
with $T=\pi$, $T_{\theta}=1$, $\sigma_{f}=0.5$ and $A=20$, and the centers
$\theta_{1}$,\thinspace$\theta_{2}$,\thinspace$\cdots$,\thinspace$\theta_{N}$
of the $N$ neurons are uniformly distributed on interval $\left[  -T_{\theta
}/2\text{,\thinspace}T_{\theta}/2\right]  $, i.e., $\theta_{n}=\left(
n-1\right)  d_{\theta}-T_{\theta}/2$, with $d_{\theta}=T_{\theta}/(N-1)\ $and
$N\geq2$. Suppose the distribution of $1$-D continuous input $x$ ($K=1$)
$p(x)$\ has the form\
\begin{equation}
p(x)=Z^{-1}\exp\left(  -\left(  \tfrac{T}{2\pi\sigma_{p}}\right)  ^{2}\left(
1-\cos\left(  \tfrac{2\pi}{T}x\right)  \right)  \right)  \text{,}%
\label{Exm_px}%
\end{equation}
where $\sigma_{p}$ is a constant set to $\pi/4$, and $Z$ is the normalization
constant. Figure~1A shows graphs of the input distribution $p(x)$ and the
tuning curves $f\left(  x\text{\textrm{;\thinspace}}\theta\right)  $ with
different centers $\theta=-\pi/4$, $0$, $\pi/4$.

To evaluate the precision of the approximation formulas, we use Monte Carlo
(MC) simulation to approximate MI $I$. For MC simulation, we first sample an
input $x_{j}$ by the distribution $p(x)$, then generate the neural response
$\mathbf{r}_{j}$ by the conditional distribution $p(\mathbf{r}_{j}|x_{j})$,
where\ $j=1$,\thinspace$2$,\thinspace$\cdots$,\thinspace$j_{\mathrm{\max}}$.
The value of MI by MC simulation is calculated by
\begin{equation}
I_{MC}^{\ast}=\frac{1}{j_{\mathrm{\max}}}\sum\limits_{j=1}^{j_{\mathrm{\max}}%
}\ln\left(  \frac{p(\mathbf{r}_{j}|x_{j})}{p(\mathbf{r}_{j})}\right)  \text{,}
\label{IMC*}%
\end{equation}
where $p(\mathbf{r}_{j})$ is given by%
\begin{equation}
p(\mathbf{r}_{j})=\sum\limits_{m=1}^{M}p(\mathbf{r}_{j}|x_{m})p(x_{m})\text{,}
\label{prj}%
\end{equation}
and $x_{m}=\left(  m-1\right)  T/M-T/2$ for $m\in\left\{  1\text{,\thinspace
}2\text{,\thinspace}\cdots\text{,\thinspace}M\right\}  $.

To evaluate the accuracy of MC simulation, we compute the standard deviation
\begin{equation}
I_{std}=\sqrt{\frac{1}{i_{\mathrm{\max}}}\sum\limits_{i=1}^{i_{\mathrm{\max}}%
}\left(  I_{MC}^{i}-I_{MC}\right)  ^{2}}\text{,} \label{Istd}%
\end{equation}
where
\begin{align}
I_{MC}^{i}  &  =\frac{1}{j_{\mathrm{\max}}}\sum\limits_{j=1}^{j_{\mathrm{\max
}}}\ln\left(  \frac{p(\mathbf{r}_{\Gamma_{j,i}}|x_{\Gamma_{j,i}}%
)}{p(\mathbf{r}_{\Gamma_{j,i}})}\right)  \text{,}\label{IMCi}\\
I_{MC}  &  =\frac{1}{i_{\mathrm{\max}}}\sum\limits_{i=1}^{i_{\mathrm{\max}}%
}I_{MC}^{i}\text{,} \label{IMC}%
\end{align}
and $\Gamma_{j,i}\in\left\{  1\text{,\thinspace}2\text{,\thinspace}%
\cdots\text{,\thinspace}j_{\mathrm{\max}}\right\}  $ is the $\left(
j,i\right)  $-th entry of the matrix $\boldsymbol{\Gamma}\in\mathcal{%
\mathbb{N}
}^{j_{\mathrm{\max}}\times i_{\mathrm{\max}}}$ with samples taken randomly
from the integer set $\{1$,\thinspace$2$,\thinspace$\cdots$,\thinspace
$j_{\mathrm{\max}}\}$ by a uniform distribution. Here we set $j_{\mathrm{\max
}}=5\times10^{5}$, $i_{\mathrm{\max}}=100$ and $M=10^{3}$.

For different $N\in\{2$,$\,3$,$\,4$,$\,6$,$\,10$,$\,14$,$\,20$,$\,30$%
,$\,50$,$\,100$,$\,200$,$\,400$,$\,700$,$\,1000\}$, we compare $I_{MC}$\ with
$I_{G}$, $I_{G_{+}}$ and $I_{F}$, which are illustrated in Figure 1B--D. Here
we define the relative error of approximation, e.g., for $I_{G}$, as
\begin{equation}
DI_{G}=\frac{I_{G}-I_{MC}}{I_{MC}}\text{,} \label{RLError}%
\end{equation}
and the relative standard deviation
\begin{equation}
DI_{std}=\frac{I_{std}}{I_{MC}}\text{.} \label{DIstd}%
\end{equation}
Figure 1B shows how the values of $I_{MC}$, $I_{G}$, $I_{G_{+}}$ and $I_{F}$
change with neuron number $N$, and Figure 1C and 1D show their relative errors
and the absolute values of the relative errors with respect to $I_{MC}$. From
Figure 1B--D we can see that the values of $I_{G}$, $I_{G_{+}}$ and $I_{F}$
are all very close to one another and the absolute values of their relative
errors are all very small. The absolute values\ are less than $1\%$ when
$N\geq10$ and less than $0.1\%$ when $N\geq100$. However, for the
high-dimensional inputs, there will be a big difference between $I_{G}$,
$I_{G_{+}}$ and $I_{F}$ in many cases (see section \ref{Sec:3.1} for more details).

\begin{figure}[ptbh]
\centering
\includegraphics[width= .96\columnwidth]{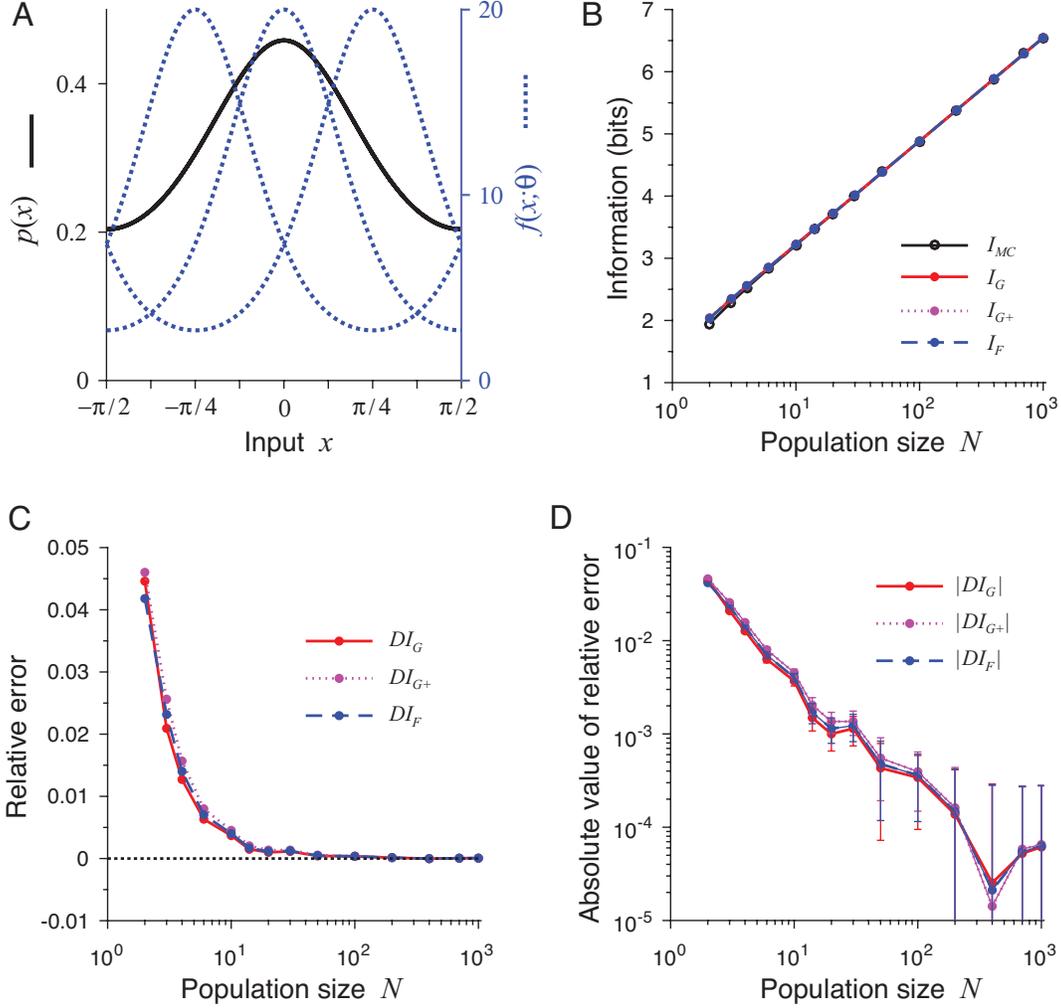}
\caption{A comparison of approximations $I_{MC}$, $I_{G}$,
$I_{G_{+}}$ and $I_{F}$ for one-dimensional input stimuli. All of them were
almost equally good, even for small population size $N$. (A) The stimulus
distribution $p(x)$ and tuning curves $f\left(  x\text{\textrm{;\thinspace}%
}\theta\right)  $ with different centers $\theta=-\pi/4$, $0$, $\pi/4$. (B)
The values of $I_{MC}$, $I_{G}$, $I_{G_{+}}$ and $I_{F}$ all increase with
neuron number $N$. (C) The relative errors $DI_{G}$, $DI_{G_{+}}$ and $DI_{F}$
for the results in panel B. (D) The absolute values of the relative errors
$\left\vert DI_{G}\right\vert $, $\left\vert DI_{G_{+}}\right\vert $, and
$\left\vert DI_{F}\right\vert $, with error bars showing standard deviations
of repeated trials. }%
\label{Fig1}%
\end{figure}

\section{Statistical Estimators and Neural Population Decoding}

\label{Sec:3}Given the neural response $\mathbf{r}$ elicited by the input
$\mathbf{x}$, we may infer or estimate the input $\mathbf{x}$ from the
response. This procedure is sometimes referred to as decoding from the
response. We need to choose an efficient estimator, or a function
$\mathbf{\hat{x}}=\mathbf{\hat{x}}(\mathbf{r})$ that maps the response
$\mathbf{r}$ to an estimate $\mathbf{\hat{x}}$ of the true stimulus
$\mathbf{x}$. The Maximum Likelihood (ML) estimator defined by
\begin{equation}
\mathbf{\hat{x}}(\mathbf{r})={\arg\max_{\mathbf{x}}\ }p(\mathbf{r}%
|\mathbf{x})={\arg\max_{\mathbf{x}}\ }l(\mathbf{r}|\mathbf{x}) \label{MLx}%
\end{equation}
is known to be efficient in large $N$ limit. According to the Cram\'{e}r-Rao
lower bound \citep{Rao(1945-information)}, we have the following relationship
between the covariance matrix of any unbiased estimator, $\boldsymbol{\Sigma
}_{\mathbf{\hat{x}}}$, and the FI matrix $\mathbf{J}\left(  \mathbf{x}\right)
$,%
\begin{equation}
\boldsymbol{\Sigma}_{\mathbf{\hat{x}}}=\left\langle \left(  \hat{\mathbf{x}%
}(\mathbf{r})-\mathbf{x}\right)  \left(  \mathbf{\hat{x}}(\mathbf{r}%
)-\mathbf{x}\right)  ^{T}\right\rangle _{\mathbf{r}|\mathbf{x}}\geq
\mathbf{J}^{-1}(\mathbf{x})\text{,} \label{CramerRao}%
\end{equation}
where $\mathbf{\hat{x}}(\mathbf{r})$\ is an unbiased estimation of
$\mathbf{x}$ from the response $\mathbf{r}$, and $\mathbf{A}\geq\mathbf{B}$
means that matrix $\mathbf{A}-\mathbf{B}$\ is positive-semidefinite. Thus
\begin{align}
I_{F}  &  =\frac{1}{2}\left\langle \ln\left(  \det\left(  \frac{\mathbf{J}%
(\mathbf{x})}{2\pi e}\right)  \right)  \right\rangle _{\mathbf{x}}+H\left(
X\right) \nonumber\\
&  \geq\frac{1}{2}\left\langle \ln\left(  \det\left(  \frac{\boldsymbol{\Sigma
}_{\mathbf{\hat{x}}}^{-1}}{2\pi e}\right)  \right)  \right\rangle
_{\mathbf{x}}+H\left(  X\right)  =I_{var}\text{.} \label{Fisher}%
\end{align}
On the other hand, the MI between $X$\ and $\hat{X}$\ is given by%
\begin{equation}
\hat{I}=H(\hat{X}\mathbf{)}-\left\langle H(\hat{X}\mathbf{|}X\mathbf{)}%
\right\rangle _{\mathbf{\hat{x},x}}\text{,} \label{I^}%
\end{equation}
where $H(\hat{X}\mathbf{)}$ is the entropy of random variable $\hat{X}$ and
$H(\hat{X}\mathbf{|}X\mathbf{)}$ is its conditional entropy of random variable
$\hat{X}$ given $X$. Since the maximum entropy probability distribution is
Gaussian, $H(\hat{X}\mathbf{|}X)$\ satisfies%
\begin{equation}
H(\hat{X}\mathbf{|}X\mathbf{)}\leq\frac{1}{2}\ln\left(  \det\left(  2\pi
e\boldsymbol{\Sigma}_{\mathbf{\hat{x}}}\right)  \right)  \text{.}
\label{H(x^|x)}%
\end{equation}
Therefore, from (\ref{I^}) and (\ref{H(x^|x)}), we get%
\begin{equation}
\hat{I}\geq\frac{1}{2}\left\langle \ln\left(  \det\left(  \dfrac
{\boldsymbol{\Sigma}_{\mathbf{\hat{x}}}^{-1}}{2\pi e}\right)  \right)
\right\rangle _{\mathbf{x}}+H(\hat{X}\mathbf{)}=\hat{I}_{var}\text{.}
\label{I^>Ix^}%
\end{equation}
The data processing inequality \citep{Cover(2006-BK-elements)} states that
post-processing cannot increase information, so that we have
\begin{equation}
I\geq\hat{I}\geq\hat{I}_{var}\text{.} \label{DPI}%
\end{equation}

Here we can not directly obtain $I\geq I_{F}$ as in \cite{Brunel(1998-mutual)}
when $H(\hat{X})=H\left(  X\right)  $ and $I_{var}=\hat{I}_{var}$. The
simulation results in Figure 1 also show that $I_{F}$ is not a lower bound of
$I$.

For biased estimators, the van Trees' Bayesian Cram\'{e}r-Rao bound
\citep{VanTrees(2007-BK-bayesian)} provides a lower bound:%
\begin{equation}
\left\langle \boldsymbol{\Sigma}_{\mathbf{\hat{x}}}\right\rangle _{\mathbf{x}%
}=\left\langle \left\langle (\mathbf{\hat{x}}(\mathbf{r})-\mathbf{x}%
)(\mathbf{\hat{x}}(\mathbf{r})-\mathbf{x})^{T}\right\rangle _{\mathbf{r}%
|\mathbf{x}}\right\rangle _{\mathbf{x}}\geq\left(  \left\langle \mathbf{J}%
(\mathbf{x})\right\rangle _{\mathbf{x}}+\mathbf{P}_{+}\right)  ^{-1}%
=\left\langle \mathbf{G}_{+}(\mathbf{x})\right\rangle _{\mathbf{x}}%
^{-1}\text{.} \label{vanTrees}%
\end{equation}
It follows from (\ref{IG0}), (\ref{I^>Ix^}) and (\ref{vanTrees}) that%
\begin{align}
I_{G_{+}}  &  \leq\frac{1}{2}\ln\left(  \det\left(  \frac{\left\langle
\mathbf{G}_{+}(\mathbf{x})\right\rangle _{\mathbf{x}}}{2\pi e}\right)
\right)  +H(X)=I_{VT}\text{,}\label{vt1}\\
I_{VT}  &  \geq\frac{1}{2}\ln\left(  \det\left(  \frac{\left\langle
\boldsymbol{\Sigma}_{\mathbf{\hat{x}}}\right\rangle _{\mathbf{x}}^{-1}}{2\pi
e}\right)  \right)  +H(X)=\tilde{I}_{var}\text{,}\label{vt2}\\
I_{var}  &  \geq\tilde{I}_{var}\text{.} \label{vt3}%
\end{align}

We may also regard decoding as Bayesian inference. By Bayes' rule,
\begin{equation}
p(\mathbf{x}|\mathbf{r})=\frac{p(\mathbf{r}|\mathbf{x})p(\mathbf{x}%
)}{p(\mathbf{r})}\text{.} \label{MAP.1}%
\end{equation}
According to the Bayesian decision theory, if we know the response
$\mathbf{r}$, from the prior $p(\mathbf{x})$ and the likelihood $p(\mathbf{r}%
|\mathbf{x})$, we can infer an estimation of the true stimulus $\mathbf{x}$,
$\mathbf{\hat{x}}(\mathbf{r})$, for example,
\begin{equation}
\mathbf{\hat{x}}(\mathbf{r})={\arg\max_{\mathbf{x}}\ }p(\mathbf{x}%
|\mathbf{r})={\arg\max_{\mathbf{x}}\ }L(\mathbf{r}|\mathbf{x})\text{,}
\label{MAP.2}%
\end{equation}
which is also called Maximum A Posteriori (MAP) estimation.

Consider a loss function $\varphi(\mathbf{\hat{x}}(\mathbf{r})|\mathbf{x})$
for estimation,
\begin{equation}
\varphi(\mathbf{\hat{x}}(\mathbf{r})|\mathbf{x})=-\ln p(\mathbf{x}%
|\mathbf{r})\text{,} \label{MAP.3}%
\end{equation}
which is minimized when $p(\mathbf{x}|\mathbf{r})$ reaches its maximum. Now
the conditional risk is
\begin{equation}
R(\mathbf{\hat{x}}(\mathbf{r})|\mathbf{r})=\left\langle \varphi(\mathbf{\hat
{x}}(\mathbf{r})|\mathbf{x})\right\rangle _{\mathbf{x}|\mathbf{r}}\text{,}
\label{MAP.4}%
\end{equation}
and the overall risk is
\begin{equation}
R_{o}=\left\langle R(\mathbf{\hat{x}}(\mathbf{r})|\mathbf{r})\right\rangle
_{\mathbf{r}}=\left\langle \left\langle \varphi(\mathbf{\hat{x}}%
(\mathbf{r})|\mathbf{x})\right\rangle _{\mathbf{x}|\mathbf{r}}\right\rangle
_{\mathbf{r}}=-\left\langle \ln p(\mathbf{x}|\mathbf{r})\right\rangle
_{\mathbf{x}\text{,\thinspace}\mathbf{r}}\text{.} \label{MAP.5}%
\end{equation}
Then it follows from (\ref{MI1}) and (\ref{MAP.5}) that
\begin{equation}
I=\left\langle \ln p(\mathbf{x}|\mathbf{r})\right\rangle _{\mathbf{r}%
\text{,\thinspace}\mathbf{x}}+H(X)=-R_{o}+H(X)\text{.} \label{MAP.6}%
\end{equation}
Comparing (\ref{IG}), (\ref{I=IG}) and (\ref{MAP.6}), we find
\begin{equation}
R_{o}\simeq-\frac{1}{2}\left\langle \ln\left(  \det\left(  \frac
{\mathbf{G}(\mathbf{x})}{2\pi e}\right)  \right)  \right\rangle _{\mathbf{x}%
}\text{.} \label{MAP.7}%
\end{equation}
Hence, maximizing MI $I$ (or $I_{G}$) means minimizing the overall risk
$R_{o}$ for a determinate $H(X)$. Therefore, we can get the optimal Bayesian
inference via optimizing MI $I$ (or $I_{G}$).

By the Cram\'{e}r-Rao lower bound, we know that the inverse of FI matrix
$\mathbf{J}^{-1}(\mathbf{x})$ reflects the accuracy of decoding (see Eq.
\ref{CramerRao}). $\mathbf{P}(\mathbf{x})$ provides some knowledge about the
prior distribution $p(\mathbf{x})$; for example, $\mathbf{P}^{-1}\left(
\mathbf{x}\right)  $ is the covariance matrix of input $\mathbf{x}$ when
$p(\mathbf{x})$ is a normal distribution. $\left\Vert \mathbf{P}%
(\mathbf{x})\right\Vert $ is small for a flat prior (poor prior) and large for
a sharp prior (good prior). Hence, if the prior $p(\mathbf{x})$ is flat or
poor and the knowledge about model is rich, then the MI $I$ is governed by the
knowledge of model, which results in a small $\varsigma_{1}$ (Eq.
\ref{Thm1a.3a}) and $I\simeq I_{G}\simeq I_{F}$. Otherwise, the prior
knowledge has a great influence on MI $I$, which results in a large
$\varsigma_{1}$ and $I\simeq I_{G} \not \simeq I_{F}$.

\section{Variable Transformation and Dimensionality Reduction in Neural
Population Coding}

\label{Sec:3.0} For low-dimensional input $\mathbf{x}$ and large $N$, both
$I_{G}$ are $I_{F}$ are good approximations of MI $I$, but for
high-dimensional input $\mathbf{x}$, a large value of $\varsigma_{1}$ may lead
to a large error of $I_{F}$, in which case $I_{G}$ (or $I_{G_{+}}$) is a
better approximation. It is difficult to directly apply the approximation
formula $I\simeq I_{G}$ when we do not have an explicit expression of
$p\left(  \mathbf{x}\right)  $ or $\mathbf{P}\left(  \mathbf{x}\right)  $. For
many applications, we do not need to know the exact value of $I_{G}$ and only
care about the value of $\left\langle \ln\left(  \det\left(  \mathbf{G}%
(\mathbf{x})\right)  \right)  \right\rangle _{\mathbf{x}}$ (see section
\ref{Sec:4}). From (\ref{IG}), (\ref{I_Gau1}) and (\ref{PN}), we know that if
$p\left(  \mathbf{x}\right)  $ is close to a normal distribution, we can
easily approximate $\mathbf{P}\left(  \mathbf{x}\right)  $\ and $H(X)$ ot
obtain $\left\langle \ln\left(  \det\left(  \mathbf{G}(\mathbf{x})\right)
\right)  \right\rangle _{\mathbf{x}}$ and $I_{G}$. When $p\left(
\mathbf{x}\right)  $ is not a normal distribution, we can employ a technique
of variable transformation to make it closer to a normal distribution, as
discussed below.

\subsection{Variable Transformation}

\label{Sec:3.1}Suppose $\mathbf{T}:{{\mathcal{X}}}\rightarrow{{\mathcal{\tilde
{X}}}}$ is an invertible and differentiable mapping:
\begin{equation}
\mathbf{\tilde{x}}=\mathbf{T}(\mathbf{x})=\left(  T_{1}(\mathbf{x}%
)\text{,\thinspace}T_{2}(\mathbf{x})\text{,\thinspace}\cdots\text{,\thinspace
}T_{K}(\mathbf{x})\right)  ^{T}\text{,} \label{x}%
\end{equation}
$\mathbf{x}=\mathbf{T}^{-1}(\mathbf{\tilde{x}})$ and $\mathbf{\tilde{x}}%
\in{{\mathcal{\tilde{X}}}}\subseteq%
\mathbb{R}
^{K}$. Let $p(\mathbf{\tilde{x}})$ denotes the p.d.f. of random variable
$\tilde{X}$ and
\begin{equation}
p(\mathbf{r}|\mathbf{\tilde{x}})=\left.  p(\mathbf{r}|\mathbf{x})\right\vert
_{\mathbf{x}=\mathbf{T}^{-1}(\mathbf{\tilde{x}})}\text{.} \label{prxT}%
\end{equation}
Then we have the following conclusions, the proofs of which are given in Appendix.

\begin{theorem}
\label{Theorem 2} The MI is equivariant under the invertible transformations.
More specifically, for the above invertible transformation $\mathbf{T}$, the
MI $I(X;R)$ in (\ref{MI}) is equal to
\begin{equation}
I(\tilde{X};R)=\left\langle \ln\frac{p(\mathbf{r}|\mathbf{\tilde{x}}%
)}{p(\mathbf{r})}\right\rangle _{\mathbf{r}\text{\textbf{,}}\mathbf{\,\tilde
{x}}}\text{.} \label{IT}%
\end{equation}
Furthermore, suppose $p(\mathbf{\tilde{x}})$ and $p(\mathbf{r}|\tilde
{\mathbf{x}})$ fulfill the conditions \textbf{C1, C2} and\textbf{ }%
$\xi=O\left(  N^{-1}\right)  $, then we have
\begin{align}
I(\tilde{X};R)  &  =\tilde{I}_{G}+O\left(  N^{-1}\right)  \text{,}%
\label{ITP}\\
\tilde{I}_{G}  &  =\frac{1}{2}\left\langle \ln\left(  \det\left(
\frac{\mathbf{G}(\mathbf{\tilde{x}})}{2\pi e}\right)  \right)  \right\rangle
_{\mathbf{\tilde{x}}}+H({\tilde{X}})\nonumber\\
&  =\frac{1}{2}\left\langle \ln\left(  \det\left(  \frac{\mathbf{G}%
(\mathbf{x})}{2\pi e}\right)  \right)  \right\rangle _{\mathbf{x}}%
+H({X})\nonumber\\
&  =I_{G}\text{,} \label{Iap}%
\end{align}
where $H({\tilde{X}})$ is the entropy of random variable ${\tilde{X}}$ and
satisfies
\begin{equation}
H({\tilde{X}})=-\left\langle \ln p(\mathbf{\tilde{x}})\right\rangle
_{\mathbf{\tilde{x}}}={H(X)}+\left\langle {\ln\left\vert \det\left(
D\mathbf{T}(\mathbf{x})\right)  \right\vert }\right\rangle _{\mathbf{x}%
}\text{,} \label{Thm2.2}%
\end{equation}
and $D\mathbf{T}(\mathbf{x})$ denotes the Jacobian matrix of $\mathbf{T}%
(\mathbf{x})$,
\begin{equation}
\left(  D\mathbf{T}(\mathbf{x})\right)  _{i\text{,\thinspace}j}=\frac{\partial
T_{i}(\mathbf{x})}{\partial x_{j}}\text{, }\quad\forall i\text{,\thinspace
}j=1\text{,\thinspace}2\text{,\thinspace}\cdots\text{,\thinspace}K\text{.}
\label{Thm2.3}%
\end{equation}

\end{theorem}

\begin{corollary}
\label{Corollary 2} Suppose $p(\mathbf{r}|\mathbf{x})$ is a normal
distribution,
\begin{equation}
p(\mathbf{r}|\mathbf{x})=\mathcal{N}\left(  \mathbf{A}^{T}\mathbf{y}%
\text{,\thinspace}\mathbf{I}_{N}\right)  \text{,} \label{Cly2.1}%
\end{equation}
where $\mathbf{y}=\mathbf{f}\left(  \mathbf{B}^{T}\mathbf{x}\right)  =\left(
y_{1}\text{,\thinspace}y_{2}\text{,\thinspace}\cdots\text{,\thinspace}%
y_{K}\right)  ^{T}$, $y_{k}=f_{k}(\mathbf{b}_{k}^{T}\mathbf{x})$ for
$k=1$,\thinspace$2$,\thinspace$\cdots$,\thinspace$K$, $\mathbf{A}$ is a
deterministic $K\times N$\ matrix, $\mathbf{B}=\left[  \mathbf{b}%
_{1}\text{,\thinspace}\mathbf{b}_{2}\text{,\thinspace}\cdots\text{,\thinspace
}\mathbf{b}_{K}\right]  $ is a deterministic invertible matrix and $f_{k}$ is
an invertible and differentiable function. If $Y$ has also a normal
distribution, $p(\mathbf{y})={\mathcal{N}}\left(  \boldsymbol{\mu}%
_{\mathbf{f}}\text{,\thinspace}\boldsymbol{\Sigma}_{\mathbf{f}}\right)  $,
then
\begin{align}
I_{G}  &  =I_{G_{+}}=I(X\text{;\thinspace}R)=I(Y\text{;\thinspace
}R)\nonumber\\
&  =\dfrac{1}{2}\ln\left(  \det\left(  \dfrac{1}{2\pi e}\left(  \mathbf{AA}%
^{T}+\boldsymbol{\Sigma}_{\mathbf{f}}^{-1}\right)  \right)  \right)
+H(Y)\nonumber\\
&  =\dfrac{1}{2}\left\langle \ln\left(  \det\left(  \dfrac{1}{2\pi e}\left(
\mathbf{J}(\mathbf{x})+\mathbf{P}(\mathbf{x})\right)  \right)  \right)
\right\rangle _{\mathbf{x}}+H(X)\text{,} \label{Cly2.4}%
\end{align}
where
\begin{align}
&  H(Y)=\dfrac{1}{2}\ln\left(  \det\left(  2\pi e\boldsymbol{\Sigma
}_{\mathbf{f}}\right)  \right)  =H(X)+\left\langle \ln\left\vert \det\left(
\mathbf{D}(\mathbf{x})\right)  \right\vert \right\rangle _{\mathbf{x}}%
\text{,}\label{Cly2.5}\\
&  {\mathbf{D}(\mathbf{x})=\left(  f_{1}^{\prime}(\mathbf{b}_{1}^{T}%
\mathbf{x})\mathbf{b}_{1}\text{,\thinspace}f_{2}^{\prime}(\mathbf{b}_{2}%
^{T}\mathbf{x})\mathbf{b}_{2}\text{,\thinspace}\cdots\text{,\thinspace}%
f_{K}^{\prime}(\mathbf{b}_{K}^{T}\mathbf{x})\mathbf{b}_{K}\right)  ^{T}%
}\text{,}\label{Cly2.5b}\\
&  {f_{k}^{\prime}(\mathbf{b}_{k}^{T}\mathbf{x})=\left.  \dfrac{\partial
f_{k}(y_{k})}{\partial y_{k}}\right\vert _{y_{k}=\mathbf{b}_{k}^{T}\mathbf{x}%
}}\text{, }\quad\forall k=1\text{,\thinspace}2\text{,\thinspace}%
\cdots\text{,\thinspace}K\text{.} \label{Cly2.5c}%
\end{align}

\end{corollary}

\begin{remark}
From {Corollary \ref{Corollary 2}} and Eq. (\ref{PN}) we know that the
approximation accuracy for $I_{G}\simeq I(X;R)$\ is improved when we employ an
invertible transformation on\ the input random variable $X$\textbf{ } to make
the new random variable $Y$ closer to a normal distribution (see section
\ref{Sec:3.2}).\qed

\end{remark}

Consider the eigendecompositons of $\mathbf{AA}^{T}$ and $\boldsymbol{\Sigma
}_{\mathbf{f}}$\ as given by%
\begin{align}
&  \mathbf{AA}^{T}=\mathbf{U}_{\mathbf{A}}\boldsymbol{\hat{\Sigma}}%
\mathbf{U}_{\mathbf{A}}^{T}\text{,}\label{AA'}\\
&  \boldsymbol{\Sigma}_{\mathbf{f}}=\mathbf{U}_{\mathbf{f}}\boldsymbol{\tilde
{\Sigma}}\mathbf{U}_{\mathbf{f}}^{T}\text{,} \label{SIGf}%
\end{align}
where $\mathbf{U}_{\mathbf{A}}$ and $\mathbf{U}_{\mathbf{f}}$ are $K\times K$
orthogonal matrices; $\boldsymbol{\hat{\Sigma}}=\mathrm{diag}\left(
\hat{\sigma}_{1}^{2}\text{,\thinspace}\hat{\sigma}_{2}^{2}\text{,\thinspace
}\cdots\text{,\thinspace}\hat{\sigma}_{K}^{2}\right)  $ and
$\boldsymbol{\tilde{\Sigma}}=\mathrm{diag}\left(  \tilde{\sigma}_{1}%
^{2}\text{,\thinspace}\tilde{\sigma}_{2}^{2}\text{,\thinspace}\cdots
\text{,\thinspace}\tilde{\sigma}_{K}^{2}\right)  $ are $K\times K$ eigenvalue
matrices, $\hat{\sigma}_{1}\geq\hat{\sigma}_{2}\geq\cdots\geq\hat{\sigma}%
_{K}>0$ and $\tilde{\sigma}_{1}\geq\tilde{\sigma}_{2}\geq\cdots\geq
\tilde{\sigma}_{K}>0$. Then by (\ref{IF}) and (\ref{Cly2.4}) we have
\begin{align}
I_{G}  &  =I_{G_{+}}=I(X;R)=I(Y;R)\nonumber\\
&  =\dfrac{1}{2}\ln\left(  \det\left(  \dfrac{1}{2\pi e}\left(  \mathbf{U}%
_{\mathbf{A}}\boldsymbol{\hat{\Sigma}}\mathbf{U}_{\mathbf{A}}^{T}%
+\mathbf{U}_{\mathbf{f}}\boldsymbol{\tilde{\Sigma}}^{-1}\mathbf{U}%
_{\mathbf{f}}^{T}\right)  \right)  \right)  +H(Y)\text{,}\label{IG_1}\\
I_{F}  &  =\dfrac{1}{2}\ln\left(  \det\left(  \dfrac{\boldsymbol{\hat{\Sigma}%
}}{2\pi e}\right)  \right)  +H(Y)\text{,} \label{IF_1}%
\end{align}
and%
\begin{equation}
I_{F}-I_{G}=-\dfrac{1}{2}\ln\left(  \det\left(  \mathbf{I}_{K}%
+\boldsymbol{\hat{\Sigma}}^{-1/2}\mathbf{U}_{\mathbf{A}}^{T}\mathbf{U}%
_{\mathbf{f}}\boldsymbol{\tilde{\Sigma}}^{-1}\mathbf{U}_{\mathbf{f}}%
^{T}\mathbf{U}_{\mathbf{A}}\boldsymbol{\hat{\Sigma}}^{-1/2}\right)  \right)
\text{.} \label{Ixr-IF}%
\end{equation}
Now consider two special cases. If $\boldsymbol{\tilde{\Sigma}}=\mathbf{I}%
_{K}$, then by (\ref{Ixr-IF}) we get%
\begin{equation}
I_{F}-I_{G}=-\dfrac{1}{2}\sum_{k=1}^{K}\ln\left(  1+\hat{\sigma}_{k}%
^{-2}\right)  \text{.} \label{IG-IF_1}%
\end{equation}
If $\mathbf{U}_{\mathbf{A}}=\mathbf{U}_{\mathbf{f}}$, then
\begin{equation}
I_{F}-I_{G}=-\dfrac{1}{2}\sum_{k=1}^{K}\ln\left(  1+\hat{\sigma}_{k}%
^{-2}\tilde{\sigma}_{k}^{-2}\right)  \text{.} \label{IG-IF_2}%
\end{equation}
Here $\mathbf{J}(\mathbf{x})=\mathbf{U}_{\mathbf{A}}\boldsymbol{\hat{\Sigma}%
}\mathbf{U}_{\mathbf{A}}^{T}$, $\mathbf{P}^{-1}(\mathbf{x})=\mathbf{U}%
_{\mathbf{f}}\boldsymbol{\tilde{\Sigma}}\mathbf{U}_{\mathbf{f}}^{T}$. The FI
matrix $\mathbf{J}(\mathbf{x})$ and $\mathbf{P}^{-1}(\mathbf{x})$\ become
degenerate when $\hat{\sigma}_{K}^{2}\rightarrow0$ and $\tilde{\sigma}_{K}%
^{2}\rightarrow0$.

From (\ref{IG-IF_1}) and (\ref{IG-IF_2}) we see that if either $\mathbf{J}%
(\mathbf{x})$ or $\mathbf{P}^{-1}(\mathbf{x})$\ becomes degenerate, then
$(I_{F}-I_{G})\rightarrow-\infty$. This may happen for high-dimensional
stimuli. For a specific example, consider a random matrix $\mathbf{A}$ defined
as follows. Here we first generate $K\times N$\ elements $A_{k,n}$, ($k=1$,
$2$, $\cdots$, $K$; $n=1$, $2$, $\cdots$, $N$) from a normal distribution
${\mathcal{N}}\left(  0\text{,\thinspace}1\right)  $. Then each column of
matrix $\mathbf{A}$ is normalized by $A_{k,n}\leftarrow A_{k,n}{\Big /}\sqrt{%
{\textstyle\sum_{k=1}^{K}}
A_{k,n}^{2}}$. We randomly sample $M$ (set to $2\times10^{4}$) image patches
with size $w\times w$\ from Olshausen's nature image dataset
\citep{Olshausen(1996-emergence)} as the inputs. Each input image patch was
centered by subtracting its mean, i.e., $\mathbf{x}_{m}\leftarrow
\mathbf{x}_{m}-\frac{1}{K}\sum_{k=1}^{K}x_{k,m}$, then let $\mathbf{x}%
_{m}\leftarrow\mathbf{x}_{m}-\frac{1}{M}\sum_{m^{\prime}=1}^{M}\mathbf{x}%
_{m^{\prime}}$\ for $\forall m\in\{1$, $2$, $\cdots$, $M\}$. Define matrix
$\mathbf{X}=[\mathbf{x}_{1}$,\thinspace$\mathbf{x}_{2}$,\thinspace$\cdots
$,\thinspace$\mathbf{x}_{M}]$ and compute eigendecomposition%
\begin{equation}
\frac{1}{M}\mathbf{XX}^{T}=\mathbf{U}_{\mathbf{x}}\boldsymbol{\check{\Sigma}%
}\mathbf{U}_{\mathbf{x}}^{T}\text{,} \label{XX}%
\end{equation}
where $\mathbf{U}_{\mathbf{x}}$ is a $K\times K$ orthogonal matrix and
$\boldsymbol{\check{\Sigma}}=\mathrm{diag}\left(  \check{\sigma}_{1}%
^{2}\text{,\thinspace}\check{\sigma}_{2}^{2}\text{,\thinspace}\cdots
\text{,\thinspace}\check{\sigma}_{K}^{2}\right)  $ is a $K\times K$ eigenvalue
matrix with $\check{\sigma}_{1}\geq\check{\sigma}_{2}\geq\cdots\geq
\check{\sigma}_{K}>0$. Define%
\begin{equation}
\mathbf{y}=\mathbf{U}_{\mathbf{x}}^{T}\mathbf{x}\text{,} \label{y}%
\end{equation}
then%
\begin{equation}
\frac{1}{M}\sum_{m=1}^{M}\mathbf{y}_{m}\mathbf{y}_{m}^{T}=\boldsymbol{\check
{\Sigma}}\text{.} \label{3.1_1}%
\end{equation}

The distribution of random variable $Y$ can be approximated by a normal
distribution (see section \ref{Sec:3.2} for more details). When $p(\mathbf{y}%
)={\mathcal{N}\left(  \boldsymbol{\check{\mu}},\boldsymbol{\check{\Sigma}%
}\right)  }$, we have
\begin{align}
I_{G}  &  =I_{G_{+}}=I(X;R)=I(Y;R)\text{,}\label{3.1_2a}\\
I_{G}  &  =\dfrac{1}{2}\ln\left(  \det\left(  \dfrac{1}{2\pi e}\left(
\mathbf{AA}^{T}+\boldsymbol{\check{\Sigma}}^{-1}\right)  \right)  \right)
+H(Y)\nonumber\\
&  =\dfrac{1}{2}\ln\left(  \det\left(  \dfrac{1}{2\pi e}\left(
\boldsymbol{\check{\Sigma}}^{1/2}\mathbf{AA}^{T}\boldsymbol{\check{\Sigma}%
}^{1/2}+\mathbf{I}_{K}\right)  \right)  \right)  \text{,}\label{3.1_2b}\\
I_{F}  &  =\dfrac{1}{2}\ln\left(  \det\left(  \dfrac{\mathbf{AA}^{T}}{2\pi
e}\right)  \right)  +H(Y)\text{.} \label{3.1_2c}%
\end{align}
The error of approximation $I_{F}$\ is given by
\begin{align}
dI_{F}  &  =I_{F}-I(X;R) =I_{F}-I_{G}\nonumber\\
&  =-\dfrac{1}{2}\ln\left(  \det\left(  \mathbf{I}_{K}+(\mathbf{AA}^{T}%
)^{-1}\boldsymbol{\check{\Sigma}}^{-1}\right)  \right)  \text{,} \label{3.1_3}%
\end{align}
and the relative error for $I_{F}$ is
\begin{equation}
DI_{F}=\frac{dI_{F}}{I_{G}}\text{.} \label{3.1_4}%
\end{equation}

Figure 2A shows how the values of $I_{G}$\ and $I_{F}$\ vary with the input
dimension $K=w\times w$ and the number of neurons $N$ (with $w=2$, $4$, $6$,
$\cdots$, $30$ and $N=10^{4}$, $2\times10^{4}$, $5\times10^{4}$, $10^{5}$).
The relative error $DI_{F}$\ is shown in Figure~2B. The absolute value of the
relative error tends to decrease with $N$ but may grow quite large as $K$
increases. In Figure~2B, the largest absolute value of relative error
$\left\vert DI_{F}\right\vert $ is greater than $5000\%$, which occurs when
$K=900$ and $N=10^{4}$. Even the smallest $\left\vert DI_{F}\right\vert $ is
still greater than $80\%$, which occurs when $K=100$ and $N=10^{5}$. In this
example, $I_{F}$ is a bad approximation of MI $I$ whereas $I_{G}$ and
$I_{G_{+}}$ are strictly equal to the true MI $I$ across all parameters.

\begin{figure}[ptbh]
\centering
\includegraphics[width=0.93\columnwidth]{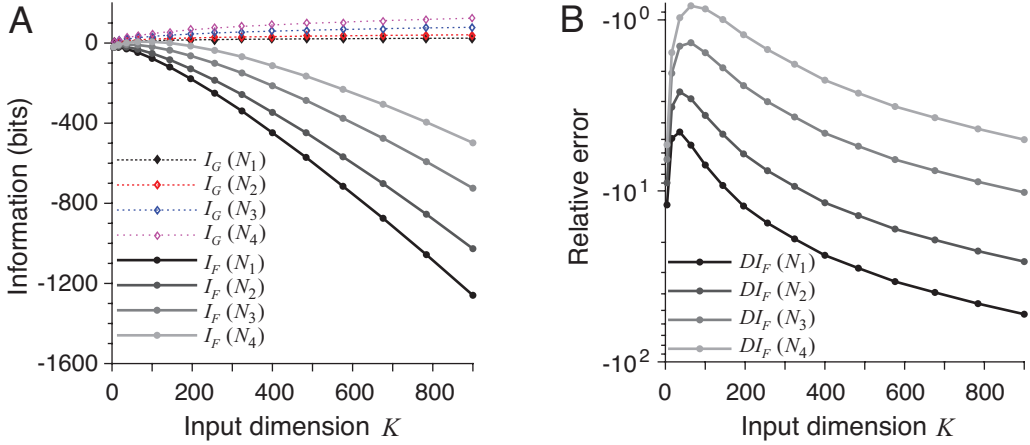}
\caption{A comparison of approximations $I_{G}$ and $I_{F}$ for
different input dimensions. Here $I_{G}$ is always equal to the true MI with
$I_{G}=I_{G_{+}}=I(X;R)$, whereas $I_{F}$ always has nonzero errors. (A) The
value $I_{G}$\ and $I_{F}$ vary with input dimension $K=w^{2}$ with $w=2$,
$4$, $6$, $\cdots$, $30$, and the number of neurons $N=N_{i}$ with
$N_{1}=10^{4}$, $N_{2}=2\times10^{4}$, $N_{3}=5\times10^{4}$, $N_{4}=10^{5}$.
(B) The relative error $DI_{F}$ changes with input dimension $K$ for different
$N$.}%
\label{Fig2}%
\end{figure}

\subsection{Dimensionality Reduction for Asymptotic Approximations}

Suppose $\mathbf{x}=(x_{1},\cdots,x_{K})^{T}$ is partitioned into two sets of
components, $\mathbf{x}=(\mathbf{x}_{1}^{T}$\textrm{,\thinspace}%
$\mathbf{x}_{2}^{T})^{T}$ with
\begin{align}
{\mathbf{x}_{1}}  &  {=({x}_{1}}\text{,\thinspace}{{x}}_{2}\text{,\thinspace
}\cdots\text{,\thinspace}{{x}}_{K_{1}}{)^{T}}\text{,}\label{x1x2}\\
{\mathbf{x}_{2}}  &  {=({x}_{K_{1}+1}}\text{,\thinspace}{{x}}_{K_{1}%
+2}\text{,\thinspace}\cdots\text{,\thinspace}{{x}_{K})^{T}}\text{,}
\label{x1x2b}%
\end{align}
where $\mathbf{x}_{1}\in{{\mathcal{X}}}_{1}\subseteq%
\mathbb{R}
^{K_{1}}$, $\mathbf{x}_{2}\in{{\mathcal{X}}}_{2}\subseteq%
\mathbb{R}
^{K_{2}}$, $K_{1}+K_{2}=K$, $K\geq2$, $K_{1}\geq1$ and $K_{2}\geq1$. Then, by
Fubini's theorem, the MI $I$ in (\ref{MI}) can be written as
\begin{equation}
I=\int_{{{\mathcal{X}}}_{2}}\int_{{{\mathcal{X}}}_{1}}\int_{{\mathcal{R}}%
}p(\mathbf{r}|\mathbf{x}_{1}\text{,\thinspace}\mathbf{x}_{2})p(\mathbf{x}%
_{1}\text{,\thinspace}\mathbf{x}_{2})\ln\frac{p(\mathbf{r}|\mathbf{x}%
_{1}\text{,\thinspace}\mathbf{x}_{2})}{p(\mathbf{r})}d\mathbf{r}\,
d\mathbf{x}_{1}d\mathbf{x}_{2}\text{,} \label{I12}%
\end{equation}
where $p(\mathbf{x}_{1}$\textrm{,\thinspace}$\mathbf{x}_{2})=p(\mathbf{x})$
and $p(\mathbf{r}|\mathbf{x}_{1}$\textrm{,\thinspace}$\mathbf{x}%
_{2})=p(\mathbf{r}|\mathbf{x})$.

First define
\begin{subequations}
\begin{align}
&  \mathbf{G}\left(  {\mathbf{x}}\right)  ={\left(
\begin{array}
[c]{cc}%
\mathbf{G}{_{1\text{,\thinspace}1}}\left(  {\mathbf{x}}\right)  &
\mathbf{G}{_{1\text{,\thinspace}2}}\left(  {\mathbf{x}}\right) \\
\mathbf{G}{_{2\text{,\thinspace}1}}\left(  {\mathbf{x}}\right)  &
{\mathbf{G}_{2\text{,\thinspace}2}}\left(  {\mathbf{x}}\right)
\end{array}
\right)  }\text{,}\label{3.9a}\\
&  \mathbf{G}{_{i\text{,\thinspace}j}\left(  {\mathbf{x}}\right)
=\mathbf{J}_{i\text{,\thinspace}j}\left(  {\mathbf{x}}\right)  +\mathbf{P}%
_{i\text{,\thinspace}j}}\left(  {\mathbf{x}}\right)  \text{,} \label{3.9b}%
\end{align}
where $i\mathrm{,\,}j\in\left\{  1\mathrm{,\,}2\right\}  $, and
\end{subequations}
\begin{subequations}
\begin{align}
&  {\mathbf{J}_{i\text{,\thinspace}j}\left(  {\mathbf{x}}\right)
=\left\langle \dfrac{\partial\ln p(\mathbf{r}|\mathbf{x})}{\partial
\mathbf{x}_{i}}\dfrac{\partial\ln p(\mathbf{r}|\mathbf{x})}{\partial
\mathbf{x}_{j}^{T}}\right\rangle _{\mathbf{r}|\mathbf{x}}}\text{,}%
\label{3.10a}\\
&  {\mathbf{P}_{i\text{,\thinspace}j}\left(  {\mathbf{x}}\right)
=-\dfrac{\partial^{2}\ln p(\mathbf{x})}{\partial\mathbf{x}_{i}\partial
\mathbf{x}_{j}^{T}}}\text{.} \label{3.10b}%
\end{align}
Then we have the following results and their proofs are given in Appendix.
\end{subequations}
\begin{theorem}
\label{Theorem 3}Suppose matrices $\mathbf{G}\left(  \mathbf{x}\right)  $,
$\mathbf{G}_{1\text{,\thinspace}1}\left(  {\mathbf{x}}\right)  $ and
$\mathbf{G}_{2\text{,\thinspace}2}\left(  {\mathbf{x}}\right)  $ are
positive-definite. If the matrix $\mathbf{A}_{\mathbf{x}}\in%
\mathbb{R}
^{K\times K}$ satisfies
\begin{equation}
\left\vert \mathrm{Tr}\left(  \left\langle \mathbf{A}_{\mathbf{x}%
}\right\rangle _{\mathbf{x}}\right)  \right\vert \ll1 \label{Thm3.1}%
\end{equation}
with
\begin{align}
&  \mathbf{A}_{\mathbf{x}}=\mathbf{G}_{2\text{,\thinspace}2}^{-1/2}\left(
{\mathbf{x}}\right)  \mathbf{G}_{2\text{,\thinspace}1}\left(  {\mathbf{x}%
}\right)  \mathbf{G}_{1\text{,\thinspace}1}^{-1}\left(  {\mathbf{x}}\right)
\mathbf{G}_{1\text{,\thinspace}2}\left(  {\mathbf{x}}\right)  \mathbf{G}%
_{2\text{,\thinspace}2}^{-1/2}\text{,} \label{Cx}%
\end{align}
then we have
\begin{equation}
I_{G}\simeq{I}_{G_{1}}\text{,} \label{Thm3.2}%
\end{equation}
with strict equality if and only if
\begin{equation}
\mathbf{G}_{2\text{,\thinspace}1}\left(  {\mathbf{x}}\right)  \mathbf{G}%
_{1\text{,\thinspace}1}^{-1}\left(  {\mathbf{x}}\right)  \mathbf{G}%
_{1\text{,\thinspace}2}\left(  {\mathbf{x}}\right)  =\mathbf{0}\text{,}
\label{Thm3.3}%
\end{equation}
where%
\begin{equation}
I_{G_{1}}=\dfrac{1}{2}\left\langle \ln\left(  \det\left(  \dfrac
{\mathbf{G}_{1\text{,\thinspace}1}\left(  {\mathbf{x}}\right)  }{2\pi
e}\right)  \right)  \right\rangle _{\mathbf{x}}+{\dfrac{1}{2}}\left\langle
{\ln\left(  \det\left(  \dfrac{\mathbf{G}_{2\text{,\thinspace}2}\left(
{\mathbf{x}}\right)  }{2\pi e}\right)  \right)  }\right\rangle _{\mathbf{x}%
}+{H(X)}\text{.} \label{IC}%
\end{equation}

\end{theorem}

\begin{theorem}
\label{Theorem 4}Suppose matrices $\mathbf{G}\left(  \mathbf{x}\right)  $,
$\mathbf{G}_{1\text{,\thinspace}1}\left(  {\mathbf{x}}\right)  $ and
$\mathbf{P}_{2\text{,\thinspace}2}\left(  {\mathbf{x}}\right)  $ are
positive-definite. If the matrix $\mathbf{B}_{\mathbf{x}}\in%
\mathbb{R}
^{K_{2}\times K_{2}}$ is positive-semidefinite and satisfies
\begin{equation}
0\leq{\mathrm{Tr}}\left(  \left\langle \mathbf{B}_{\mathbf{x}}\right\rangle
_{\mathbf{x}}\right)  \ll1 \label{Thm4.2}%
\end{equation}
with
\begin{align}
&  \mathbf{B}_{\mathbf{x}}=\mathbf{P}_{2\text{,\thinspace}2}^{-1/2}\left(
{\mathbf{x}}\right)  \mathbf{C}_{\mathbf{x}}\mathbf{P}_{2\text{,\thinspace}%
2}^{-1/2}\left(  {\mathbf{x}}\right)  \text{,}\label{Ax}\\
&  \mathbf{C}_{\mathbf{x}}=\mathbf{J}_{2\text{,\thinspace}2}\left(
{\mathbf{x}}\right)  -\mathbf{G}_{2\text{,\thinspace}1}\left(  {\mathbf{x}%
}\right)  \mathbf{G}_{1\text{,\thinspace}1}^{-1}\left(  {\mathbf{x}}\right)
\mathbf{G}_{1\text{,\thinspace}2}\left(  {\mathbf{x}}\right)  \text{,}
\label{Bx}%
\end{align}
then we have
\begin{equation}
I_{G}\simeq I_{G_{2}}\text{,} \label{Thm4.3}%
\end{equation}
with strict equality if and only if
\begin{equation}
\mathbf{C}_{\mathbf{x}}=\mathbf{0}\text{,} \label{Thm4.4}%
\end{equation}
where%
\begin{equation}
I_{G_{2}}={\dfrac{1}{2}}\left\langle {\ln\left(  \det\left(  \dfrac
{\mathbf{G}_{1\text{,\thinspace}1}\left(  {\mathbf{x}}\right)  }{2\pi
e}\right)  \right)  }\right\rangle _{\mathbf{x}}+{\dfrac{1}{2}}\left\langle
{\ln\left(  \det\left(  \dfrac{\mathbf{P}_{2\text{,\thinspace}2}\left(
{\mathbf{x}}\right)  }{2\pi e}\right)  \right)  }\right\rangle _{\mathbf{x}%
}+{H(X)}\text{.} \label{Id}%
\end{equation}

\end{theorem}

\begin{corollary}
\label{Corollary 3}If the random variables $X_{1}\ $and $X_{2}\ $are
independent so that $p(\mathbf{x})=p(\mathbf{x}_{1})p(\mathbf{x}_{2})$,
$p(\mathbf{x}_{2})=\mathcal{N}\left(  {\boldsymbol{\mu}}_{2}\text{,\thinspace
}\boldsymbol{\Sigma}_{\mathbf{x}_{2}}\right)  $ is a normal distribution, and
$\mathbf{G}\left(  \mathbf{x}\right)  $, $\mathbf{G}_{1\text{,\thinspace}%
1}\left(  {\mathbf{x}}\right)  $, $\mathbf{P}_{1\text{,\thinspace}1}\left(
{\mathbf{x}}\right)  $\ and $\mathbf{P}_{2\text{,\thinspace}2}\left(
{\mathbf{x}}\right)  $ are all positive-definite and satisfy (\ref{Thm4.2}),
then we have
\begin{align}
I_{G}  &  \simeq I_{G_{1}^{\prime}}\text{,}\label{Cly3.1}\\
I_{G_{1}^{\prime}}  &  =\frac{1}{2}\left\langle \ln\left(  \det\left(
\frac{\mathbf{G}_{1\text{,\thinspace}1}\left(  {\mathbf{x}}\right)  }{2\pi
e}\right)  \right)  \right\rangle _{\mathbf{x}}+H(X_{1})\text{,} \label{Id1}%
\end{align}
with strict equality if and only if
\begin{equation}
\mathbf{C}_{\mathbf{x}}=\mathbf{J}_{2\text{,\thinspace}2}\left(  {\mathbf{x}%
}\right)  -\mathbf{J}_{2\text{,\thinspace}1}\left(  {\mathbf{x}}\right)
\mathbf{G}_{1\text{,\thinspace}1}^{-1}\left(  {\mathbf{x}}\right)
\mathbf{J}_{1\text{,\thinspace}2}\left(  {\mathbf{x}}\right)  =\mathbf{0}%
\text{,} \label{Cly3.2}%
\end{equation}
where
\begin{subequations}
\begin{align}
&  H(X_{1})=-\left\langle \ln p(\mathbf{x}_{1})\right\rangle _{\mathbf{x}_{1}%
}\text{,}\label{Cly3.3a}\\
&  \mathbf{G}_{1\text{,\thinspace}1}\left(  {\mathbf{x}}\right)
=\mathbf{J}_{1\text{,\thinspace}1}\left(  {\mathbf{x}}\right)  +\mathbf{P}%
_{1\text{,\thinspace}1}\left(  {\mathbf{x}}\right)  \text{,}\label{Cly3.3b}\\
&  \mathbf{P}_{1\text{,\thinspace}1}\left(  {\mathbf{x}}\right)
=-\frac{\partial^{2}\ln p(\mathbf{x}_{1})}{\partial\mathbf{x}_{1}%
\partial\mathbf{x}_{1}^{T}}\text{.} \label{Cly3.3c}%
\end{align}

\end{subequations}
\end{corollary}

\begin{remark}
\label{Remark 4}Sometimes we are concerned only with calculating the
determinant of matrix $\mathbf{G}(\mathbf{x})$ with a given $p(\mathbf{x})$.
\textbf{Theorem \ref{Theorem 3}\textit{ }}and\textbf{\textit{ }Theorem
\ref{Theorem 4}\textit{ }}provide a dimensionality reduction method for
computing $\mathbf{G}\left(  \mathbf{x}\right)  $ or $\det\left(
\mathbf{G}\left(  \mathbf{x}\right)  \right)  $, by which we only need to
compute $\mathbf{G}_{1\text{,\thinspace}1}\left(  {\mathbf{x}}\right)  $ and
$\mathbf{G}_{2\text{,\thinspace}2}\left(  {\mathbf{x}}\right)  $ separately.
To apply the approximation (\ref{Thm3.2}), we do not need to strictly require
$\left\vert {\mathrm{Tr}}\left(  \left\langle \mathbf{A}_{\mathbf{x}%
}\right\rangle _{\mathbf{x}}\right)  \right\vert \ll1$; instead we only need
to require%
\begin{equation}
\left\vert {\mathrm{Tr}}\left(  \left\langle \mathbf{A}_{\mathbf{x}%
}\right\rangle _{\mathbf{x}}\right)  \right\vert \ll\left\vert \left\langle
\ln\left(  \det\left(  \mathbf{G}_{1\text{,\thinspace}1}\left(  {\mathbf{x}%
}\right)  \right)  \det\left(  \mathbf{G}_{2\text{,\thinspace}2}\left(
{\mathbf{x}}\right)  \right)  \right)  \right\rangle _{\mathbf{x}}\right\vert
\text{.} \label{Rem4.1}%
\end{equation}
Similarly, the inequality $\left\vert {\mathrm{Tr}}\left(  \left\langle
\mathbf{B}_{\mathbf{x}}\right\rangle _{\mathbf{x}}\right)  \right\vert \ll1$
can be substituted by
\begin{equation}
\left\vert {\mathrm{Tr}}\left(  \left\langle \mathbf{B}_{\mathbf{x}%
}\right\rangle _{\mathbf{x}}\right)  \right\vert \ll\left\vert \left\langle
\ln\left(  \det\left(  \mathbf{G}_{1\text{,\thinspace}1}\left(  {\mathbf{x}%
}\right)  \right)  \det\left(  \mathbf{P}_{2\text{,\thinspace}2}\left(
{\mathbf{x}}\right)  \right)  \right)  \right\rangle _{\mathbf{x}}\right\vert
\text{.} \label{Rem4.2}%
\end{equation}

By (\ref{Cly3.1}) and the second mean value theorem for integrals, we get%
\begin{equation}
I_{G_{1}^{\prime}}=\frac{1}{2}\left\langle \ln\left(  \det\left(
\frac{\mathbf{G}_{1\text{,\thinspace}1}\left(  \mathbf{x}_{1}\text{,\thinspace
}\mathbf{\ddot{x}}_{2}\right)  }{2\pi e}\right)  \right)  \right\rangle
_{\mathbf{x}_{1}}+H(X_{1}) \label{Id1a}%
\end{equation}
for some fixed $\mathbf{\ddot{x}}_{2}\in{{\mathcal{X}}}_{2}$. When $\left\Vert
\boldsymbol{\Sigma}_{\mathbf{x}_{2}}\right\Vert $ is small, $\mathbf{\ddot{x}%
}_{2}$ should be close to the mean: $\mathbf{\ddot{x}}_{2}\approx
{\boldsymbol{\mu}}_{2}$. It follows from {Theorem \ref{Theorem 1}} and
{Corollary \ref{Corollary 3}} that the approximate relationship $I\simeq
I_{G_{1}^{\prime}}$ holds. However, Eq. (\ref{Id1a})\ implies that
$I_{G_{1}^{\prime}}$ is determined only by the first component $\mathbf{x}%
_{1}$. Hence, there is little impact on information transfer by the minor
component (i.e. $\mathbf{x}_{2}$) for the high-dimensional input $\mathbf{x}$.
In other words, the information transfer is mainly determined\ by the first
component $\mathbf{x}_{1}$ and we can omit the minor component $\mathbf{x}%
_{2}$.
\qed

\end{remark}

\subsection{Further Discussion}

\label{Sec:3.2}Suppose $\mathbf{x}$ is a zero-mean vector, and if it is not,
then let $\mathbf{x}\leftarrow\mathbf{x}-\left\langle \mathbf{x}\right\rangle
_{\mathbf{x}}\text{.}
$ The covariance matrix of $\mathbf{x}$ is given by
\begin{equation}
\boldsymbol{\Sigma}_{\mathbf{x}}=\left\langle \mathbf{xx}^{T}\right\rangle
_{\mathbf{x}}=\mathbf{U}\boldsymbol{\Sigma}\mathbf{U}^{T}\text{,} \label{Sigx}%
\end{equation}
where $\mathbf{U}$ is a $K\times K$ orthogonal matrix whose \textit{k}-th
column is the eigenvector $\mathbf{u}_{k}$ of $\boldsymbol{\Sigma}%
$\textbf{$_{\mathbf{x}}$}, and $\boldsymbol{\Sigma}$ is diagonal matrix whose
diagonal elements are the corresponding eigenvalues, i.e., $\boldsymbol{\Sigma
} =\mathrm{diag}\left(  \sigma_{1}^{2}\text{,\thinspace}\sigma_{2}%
^{2}\text{,\thinspace}\cdots\text{,\thinspace}\sigma_{K}^{2}\right)  $
with $\sigma_{1} \geq\sigma_{2}\geq\cdots\geq\sigma_{K}>0$.
With the whitening transformation,
\begin{equation}
\mathbf{\tilde{x}}=\boldsymbol{\Sigma}^{-1/2}\mathbf{U}^{T}\mathbf{x}%
\text{,\thinspace} \label{xtilde}%
\end{equation}
the covariance matrix of $\mathbf{\tilde{x}}$ becomes an identity matrix:
\begin{equation}
\boldsymbol{\Sigma}_{\mathbf{\tilde{x}}}=\left\langle \mathbf{\tilde{x}%
\tilde{x}}^{T}\right\rangle _{\mathbf{\tilde{x}}}=\boldsymbol{\Sigma}%
^{-1/2}\mathbf{U}^{T}\left\langle \mathbf{xx}^{T}\right\rangle _{\mathbf{x}%
}\mathbf{U}\boldsymbol{\Sigma}^{-1/2}=\mathbf{I}_{K}\text{.} \label{Sigxtilde}%
\end{equation}

By the central limit theorem, the distribution of random variable $\tilde{X}$
should be closer to a normal distribution than the distribution of the
original random variable $X$; that is, $p(\mathbf{\tilde{x}})\simeq
{\mathcal{N}\left(  \mathbf{0},\mathbf{I}_{K}\right)  }$. Using Laplace's
method asymptotic expansion \citep{MacKay(2003-BK-information)}, we get
\begin{align}
\mathbf{P}(\mathbf{\tilde{x}})  &  =-\dfrac{\partial^{2}\ln p(\mathbf{\tilde
{x}})}{\partial\mathbf{\tilde{x}}\partial\mathbf{\tilde{x}}^{T}}%
\simeq\boldsymbol{\Sigma}_{\mathbf{\tilde{x}}}^{-1}=\mathbf{I}_{K}%
\text{,}\label{3.14}\\
\mathbf{P}_{+}  &  =\left\langle \mathbf{P}(\mathbf{\tilde{x}})\right\rangle
_{\mathbf{\tilde{x}}}\simeq\boldsymbol{\Sigma}_{\mathbf{\tilde{x}}}%
^{-1}=\mathbf{I}_{K}\text{.} \label{3.14b}%
\end{align}
In principal component analysis (PCA), the dataset is modeled by a
multivariate gaussian. By a PCA-like whitening transformation (\ref{xtilde})
we can use the approximation (\ref{3.14b}) with Laplace's method, which only
requires that the peak be close to the mean and the random variable $\tilde
{X}$\ does not need to be an exact Gaussian distribution.

By \textbf{Theorem \ref{Theorem 2}}, we have
\begin{equation}
I\left(  \tilde{X}{\text{;\thinspace}}R\right)  \simeq I_{G}=\frac{1}%
{2}\left\langle \ln\left(  \det\left(  \frac{\mathbf{G}(\mathbf{\tilde{x}}%
)}{2\pi e}\right)  \right)  \right\rangle _{\mathbf{\tilde{x}}}+H(\tilde
{X})\text{,} \label{3.15}%
\end{equation}
where
\begin{align}
\mathbf{G}{(}\mathbf{\tilde{x}}{)}  &  =\mathbf{J}{(}\mathbf{\tilde{x}}%
{)}+\mathbf{I}_{K}\text{,}\label{3.16}\\
\mathbf{J}{(}\mathbf{\tilde{x}}{)}  &  ={\left\langle \dfrac{\partial\ln
p(\mathbf{r}|\mathbf{\tilde{x}})}{\partial\mathbf{\tilde{x}}}\dfrac
{\partial\ln p(\mathbf{r}|\mathbf{\tilde{x}})}{\partial\mathbf{\tilde{x}}^{T}%
}\right\rangle _{\mathbf{r}|\mathbf{\tilde{x}}}}\label{3.16b}\\
&  =\boldsymbol{\Sigma}{^{1/2}\mathbf{U}^{T}\left\langle \dfrac{\partial\ln
p(\mathbf{r}|\mathbf{x})}{\partial\mathbf{x}}\dfrac{\partial\ln p(\mathbf{r}%
|\mathbf{x})}{\partial\mathbf{x}^{T}}\right\rangle _{\mathbf{r}|\mathbf{x}%
}\mathbf{U}}\boldsymbol{\Sigma}{^{1/2}}\label{3.16c}\\
&  =\boldsymbol{\Sigma}{^{1/2}\mathbf{U}^{T}}\mathbf{J}{(\mathbf{x}%
)\mathbf{U}}\boldsymbol{\Sigma}{^{1/2}}\text{,}\label{3.16d}\\
{H(\tilde{X})}  &  ={-}\left\langle \ln p(\mathbf{\tilde{x}})\right\rangle
_{\mathbf{\tilde{x}}}={H(X)}-{\dfrac{1}{2}\ln\left(  \det(\boldsymbol{\Sigma
})\right)  }\text{.} \label{3.16e}%
\end{align}

Given a $K\times K$ orthogonal matrix $\mathbf{B}\in%
\mathbb{R}
^{K\times K}$, {\normalsize we} define
\begin{equation}
\mathbf{y}=\mathbf{B}^{T}\mathbf{\tilde{x}}\text{.} \label{3.17}%
\end{equation}
Then it follows from (\ref{3.15})--(\ref{3.17}) that
\begin{equation}
I\left(  Y{\text{;\thinspace}}R\right)  \simeq I_{G}=\frac{1}{2}\left\langle
\ln\left(  \det\left(  \frac{\mathbf{G}(\mathbf{y})}{2\pi e}\right)  \right)
\right\rangle _{\mathbf{y}}+H\left(  Y\right)  \text{,} \label{3.18}%
\end{equation}
where
\begin{align}
\mathbf{G}{(\mathbf{y})}  &  =\mathbf{J}{(\mathbf{y})}+\mathbf{I}_{K}%
\text{,}\label{3.19}\\
\mathbf{J}{(\mathbf{y})}  &  ={\mathbf{B}^{T}}\mathbf{J}{(}\mathbf{\tilde{x}%
}{)}\mathbf{B}\text{,}\label{3.19b}\\
{H(}Y{)}  &  ={H(\tilde{X})}\text{.} \label{3.19c}%
\end{align}
Suppose $\mathbf{y}$ is partitioned into two sets of components,
$\mathbf{y}=(\mathbf{y}_{1}^{T}$\textrm{,\thinspace}$\mathbf{y}_{2}^{T})^{T}$
and
\begin{align}
{\mathbf{y}_{1}}  &  ={(}y_{1}\text{,\thinspace}y_{2}\text{,\thinspace}%
\cdots\text{,\thinspace}y_{K_{1}}{)^{T}}\text{,}\label{3.20}\\
{\mathbf{y}_{2}}  &  ={(}y_{K_{1}+1}\text{,\thinspace}y_{K_{1}+2}%
\text{,\thinspace}\cdots\text{,\thinspace}y_{K}{)^{T}}\text{,} \label{3.20b}%
\end{align}
where $K_{1}+K_{2}=K$, $K\geq2$, $K_{1}\geq1$ and $K_{2}\geq1$. Let
\begin{equation}
\mathbf{G}(\mathbf{y})=\left(
\begin{array}
[c]{cc}%
\mathbf{J}{_{1\text{,\thinspace}1}(\mathbf{y})}+\mathbf{I}_{K_{1}} &
\mathbf{J}{_{1\text{,\thinspace}2}(\mathbf{y})}\\
\mathbf{J}{_{2\text{,\thinspace}1}(\mathbf{y})} & \mathbf{J}%
{_{2\text{,\thinspace}2}(\mathbf{y})}+\mathbf{I}_{K_{2}}%
\end{array}
\right)  \text{,} \label{3.21}%
\end{equation}
where
\begin{equation}
\mathbf{J}_{i\text{,\thinspace}j}(\mathbf{y})=\left\langle \frac{\partial\ln
p(\mathbf{r}|\mathbf{y})}{\partial\mathbf{y}_{i}}\frac{\partial\ln
p(\mathbf{r}|\mathbf{y})}{\partial\mathbf{y}_{j}^{T}}\right\rangle
_{\mathbf{r}|\mathbf{y}}\text{, } \quad\forall i\mathrm{,\,}j=1\mathrm{,\,}%
2\text{.} \label{3.22}%
\end{equation}

When $K\gg1$, suppose we can find an orthogonal matrix $\mathbf{B}$ and
$K_{1}$ that satisfy the condition (\ref{Thm4.2}) in \textbf{Theorem
\ref{Theorem 4} }or\textbf{ }condition (\ref{Rem4.2}), i.e.
\begin{align}
0  &  \leq\left\langle {\mathrm{Tr}}\left(  \mathbf{B}_{\mathbf{y}}\right)
\right\rangle _{\mathbf{y}}\ll\gamma\text{,}\label{3.23}\\
\mathbf{B}_{\mathbf{y}}  &  =\mathbf{J}_{2\text{,\thinspace}2}(\mathbf{y}%
)-\mathbf{J}_{2\text{,\thinspace}1}(\mathbf{y})\left(  \mathbf{J}%
_{1\text{,\thinspace}1}(\mathbf{y})+\mathbf{I}_{K_{1}}\right)  ^{-1}%
\mathbf{J}_{1\text{,\thinspace}2}(\mathbf{y})\text{,}\label{3.23b}\\
\gamma &  =\left\langle \ln\left(  \det\left(  \mathbf{J}{_{1\text{,\thinspace
}1}(\mathbf{y})}+\mathbf{I}_{K_{1}}\right)  \right)  \right\rangle
_{\mathbf{y}}. \label{3.23c}%
\end{align}
Here matrix $\mathbf{B}_{\mathbf{y}}$ is positive-semidefinite because
\begin{equation}
\mathbf{J}_{2\text{,\thinspace}2}(\mathbf{y})-\mathbf{J}_{2\text{,\thinspace
}1}(\mathbf{y})\left(  \mathbf{J}_{1\text{,\thinspace}1}(\mathbf{y}%
)+\mathbf{I}_{K_{1}}\right)  ^{-1}\mathbf{J}_{1\text{,\thinspace}2}%
(\mathbf{y})=\left\langle \boldsymbol{\rho}\mathbf{(\mathbf{r}|\mathbf{y}%
)}\boldsymbol{\rho}(\mathbf{r}|\mathbf{y})^{T}\right\rangle _{\mathbf{r}%
|\mathbf{y}}\text{,} \label{3.24}%
\end{equation}
where
\begin{equation}
\boldsymbol{\rho}(\mathbf{r}|\mathbf{y})=\dfrac{\partial\ln p(\mathbf{r}%
|\mathbf{y})}{\partial\mathbf{y}_{2}}-\mathbf{J}_{2\text{,\thinspace}%
1}(\mathbf{y})\left(  \mathbf{J}_{1\text{,\thinspace}1}(\mathbf{y}%
)+\mathbf{I}_{K_{1}}\right)  ^{-1}\left(  \dfrac{\partial\ln p(\mathbf{r}%
|\mathbf{y})}{\partial\mathbf{y}_{1}}+\mathbf{a}\left(  \mathbf{r}\right)
\right)  \label{3.25}%
\end{equation}
and $\mathbf{a}\left(  \mathbf{r}\right)  $ is a $K_{1}$-dimensional random
vector that satisfies%
\begin{align}
\left\langle \dfrac{\partial\ln p(\mathbf{r}|\mathbf{y})}{\partial
\mathbf{y}_{2}}\mathbf{a}\left(  \mathbf{r}\right)  ^{T}\right\rangle
_{\mathbf{r}|\mathbf{y}}  &  =\left\langle \dfrac{\partial\ln p(\mathbf{r}%
|\mathbf{y})}{\partial\mathbf{y}_{2}}\right\rangle _{\mathbf{r}|\mathbf{y}%
}\left\langle \mathbf{a}\left(  \mathbf{r}\right)  ^{T}\right\rangle
_{\mathbf{r}|\mathbf{y}}=\mathbf{0}\text{,}\label{3.25A}\\
\left\langle \mathbf{a}\left(  \mathbf{r}\right)  \mathbf{a}\left(
\mathbf{r}\right)  ^{T}\right\rangle _{\mathbf{r}|\mathbf{y}}  &
=\mathbf{I}_{K_{1}}\text{.} \label{3.25B}%
\end{align}

Assuming that $\mathbf{J}_{1\text{,\thinspace}1}(\mathbf{y})$ is
positive-definite, $\left\Vert \mathbf{J}_{1\text{,\thinspace}1}%
^{-1}(\mathbf{y})\right\Vert =O\left(  N^{-1}\right)  $ and $\left\Vert
\mathbf{J}_{1\text{,\thinspace}2}(\mathbf{y})\right\Vert =\left\Vert
\mathbf{J}_{2\text{,\thinspace}1}(\mathbf{y})\right\Vert =O\left(  N\right)
$, we have
\begin{equation}
\left(  \mathbf{J}_{1\text{,\thinspace}1}(\mathbf{y})+\mathbf{I}_{K_{1}%
}\right)  ^{-1}=\mathbf{J}_{1\text{,\thinspace}1}^{-1}(\mathbf{y}%
)-\mathbf{J}_{1\text{,\thinspace}1}^{-2}(\mathbf{y})+O\left(  \mathbf{J}%
_{1\text{,\thinspace}1}^{-3}(\mathbf{y})\right)  \label{3.25a}%
\end{equation}
and%
\begin{align}
\mathrm{Tr}\left(  \mathbf{C}_{\mathbf{x}}\right)   &  =\mathrm{Tr}\left(
\mathbf{J}_{2\text{,\thinspace}2}(\mathbf{y})-\mathbf{J}_{2\text{,\thinspace
}1}(\mathbf{y})\mathbf{J}_{1\text{,\thinspace}1}^{-1}(\mathbf{y}%
)\mathbf{J}_{1\text{,\thinspace}2}(\mathbf{y})\right) \nonumber\\
&  +\mathrm{Tr}\left(  \mathbf{J}_{2\text{,\thinspace}1}(\mathbf{y}%
)\mathbf{J}_{1\text{,\thinspace}1}^{-2}(\mathbf{y})\mathbf{J}%
_{1\text{,\thinspace}2}(\mathbf{y})\right)  +O\left(  N^{-1}\right)  \text{.}
\label{3.25b}%
\end{align}
Hence, if
\begin{align}
\left\vert \mathrm{Tr}\left(  \mathbf{J}_{2\text{,\thinspace}2}(\mathbf{y}%
)-\mathbf{J}_{2\text{,\thinspace}1}(\mathbf{y})\mathbf{J}_{1\text{,\thinspace
}1}^{-1}(\mathbf{y})\mathbf{J}_{1\text{,\thinspace}2}(\mathbf{y})\right)
\right\vert  &  \ll\gamma\text{,}\label{3.25c}\\
\left\vert \mathrm{Tr}\left(  \mathbf{J}_{2\text{,\thinspace}1}(\mathbf{y}%
)\mathbf{J}_{1\text{,\thinspace}1}^{-2}(\mathbf{y})\mathbf{J}%
_{1\text{,\thinspace}2}(\mathbf{y})\right)  \right\vert  &  \ll\gamma\text{,}
\label{3.25d}%
\end{align}
then (\ref{3.23}) holds. Notice that the matrix $\left(  \mathbf{J}%
_{2\text{,\thinspace}2}(\mathbf{y})-\mathbf{J}_{2\text{,\thinspace}%
1}(\mathbf{y})\mathbf{J}_{1\text{,\thinspace}1}^{-1}(\mathbf{y})\mathbf{J}%
_{1\text{,\thinspace}2}(\mathbf{y})\right)  $ is positive-semidefinite which
is similar to (\ref{3.24}) and $0\leq\mathrm{Tr}\left(  \mathbf{J}%
_{2\text{,\thinspace}1}(\mathbf{y})\mathbf{J}_{1\text{,\thinspace}1}%
^{-1}(\mathbf{y})\mathbf{J}_{1\text{,\thinspace}2}(\mathbf{y})\right)
\leq\mathrm{Tr}\left(  \mathbf{J}_{2\text{,\thinspace}2}(\mathbf{y})\right)
$. Hence, if
\begin{equation}
\mathrm{Tr}\left(  \mathbf{J}_{2\text{,\thinspace}2}(\mathbf{y})\right)
\ll\gamma\text{,} \label{3.26}%
\end{equation}
then (\ref{3.25c}) and (\ref{3.25d}) hold and (\ref{3.23}) holds.

\section{Optimization of Information Transfer in Neural Population Coding}

\label{Sec:4}

\subsection{Population Density Distribution of Parameters in Neural
Populations}

\label{Sec:4.1.2}If $p(\mathbf{r}|\mathbf{x})$ is conditional independent, we
can write
\begin{equation}
p(\mathbf{r}|\mathbf{x})=\prod_{n=1}^{N}p(r_{n}|\mathbf{x}\text{;\thinspace
}\boldsymbol{\theta}_{n})\text{,} \label{4.1}%
\end{equation}
where $\boldsymbol{\theta}_{n}\in%
\mathbb{R}
^{\tilde{K}}$ denotes a $\tilde{K}$-dimensional vector for parameters of the
\textit{n}-th neuron, and $p(r_{n}|\mathbf{x}$;\thinspace$\boldsymbol{\theta
}_{n})$ is the conditional p.d.f.\ of the output $r_{n}$ given $\mathbf{x}$.
With the definition in (\ref{Jx}), we have following proposition.

\begin{proposition}
\label{Proposition 1}If $p(\mathbf{r}|\mathbf{x})$ is conditional independent
as in Eq. (\ref{4.1}), we have
\begin{equation}
\mathbf{J}(\mathbf{x})=N\int_{{\Theta}}p({\boldsymbol{\theta}})\mathbf{S}%
(\mathbf{x}\text{;\thinspace}{\boldsymbol{\theta}})d{\boldsymbol{\theta}%
}\text{,} \label{Ppn1.1}%
\end{equation}
where
\begin{equation}
\mathbf{S}(\mathbf{x}\text{;\thinspace}\boldsymbol{\theta})=\int%
_{\mathfrak{R}}p(r|\mathbf{x}\text{;\thinspace}\boldsymbol{\theta}%
)\frac{\partial\ln p(r|\mathbf{x}\text{;\thinspace}\boldsymbol{\theta}%
)}{\partial\mathbf{x}}\frac{\partial\ln p(r|\mathbf{x}\text{;\thinspace
}\boldsymbol{\theta})}{\partial\mathbf{x}^{T}}dr\text{,} \label{Ppn1.2}%
\end{equation}
$r\in\mathfrak{R}\subseteq%
\mathbb{R}
$, $\boldsymbol{\theta}\in\Theta\subseteq%
\mathbb{R}
^{\tilde{K}}$, and $p(\boldsymbol{\theta})$ is the population density function
of parameter vector ${\boldsymbol{\theta}}$:
\begin{equation}
p(\boldsymbol{\theta})=\frac{1}{N}\sum_{n=1}^{N}\delta(\boldsymbol{\theta
}-\boldsymbol{\theta}_{n})\text{,} \label{Ppn1.3}%
\end{equation}
with $\delta(\cdot)$\ being the Dirac delta function.
\end{proposition}

\begin{proof}%
\begin{align}
\mathbf{J}{(\mathbf{x})}  &  ={\int_{\mathcal{R}}p(\mathbf{r}|\mathbf{x}%
)\dfrac{\partial\ln p(\mathbf{r}|\mathbf{x})}{\partial\mathbf{x}}%
\dfrac{\partial\ln p(\mathbf{r}|\mathbf{x})}{\partial\mathbf{x}^{T}}%
d}\mathbf{r}\nonumber\\
&  ={\sum_{n=1}^{N}\int_{\mathfrak{R}}p(}r{_{n}|\mathbf{x}\text{;\thinspace
}{\boldsymbol{\theta}}_{n})\dfrac{\partial\ln p(r_{n}|\mathbf{x}%
\text{;\thinspace}{\boldsymbol{\theta}}_{n})}{\partial\mathbf{x}}%
\dfrac{\partial\ln p(r_{n}|\mathbf{x}\text{;\thinspace}{\boldsymbol{\theta}%
}_{n})}{\partial\mathbf{x}^{T}}dr_{n}}\nonumber\\
&  ={\int_{\Theta}\sum_{n=1}^{N}\delta({\boldsymbol{\theta}}%
-{\boldsymbol{\theta}}_{n})\left(  \int_{\mathfrak{R}}p(r|\mathbf{x}%
\text{;\thinspace}{\boldsymbol{\theta}})\dfrac{\partial\ln p(r|\mathbf{x}%
\text{;\thinspace}{\boldsymbol{\theta}})}{\partial\mathbf{x}}\dfrac
{\partial\ln p(r|\mathbf{x}\text{;\thinspace}{\boldsymbol{\theta}})}%
{\partial\mathbf{x}^{T}}dr\right)  d}{\boldsymbol{\theta}}\nonumber\\
&  =N{\int_{{\Theta}}p({\boldsymbol{\theta}})\mathbf{S}(\mathbf{x}%
\text{;\thinspace}{\boldsymbol{\theta}})d}{\boldsymbol{\theta}}\text{.}
\label{Ppn1.4}%
\end{align}

\end{proof}

\begin{remark}
\label{Remark 5.1} \textbf{Proposition \ref{Proposition 1}} shows that
$\mathbf{J}(\mathbf{x})$ can be regarded as a function of the population
density of parameters, $p({\boldsymbol{\theta}})$. If the p.d.f.\ of the input
$p(\mathbf{x})$ is given, we can find an appropriate $p({\boldsymbol{\theta}%
})$ to maximize MI $I$. \qed

\end{remark}

For neuron model with Poisson spikes, we have
\begin{align}
&  {p(\mathbf{r}|\mathbf{x})=\prod_{n=1}^{N}p(}r{_{n}|\mathbf{x}%
\text{;\thinspace}{\boldsymbol{\theta}}_{n})}\text{,}\label{PoissNeuron.a}\\
&  {p(r_{n}|\mathbf{x}\text{;\thinspace}{\boldsymbol{\theta}}_{n}%
)=\dfrac{f(\mathbf{x}\text{;\thinspace}{\boldsymbol{\theta}}_{n})^{r_{n}}%
}{r_{n}!}\exp\left(  -f(\mathbf{x}\text{;\thinspace}{\boldsymbol{\theta}}%
_{n})\right)  }\text{,} \label{PoissNeuron}%
\end{align}
where $f(\mathbf{x}$;\thinspace${\boldsymbol{\theta}}_{n})$ is the tuning
curve of the \textit{n-}th neuron, $n=1$,\thinspace$2$,\thinspace$\cdots
$,\thinspace$N$. Now we have
\begin{align}
{\mathbf{S}(\mathbf{x}\text{;\thinspace}{\boldsymbol{\theta}})}  &
={\int_{\mathfrak{R}}p(}r{|\mathbf{x}\text{;\thinspace}{\boldsymbol{\theta}%
})\dfrac{\partial\ln p(r|\mathbf{x}\text{;\thinspace}{\boldsymbol{\theta}}%
)}{\partial\mathbf{x}}\dfrac{\partial\ln p(r|\mathbf{x}\text{;\thinspace
}{\boldsymbol{\theta}})}{\partial\mathbf{x}^{T}}d}r\nonumber\\
&  ={\dfrac{1}{f(\mathbf{x}\text{;\thinspace}{\boldsymbol{\theta}})}%
\dfrac{\partial f(\mathbf{x}\text{;\thinspace}{\boldsymbol{\theta}})}%
{\partial\mathbf{x}}\dfrac{\partial f(\mathbf{x}\text{;\thinspace
}{\boldsymbol{\theta}})}{\partial\mathbf{x}^{T}}}\nonumber\\
&  ={\dfrac{\partial g(\mathbf{x}\text{;\thinspace}{\boldsymbol{\theta}}%
)}{\partial\mathbf{x}}\dfrac{\partial g(\mathbf{x}\text{;\thinspace
}{\boldsymbol{\theta}})}{\partial\mathbf{x}^{T}}}\text{,}\label{PoissNeuron.1}%
\\
{g(\mathbf{x}\text{;\thinspace}{\boldsymbol{\theta}})}  &  =2\sqrt
{f(\mathbf{x}\text{;\thinspace}{\boldsymbol{\theta}})}\text{.}
\label{PoissNeuron.2}%
\end{align}
Similarly, for neuron response model with Gaussian noise, we have
\begin{align}
&  {p(\mathbf{r}|\mathbf{x})=\prod_{n=1}^{N}p(}r{_{n}|\mathbf{x}%
\text{;\thinspace}{\boldsymbol{\theta}}_{n})}\text{,}\label{GaussNeuron.a}\\
&  {p(r_{n}|\mathbf{x}\text{;\thinspace}{\boldsymbol{\theta}}_{n})=\dfrac
{1}{\sigma\sqrt{2\pi}}\exp\left(  -\dfrac{\left(  r_{n}-f(\mathbf{x}%
\text{;\thinspace}{\boldsymbol{\theta}}_{n})\right)  ^{2}}{2\sigma^{2}%
}\right)  }\text{,} \label{GaussNeuron}%
\end{align}
where $\sigma$ is a constant standard deviation. Now we get
\begin{equation}
{\mathbf{S}(\mathbf{x}\text{;\thinspace}{\boldsymbol{\theta}})}={\dfrac
{1}{\sigma^{2}}\dfrac{\partial f(\mathbf{x}\text{;\thinspace}%
{\boldsymbol{\theta}})}{\partial\mathbf{x}}\dfrac{\partial f(\mathbf{x}%
\text{;\thinspace}{\boldsymbol{\theta}})}{\partial\mathbf{x}^{T}}}\text{.}
\label{GaussNeuron.1}%
\end{equation}

\subsection{Optimal Population Distribution for Neural Population Coding}

\label{Sec:4.1}Suppose $p(\mathbf{x})$ and $p(\mathbf{r}|\mathbf{x})$ fulfill
conditions\textbf{\ C1} and \textbf{C2 }and Eq.\textbf{ (\ref{4.1})}.
Following the discussion of section \ref{Sec:2.2}, we define the following
objective for maximizing MI $I$,
\begin{equation}
\text{\textsf{maximize}}\mathrm{\;\;}I_{G}[p({\boldsymbol{\theta}})]=\frac
{1}{2}\left\langle \ln\left(  \det\left(  \frac{\mathbf{G}(\mathbf{x})}{2\pi
e}\right)  \right)  \right\rangle _{\mathbf{x}}+H(X)\text{,} \label{maxIa}%
\end{equation}
or equivalently,
\begin{equation}
\text{\textsf{minimize}}\mathrm{\;\;}Q_{G}[p({\boldsymbol{\theta}})]=-\frac
{1}{2}\left\langle \ln\left(  \det\left(  \mathbf{G}(\mathbf{x})\right)
\right)  \right\rangle _{\mathbf{x}}\text{,} \label{maxQa}%
\end{equation}
where
\begin{align}
&  \mathbf{G}{(\mathbf{x})}=\mathbf{J}{(\mathbf{x})+\mathbf{P}}\left(
\mathbf{x}\right)  \text{,}\label{GJSP}\\
&  \mathbf{J}{(\mathbf{x})=N\int_{{{\Theta}}}p({\boldsymbol{\theta}%
})\mathbf{S}(\mathbf{x}\text{;\thinspace}{\boldsymbol{\theta}}%
)d{\boldsymbol{\theta}}}\text{,}\label{GJSPb}\\
&  {\mathbf{S}(\mathbf{x}\text{;\thinspace}{\boldsymbol{\theta}}%
)}=\left\langle {\dfrac{\partial\ln p(r|\mathbf{x}\text{;\thinspace
}{\boldsymbol{\theta}})}{\partial\mathbf{x}}\dfrac{\partial\ln p(r|\mathbf{x}%
\text{;\thinspace}{\boldsymbol{\theta}})}{\partial\mathbf{x}^{T}}%
}\right\rangle _{r{|\mathbf{x}\text{;\thinspace}{\boldsymbol{\theta}}}%
}\text{.} \label{GJSPc}%
\end{align}
Here ${\mathbf{P}}\left(  \mathbf{x}\right)  $ is given in (\ref{Px}) and it
generally can be substituted by ${\mathbf{P}}_{+}$ (see Eq. \ref{PN}).

When $\varsigma_{1}\approx0$ (see Eq. \ref{Thm1a.3a}), the object function
(\ref{maxIa}) can be reduced to
\begin{equation}
\text{\textsf{maximize}}\mathrm{\;\;}I_{F}[p({\boldsymbol{\theta}})]=\frac
{1}{2}\left\langle \ln\left(  \det\left(  \frac{\mathbf{J}(\mathbf{x})}{2\pi
e}\right)  \right)  \right\rangle _{\mathbf{x}}+H(X)\text{,} \label{maxIb}%
\end{equation}
or equivalently,
\begin{equation}
\text{\textsf{minimize}}\mathrm{\;\;}Q_{F}[p({\boldsymbol{\theta}})]=-\frac
{1}{2}\left\langle \ln\left(  \det\left(  \mathbf{J}(\mathbf{x})\right)
\right)  \right\rangle _{\mathbf{x}}\text{.} \label{maxQb}%
\end{equation}
The constraint condition for $p({\boldsymbol{\theta}})$ is given by
\begin{equation}
\text{\textsf{subject\ to}}\mathrm{\;\;}{\int_{{{\Theta}}}%
p({\boldsymbol{\theta}})d{\boldsymbol{\theta}}=1}\text{, } \quad
{p({\boldsymbol{\theta}})\geq0}\text{.} \label{conditionPtheta}%
\end{equation}

However, without further constraints on the neural populations, especially a
limit on the peak firing rate, the capacity of the system may grow
indefinitely, i.e. $I(X$;\thinspace$R)\rightarrow\infty$. The most common
limitation on neural populations is the energy or power constraint. For neuron
models with Poisson noise or Gaussian noise, a useful constraint is a
limitation on the peak power,%
\begin{equation}
\left\vert f(\mathbf{x}\text{;\thinspace}{\boldsymbol{\theta}}_{n})\right\vert
\leq E_{\max}\text{, } \quad\forall\mathbf{x}\in{\mathcal{X}} \quad\mbox{and}
\quad\forall n=1\text{,\thinspace}2\text{,\thinspace}\cdots\text{,\thinspace
}N\text{.} \label{Emax}%
\end{equation}
where $E_{\max}>0$ is the peak power. Under this constraint, maximizing
$I_{G}[p({\boldsymbol{\theta}})]$ or $I_{F}[p({\boldsymbol{\theta}})]$ for
independent neurons will result in $\max_{\mathbf{x}}\left\vert f(\mathbf{x}%
\text{;\thinspace}{\boldsymbol{\theta}}_{n})\right\vert =E_{\max}$ for
$\forall n=1$,\thinspace$2$,\thinspace$\cdots$,\thinspace$N$.

Another constraint is a limitation on average power. For Poisson neurons given
in Eq. (\ref{PoissNeuron}),
\begin{equation}
\frac{1}{N}\sum_{n=1}^{N}\left\langle \left\langle r_{n}p(r_{n}|\mathbf{x}%
\text{;\thinspace}{\boldsymbol{\theta}}_{n})\right\rangle _{r_{n}|\mathbf{x}%
}\right\rangle _{\mathbf{x}}\leq E_{\mathrm{avg}}\text{,} \label{Constrain.1}%
\end{equation}
which can also be written as
\begin{equation}
\left\langle \left\langle f(\mathbf{x}\text{;\thinspace}\boldsymbol{\theta
})\right\rangle _{\mathbf{x}}\right\rangle _{\boldsymbol{\theta}}\leq
E_{\mathrm{avg}}\text{,} \label{Constrain.2}%
\end{equation}
and for Gaussian noise neurons given in Eq. (\ref{GaussNeuron}),
\begin{equation}
\left\langle \left\langle f(\mathbf{x}\text{;\thinspace}\boldsymbol{\theta
})^{2}\right\rangle _{\mathbf{x}}\right\rangle _{\boldsymbol{\theta}}\leq
E_{\mathrm{avg}}\text{,} \label{Constrain.3}%
\end{equation}
where $E_{\mathrm{avg}}>0$ is the maximum average energy cost.

In Eq. (\ref{GJSP}), we can approximate the continuous integral by a discrete
summation for numerical computation,%
\begin{equation}
\mathbf{J}(\mathbf{x})=N\sum_{k=1}^{K_{1}}\alpha_{k}\mathbf{S}(\mathbf{x}%
\text{;\thinspace}{\boldsymbol{\theta}}_{k})\text{,} \label{Ja}%
\end{equation}
where the positive integer $K_{1}\leq N$ denotes the number of subclasses in
the neural population, and
\begin{equation}
\sum_{k=1}^{K_{1}}\alpha_{k}=1\text{, } \quad\alpha_{k}>0\text{, }\forall
k=1\text{,\thinspace}2\text{,\thinspace}\cdots\text{,\thinspace}K_{1}\text{.}
\label{alfa}%
\end{equation}

If we do not know the specific form of $p(\mathbf{x})$ but have $M$ samples,
$\mathbf{x}_{1}$, $\mathbf{x}_{2}$, $\cdots$, $\mathbf{x}_{M}$, which are
i.i.d.\ samples drawn from the distribution $p(\mathbf{x})$, then we can
approximate the integral in (\ref{maxIa}) by the sample average:
\begin{equation}
\left\langle \ln\left(  \det\left(  \mathbf{G}(\mathbf{x})\right)  \right)
\right\rangle _{\mathbf{x}}\simeq\frac{1}{M}\sum_{m=1}^{M}\ln\left(
\det\left(  \mathbf{G}(\mathbf{x}_{m})\right)  \right)  \text{.} \label{GAppx}%
\end{equation}

Optimizing the objective (\ref{maxIa}) or (\ref{maxIb}) is a convex
optimization problem (see Appendix for a proof).

\begin{proposition}
\label{Proposition 2}The functions $I_{G}[p({\boldsymbol{\theta}})]$ and
$I_{F}[p({\boldsymbol{\theta}})]$ are concave about $p({\boldsymbol{\theta}})$.
\end{proposition}

\begin{remark}
\label{Remark 5.2}For a low-dimensional input $\mathbf{x}$, we may use
(\ref{maxIb}) or (\ref{maxQb}) as the objective. Since $I_{G}%
[p({\boldsymbol{\theta}})]$ and $I_{F}[p({\boldsymbol{\theta}})]$ are concave
functions of $p({\boldsymbol{\theta}})$, we can directly use efficient
numerical methods to get the optimal solution for small $K$. However, for
high-dimensional input $\mathbf{x}$, we need to use other methods
\citep[e.g. ][]{Huang(2017-IC-information)}.\qed

\end{remark}

\subsection{Necessary and Sufficient Conditions for Optimal Population
Distribution}

\label{Sec:4.1.1}Applying the method of Lagrange multipliers for the
optimization problem (\ref{maxIa}) and (\ref{conditionPtheta}) yields
\begin{equation}
L[p({\boldsymbol{\theta}})]=I_{G}[p({\boldsymbol{\theta}})]-\lambda_{1}\left(
\int_{{\Theta}}p({\boldsymbol{\theta}})d{\boldsymbol{\theta}}-1\right)
+\int_{{\Theta}}\lambda_{2}({\boldsymbol{\theta}})p({\boldsymbol{\theta}%
})d{\boldsymbol{\theta}}\text{,} \label{LPtheta}%
\end{equation}
where $\lambda_{1}$ is a constant and $\lambda_{2}({\boldsymbol{\theta}})$ is
a function of ${\boldsymbol{\theta}}$. According to Karush-Kuhn-Tucker (KKT)
conditions \citep{Boyd(2004-BK-convex)}, we have
\begin{equation}
{\lambda_{2}({\boldsymbol{\theta}})p({\boldsymbol{\theta}})=0}\text{, }%
\quad{\lambda_{2}({\boldsymbol{\theta}})\geq0}\text{,} \label{LPtheta.1}%
\end{equation}
and the necessary condition for optimal population density,
\begin{equation}
\frac{\partial L[p({\boldsymbol{\theta}})]}{\partial p({\boldsymbol{\theta}}%
)}=\frac{1}{2}\left\langle \mathrm{Tr}\left(  N\mathbf{G}(\mathbf{x}%
)^{-1}\mathbf{S}(\mathbf{x}\text{;\thinspace}{\boldsymbol{\theta}})\right)
\right\rangle _{\mathbf{x}}-\lambda_{1}+\lambda_{2}({\boldsymbol{\theta}%
})=0\text{.} \label{LPtheta.2}%
\end{equation}
It follows from (\ref{LPtheta.1}) and (\ref{LPtheta.2}) that
\begin{align}
\frac{1}{2}\left\langle \mathrm{Tr}\left(  N\mathbf{G}(\mathbf{x}%
)^{-1}\mathbf{S}(\mathbf{x}\text{;\thinspace}{\boldsymbol{\theta}})\right)
\right\rangle _{\mathbf{x}}  &  ={\lambda_{1}}\text{{,}}{\mathrm{\;}\quad
p({\boldsymbol{\theta}})\neq0}\text{{,}}\label{LPtheta.3a}\\
\frac{1}{2}\left\langle \mathrm{Tr}\left(  N\mathbf{G}(\mathbf{x}%
)^{-1}\mathbf{S}(\mathbf{x}\text{;\thinspace}{\boldsymbol{\theta}})\right)
\right\rangle _{\mathbf{x}}  &  ={\lambda_{1}-\lambda_{2}({\boldsymbol{\theta
}})}\text{{,}}{\mathrm{\;}\quad p({\boldsymbol{\theta}})=0}\text{{.}}
\label{LPtheta.3b}%
\end{align}
Since $I_{G}[p({\boldsymbol{\theta}})]$ is a concave function of
$p({\boldsymbol{\theta}})$, Eq. (\ref{LPtheta.3a}) and (\ref{LPtheta.3b}) are
the necessary and sufficient conditions for the optimization problem
(\ref{maxIa}) and (\ref{conditionPtheta}).\textbf{\ }

\subsection{Channel Capacity for Neural Population Coding}

\label{Sec:4.1.3} If $p(\mathbf{x})$ is unknown, then by Jensen's inequality,
we have
\begin{align}
I\simeq{I_{G}\left[  p(\mathbf{x})\right]  }  &  ={\int_{{{\mathcal{X}}}%
}p(\mathbf{x})\ln\left(  p(\mathbf{x})^{-1}\det\left(  \dfrac{\mathbf{G}%
(\mathbf{x})}{2\pi e}\right)  ^{1/2}\right)  d}\mathbf{x}\nonumber\\
&  \leq\ln\int_{{\mathcal{X}}}\det\left(  \dfrac{\mathbf{G}(\mathbf{x})}{2\pi
e}\right)  ^{1/2}d\mathbf{x}\text{,} \label{4_1_3.1}%
\end{align}
and the equality holds if and only if $p(\mathbf{x})^{-1}\det\left(
\mathbf{G}(\mathbf{x})\right)  ^{1/2}$ is a constant. Thus
\begin{align}
&  {I_{G}[p^{\ast}(\mathbf{x})]=}\underset{p(\mathbf{x})}{{\max}}{\left(
I_{G}[p(\mathbf{x})]\right)  =\ln\int_{{\mathcal{X}}}\det\left(
\dfrac{\mathbf{G}(\mathbf{x})}{2\pi e}\right)  ^{1/2}d}\mathbf{x}%
\text{,}\label{4_1_3.2}\\
&  {p^{\ast}(\mathbf{x})=\dfrac{\det\left(  \mathbf{G}(\mathbf{x})\right)
^{1/2}}{\int_{{\mathcal{X}}}\det\left(  \mathbf{G}(\mathbf{\hat{x}})\right)
^{1/2}d\mathbf{\hat{x}}}}\text{,} \label{4_1_3.2b}%
\end{align}
assuming $\int_{{\mathcal{X}}}\det\left(  \mathbf{G}(\mathbf{\hat{x}})\right)
^{1/2}d\mathbf{\hat{x}}<\infty$.

Let us consider a specific example. Suppose $\mathbf{J}(\mathbf{x}%
)=\mathbf{J}_{0}$ is a constant matrix, then it follows from (\ref{IG}) that%
\begin{equation}
I_{G}=\frac{1}{2}\left\langle \ln\left(  \det\left(  \frac{\mathbf{J}%
_{0}+\mathbf{P}(\mathbf{x})}{2\pi e}\right)  \right)  \right\rangle
_{\mathbf{x}}+H(X)\text{.} \label{4_1_3.2bc}%
\end{equation}
According to the maximum entropy probability distribution, we know that
maximizing $H(X)$ results in a uniformly distributed ${p(\mathbf{x})}$. Hence
we have $\mathbf{G}(\mathbf{x})=\mathbf{J}_{0}$ and ${p^{\ast}(\mathbf{x})}$
coincides with the uniform distribution (see \ref{4_1_3.2b}). In this case,
the maximum $I_{G}[p^{\ast}(\mathbf{x})]$ can be regarded as the channel
capacity for this neural population.

If we consider a constraint on random variables $X$ and assume that the
covariance matrix of $X$ is $\boldsymbol{\Sigma}_{0}$ and satisfies%
\begin{equation}
\boldsymbol{\Sigma}_{0}^{-1}=\mathbf{P}(\mathbf{x})\text{,} \label{4_1_3.2c}%
\end{equation}
then it follows from the maximum entropy probability distribution that%
\begin{equation}
H(X)\leq\frac{1}{2}\left(  \det\left(  2\pi e\boldsymbol{\Sigma}_{0}\right)
\right)  \text{,} \label{4_1_3.2d}%
\end{equation}
and the equality holds if and only if the p.d.f.\ of the input is a normal
distribution: $p(\mathbf{x})={\mathcal{N}\left(  {\boldsymbol{\mu}%
}\text{,\thinspace}\boldsymbol{\Sigma}_{0}\right)  }$. Hence
\begin{align}
I_{G}  &  =\frac{1}{2}\ln\left(  \det\left(  \frac{\mathbf{J}_{0}%
+\boldsymbol{\Sigma}_{0}^{-1}}{2\pi e}\right)  \right)  +H(X)\nonumber\\
&  \leq\frac{1}{2}\ln\left(  \det\left(  \boldsymbol{\Sigma}_{0}\mathbf{J}%
_{0}+\mathbf{I}_{K}\right)  \right)  =I_{G}[p^{\ast}(\mathbf{x})]\text{,}
\label{4_1_3.2e}%
\end{align}
where $I_{G}[p^{\ast}(\mathbf{x})]$ is the channel capacity of neural
population. Here the equality holds if and only if $p^{\ast}(\mathbf{x}%
)={\mathcal{N}\left(  {\boldsymbol{\mu}}\text{,\thinspace}\boldsymbol{\Sigma
}_{0}\right)  }$, which is consistent with Eq. (\ref{4_1_3.2c}).

Furthermore, if $\varsigma_{1}\approx0$ (see \ref{Thm1a.3a}), we have
\begin{equation}
I\simeq{I_{G}\left[  p(\mathbf{x})\right]  }\simeq I_{F}[p(\mathbf{x}%
)]=\int_{{{\mathcal{X}}}}p(\mathbf{x})\ln\left(  p(\mathbf{x})^{-1}\det\left(
\dfrac{\mathbf{J}(\mathbf{x})}{2\pi e}\right)  ^{1/2}\right)  d\mathbf{x}%
\text{.} \label{4_1_3.4}%
\end{equation}

Similarly, we also get
\begin{align}
&  {I_{F}[p^{\ast}(\mathbf{x})]}={\underset{p(\mathbf{x})}{{\max}}\left(
I_{F}[p(\mathbf{x})]\right)  =\ln\int_{{\mathcal{X}}}\det\left(
\dfrac{\mathbf{J}(\mathbf{x})}{2\pi e}\right)  ^{1/2}d}\mathbf{x}%
\text{,}\label{4_1_3.4a}\\
&  {p^{\ast}(\mathbf{x})=\dfrac{\det\left(  \mathbf{J}(\mathbf{x})\right)
^{1/2}}{\int_{{\mathcal{X}}}\det\left(  \mathbf{J}(\mathbf{\hat{x}})\right)
^{1/2}d\mathbf{\hat{x}}}}\text{,} \label{4_1_3.4b}%
\end{align}
assuming $\int_{{\mathcal{X}}}\det\left(  \mathbf{J}(\hat{\mathbf{x}})\right)
^{1/2}d\mathbf{\hat{x}}<\infty$. Here $I_{F}[p^{\ast}(\mathbf{x})]$ is the
channel capacity of the neural population. The distribution $p^{\ast
}(\mathbf{x})$ coincides with the Jeffrey's prior in Bayesian probability
\citep{Jeffreys(1961-BK-theory)}. In this case, if we suppose the covariance
matrix of $X$ is $\boldsymbol{\Sigma}_{0}$, then similar to (\ref{4_1_3.2d})
and (\ref{4_1_3.2e}), we can get the channel capacity
\begin{equation}
I_{F}[p^{\ast}(\mathbf{x})]=\frac{1}{2}\ln\left(  \det\left(
\boldsymbol{\Sigma}_{0}\mathbf{J}_{0}\right)  \right)  \label{4_1_3.4c}%
\end{equation}
with $p^{\ast}(\mathbf{x})={\mathcal{N}\left(  {\boldsymbol{\mu}%
}\text{,\thinspace}\boldsymbol{\Sigma}_{0}\right)  }$.

For another example, consider the Poisson neuron model given in
(\ref{PoissNeuron}) and suppose the input $x$ is one dimension, $K=1$. It
follows from (\ref{PoissNeuron.1}) and (\ref{4_1_3.4b}) that
\begin{equation}
p^{\ast}(x)=\frac{\left(  {\int_{{{\Theta}}}p({\boldsymbol{\theta}})}\left(
{\frac{\partial g(x\text{;\thinspace}{\boldsymbol{\theta}})}{\partial x}%
}\right)  ^{2}{d{\boldsymbol{\theta}}}\right)  ^{1/2}}{\int_{{\mathcal{X}}%
}\left(  {\int_{{{\Theta}}}p({\boldsymbol{\theta}})}\left(  {\frac{\partial
g(\hat{x}\text{;\thinspace}{\boldsymbol{\theta}})}{\partial\hat{x}}}\right)
^{2}{d{\boldsymbol{\theta}}}\right)  ^{1/2}d\hat{x}}\text{.} \label{4_1_3.5}%
\end{equation}
If ${p({\boldsymbol{\theta}})=\delta({\boldsymbol{\theta}}}%
-{{\boldsymbol{\theta}}}_{0}{)}$, Eq.~(\ref{4_1_3.5}) becomes%
\begin{equation}
p^{\ast}(x)=\frac{\left\vert {\frac{\partial g(x\text{;\thinspace
}{\boldsymbol{\theta}}_{0})}{\partial x}}\right\vert }{\int_{{\mathcal{X}}%
}\left\vert {\frac{\partial g(\hat{x}\text{;\thinspace}{\boldsymbol{\theta}%
}_{0})}{\partial\hat{x}}}\right\vert d\hat{x}}\text{.} \label{4_1_3.6}%
\end{equation}

\cite{Atick(1990-towards)} presented a redundancy measure to approximate
Barlow's optimality\ principle:%
\begin{equation}
\mathcal{R}=1-\frac{I(X;R)}{C(R)}\text{,} \label{minR}%
\end{equation}
where $C(R)$\ is the channel capacity. Here for neural population coding we
have $C(R)\simeq I_{G}[p^{\ast}(\mathbf{x})]$ and $I(X;R)\simeq I_{G}$ (or
$C(R)\simeq I_{F}[p^{\ast}(\mathbf{x})]$ and $I(X;R)\simeq I_{F})$. Hence we
can minimize $\mathcal{R}$\ by choosing an appropriate $\mathbf{J}%
(\mathbf{x})$ to maximize $I_{G}$ (or $I_{F}$)\ and simultaneously satisfying
(\ref{4_1_3.2b}) (or \ref{4_1_3.4b}) \citep[see ][for further details]{Huang(2017-IC-information)}.

\section{Discussion}

\label{Conclusion}In this paper we have derived several information-theoretic
bounds and approximations for effective approximation of MI in the context of
neural population coding for large but finite population size. We have found
some regularity conditions under which the asymptotic bounds and
approximations hold. Generally speaking, these regularity conditions are easy
to meet. Special examples that satisfy these conditions include the cases when
the likelihood function $p(\mathbf{r|x})$ for the neural population responses
is conditionally independent or has correlated noises with a multivariate
Gaussian distribution. Under the general regularity conditions we have derived
several asymptotic bounds and approximations of MI for a neural population and
found some relationships among different approximations.

How to choose among these different asymptotic approximations of MI in a
neural population with finite size $N$? For a flat prior distribution
$p(\mathbf{x})$, we have $I_{G}\simeq I_{F}$; that is, the two approximations
$I_{G}$ and $I_{F}$ are about equally valid. For a sharply peaked prior
distribution $p(\mathbf{x})$, $I_{G}$ is generally a better approximation to
MI $I$ than $I_{F}$. Under suitable conditions (e.g. \textbf{C1} and
\textbf{C2}) for low-dimensional inputs, $I_{G}$, and $I_{F}$ are good
approximations of MI $I$ not only for large $N$ but also for small $N$. For
high-dimensional inputs, the FI matrix $\mathbf{J}(\mathbf{x})$ (see Eq.
\ref{IF}) or matrix $\mathbf{P}^{-1}(\mathbf{x})$\ (see Eq. \ref{Px}) often
becomes degenerate, which causes a large error between $I_{F}$ and MI $I$.
Hence, in this situation, $I_{G}$ is a better approximation to MI $I$ than
$I_{F}$. For more convenient computation of the approximation, we have also
introduced the approximation formula $I_{G_{+}}$ which may substitute for
$I_{G}$ as a proxy of MI $I$. For some special cases (see \textbf{Corollary
\ref{Corollary 2}}), $I_{G}$ and $I_{G_{+}}$\ are strictly equal to the true
MI $I$. Our simulation results for the one-dimensional case shows that the
approximations $I_{G}$, $I_{G_{+}}$, and $I_{F}$ are all highly precise
compared with the true MI $I$, even for small $N$ (Figure 1).

These approximation formulas satisfy additional constraints. By the
Cram\'{e}r-Rao lower bound, we know that $I_{F}$\ is related to the covariance
matrix of an unbiased estimator (see Eq. \ref{Fisher}). By the van Trees'
Bayesian Cram\'{e}r-Rao bound, we get a link between $I_{G_{+}}$ and the
covariance matrix of a biased estimator (see Eq. \ref{vt1}). From the point of
view of neural population decoding and Bayesian inference, there is a
connection between MI (or $I_{G}$)\ and MAP (see Eq.~\ref{MAP.6}).

For more efficient calculation of the approximation $I_{G}$ (or $I_{G_{+}}$)
for high-dimensional inputs, we propose to apply an invertible transformation
on the input variable so as to make the new variable closer to a normal
distribution (see section \ref{Sec:3.1}). Another useful technique is
dimensionality reduction which effectively approximates MI by further reducing
the computational complexity for high-dimensional inputs. We found that
$I_{F}$ could lead to huge errors as a proxy of the true MI $I$ for
high-dimensional inputs even when $I_{G}$ and $I_{G_{+}}$\ are strictly equal
to the true MI $I$.

These approximation formulas are potentially useful for optimization problems
of information transfer in neural population coding. We have proven that
optimizing the population density distribution of parameters
$p(\boldsymbol{\theta})$ is a convex optimization problem and have found a set
of necessary and sufficient conditions. The approximation formulas are also
useful for discussion of the channel capacity of neural population coding
(section \ref{Sec:4.1.3}).

The information theory is a powerful tool for neuroscience and other
disciplines, including diverse fields such as physics, information and
communication technology, machine learning, computer vision, and
bioinformatics. Finding effective approximation methods for computing MI is a
key for many practical applications of information theory. 
Generally speaking,
the FI matrix is  easier to evaluate or approximate than MI. 
This is because calculation of MI involves averaging over both the input variable $\bf x$
and the output variable $\bf r$ (see Eq.~2.1), and typically $p(\mathbf{r})$ also needs to be calculated
from $p(\mathbf{r}|\mathbf{x})$ by another average over $\bf x$ (see Eq.~2.2).
By contrast, the FI matrix $\mathbf{J}(\mathbf{x})$ involves averaging over $\bf r$ only (see Eq. \ref{Jx}). 
Furthermore, it is often easier to find analytical forms of FI for specific models such as a population of tuning curves with Poisson spike statistics.
Taking into account the
computational efficiency, for practical applications we suggest using $I_{G}$
or $I_{G_{+}}$ as a proxy of the true MI $I$ for most cases. These
approximations could be very useful even when we do not need to know the exact
value of MI. For example, for some optimization and learning problems, we only
need to know how MI is affected by the conditional p.d.f.\ or likelihood
function $p(\mathbf{r}|\mathbf{x})$. In such situations, we may easily solve
for the optimal parameters using the approximation formulas
\citep[]{Huang(2017-IC-information), Huang(2017-IC-information)a}. Further
discussions of the applications will be given in separate publications.

\section*{Acknowledgments}

\phantomsection\addcontentsline{toc}{section}{Acknowledgments}

This work was supported partially by an NIH grant R01 DC013698.

\noindent

\appendix
\renewcommand{\theequation}{A.\arabic{equation}} \setcounter{equation}{0} \renewcommand\thesubsection{A.\arabic{subsection}}

\section*{Appendix: The Proofs}

\phantomsection\addcontentsline{toc}{section}{Appendix: The Proofs}
\label{Appendix} We consider a Taylor expanding of $L(\mathbf{r}%
|{\mathbf{\hat{x}}})$ around $\mathbf{x}$. If $L(\mathbf{r}|{\mathbf{\hat{x}}%
})$ is twice differentiable for $\forall{\mathbf{\hat{x}}}\in{{\mathcal{X}}%
}_{\omega}(\mathbf{x})$, then by condition \textbf{C1} we get
\begin{align}
&  {L(\mathbf{r}|\mathbf{\hat{x}})-L(\mathbf{r}|\mathbf{x})}\nonumber\\
&  =\left(  {\mathbf{\hat{x}}}-{\mathbf{x}}\right)  ^{T}{L}^{\prime
}{(\mathbf{r}|\mathbf{x})}+{\dfrac{1}{2}\left(  {\mathbf{\hat{x}}}%
-{\mathbf{x}}\right)  ^{T}L^{\prime\prime}(\mathbf{r}|\mathbf{\breve{x}}%
)}\left(  {\mathbf{\hat{x}}}-{\mathbf{x}}\right) \nonumber\\
&  ={\mathbf{y}^{T}\mathbf{\tilde{v}}-\dfrac{1}{2}\mathbf{y}^{T}%
\mathbf{y}+\dfrac{1}{2}\mathbf{y}^{T}}\mathbf{By}\text{,} \label{A.Lrx}%
\end{align}
where%
\begin{equation}
{\mathbf{y}}=\mathbf{G}{{^{1/2}}\left(  {\mathbf{x}}\right)  (\mathbf{\hat{x}%
}}-{\mathbf{x})}\text{,} \label{A.y}%
\end{equation}%
\begin{equation}
\mathbf{\tilde{v}}=\mathbf{\mathbf{v}}+\mathbf{\mathbf{v}}_{1}\text{,
}\mathbf{v}=\mathbf{G}{^{-1/2}}\left(  {\mathbf{x}}\right)  {l^{\prime
}(\mathbf{r}|\mathbf{x})}\text{, }\mathbf{v}{_{1}}=\mathbf{G}{^{-1/2}}\left(
{\mathbf{x}}\right)  {{q}^{\prime}(\mathbf{x})}\text{,} \label{A.vy}%
\end{equation}%
\begin{equation}
{\mathbf{\breve{x}}}={\mathbf{x}}+t\left(  {\mathbf{\hat{x}}}-{\mathbf{x}%
}\right)  \in{{\mathcal{X}}}_{\omega}(\mathbf{x})\text{,\ }t\in\left(
0\text{,\thinspace}1\right)  \text{,} \label{A.x.t}%
\end{equation}%
\begin{equation}
\left\{
\begin{array}
[c]{l}%
\mathbf{B}=\mathbf{G}{^{-1/2}}\left(  {\mathbf{x}}\right)  {\mathbf{C}%
\mathbf{G}^{-1/2}(\mathbf{x})}=\mathbf{B}{_{0}}+\mathbf{B}{_{1}}%
+\mathbf{B}{_{2}}\text{,}\\
\mathbf{C}=\mathbf{C}{_{0}}+\mathbf{C}{_{1}}+\mathbf{C}{_{2}}\text{,}%
\end{array}
\right.  \label{A.BC}%
\end{equation}
and
\begin{equation}
\left\{
\begin{array}
[c]{l}%
\mathbf{B}{_{0}=\mathbf{G}^{-1/2}}\left(  {\mathbf{x}}\right)  \mathbf{C}%
{_{0}\mathbf{G}^{-1/2}}\left(  {\mathbf{x}}\right)  \text{,}\\
\mathbf{B}{_{1}=\mathbf{G}^{-1/2}}\left(  {\mathbf{x}}\right)  \mathbf{C}%
{_{1}\mathbf{G}^{-1/2}}\left(  {\mathbf{x}}\right)  \text{,}\\
\mathbf{B}{_{2}=\mathbf{G}^{-1/2}}\left(  {\mathbf{x}}\right)  \mathbf{C}%
{_{2}\mathbf{G}^{-1/2}}\left(  {\mathbf{x}}\right)  \text{,}\\
\mathbf{C}{_{0}=l^{\prime\prime}(\mathbf{r}|\mathbf{x})}-\left\langle
{l^{\prime\prime}(\mathbf{r}|\mathbf{x})}\right\rangle _{{\mathbf{r}%
|\mathbf{x}}}\text{,}\\
\mathbf{C}{_{1}=l^{\prime\prime}(\mathbf{r}|\mathbf{\breve{x}})-l^{\prime
\prime}(\mathbf{r}|\mathbf{x})}\text{,}\\
\mathbf{C}{_{2}= {q}^{\prime\prime}\left(  {\mathbf{\breve{x}}}\right)  }-
{q}^{\prime\prime}\left(  \mathbf{x}\right)  \text{.}%
\end{array}
\right.  \label{A.BC.1}%
\end{equation}

By condition \textbf{C1}, we know that the matrix $\mathbf{B}{_{1}}%
+\mathbf{B}{_{2}}$ is continuous and symmetric for ${\mathbf{\breve{x}}}%
\in{{\mathcal{X}}}_{\omega}$ and ${\left\Vert \mathbf{B}{_{1}}+\mathbf{B}%
{_{2}}\right\Vert }=O\left(  1\right)  $. By the definition of continuous
functions, we can prove the following: for any $\epsilon\in\left(
0\text{,\thinspace}1\right)  $, there is an ${\varepsilon}\in\left(
0\text{,\thinspace}\omega\right)  $\ such that for all ${\mathbf{y}}%
\in{\mathcal{Y}_{\varepsilon}}$%
\begin{equation}
-\epsilon\mathbf{I}_{K}\leq\mathbf{B}{_{1}}+\mathbf{B}{_{2}}\leq
\epsilon\mathbf{I}_{K}\text{,} \label{A.BI}%
\end{equation}
where
\begin{equation}
{\mathcal{Y}_{\varepsilon}}=\left\{  {\mathbf{y}}\in%
\mathbb{R}
^{K}:\left\Vert \mathbf{y}\right\Vert <\varepsilon\sqrt{N}\right\}  \text{.}
\label{A.omega}%
\end{equation}
Hence,
\begin{equation}
\left\vert {\mathbf{y}^{T}}\left(  \mathbf{B}{_{1}}+\mathbf{B}{_{2}}\right)
{\mathbf{y}}\right\vert <\epsilon{\left\Vert {\mathbf{y}}\right\Vert ^{2}%
}\text{\textrm{.}} \label{A.yy}%
\end{equation}
Here ${\mathbf{\breve{x}}}={\mathbf{x}}+t\mathbf{G}{{^{-1/2}}}\left(
{{\mathbf{x}}}\right)  \mathbf{y}$, $\varepsilon$\ is a function of
${\mathbf{r}}$, $\varepsilon=\varepsilon\left(  {\mathbf{r}}\right)  =O\left(
1\right)  $, and
\begin{equation}
{\mathcal{Y}_{\varepsilon}\subseteq\mathcal{Y}_{\omega}}=\left\{  {\mathbf{y}%
}\in%
\mathbb{R}
^{K}:\left\Vert \mathbf{y}\right\Vert <\omega\sqrt{N}\right\}  \text{.}
\label{A.Yo}%
\end{equation}

We define the sets
\begin{equation}
\left\{
\begin{array}
[c]{l}%
{\mathcal{\bar{Y}}_{\varepsilon}=\left\{  \mathbf{y}\in%
\mathbb{R}
^{K}:\left\Vert \mathbf{y}\right\Vert \geq\varepsilon\sqrt{N}\right\}
}\text{,}\\
{\mathcal{Z}_{\hat{\varepsilon}}=\left\{  \mathbf{z}\in%
\mathbb{R}
^{K}:\left\vert z_{k}\right\vert <\hat{\varepsilon}\sqrt{N/K}\text{,\thinspace
}\forall k=1\text{,\thinspace}2\text{,\thinspace}\cdots\text{,\thinspace
}K\right\}  }\text{,}\\
{\bar{\mathcal{Z}}_{\hat{\varepsilon}}=\left\{  \mathbf{z}\in%
\mathbb{R}
^{K}:\left\vert z_{k}\right\vert \geq\hat{\varepsilon}\sqrt{N/K}%
\text{,\thinspace}\forall k=1\text{,\thinspace}2\text{,\thinspace}%
\cdots\text{,\thinspace}K\right\}  }\text{,}\\
{\mathcal{\tilde{Z}}_{\varepsilon}=\left\{  \mathbf{z}\in%
\mathbb{R}
^{K}:\left\Vert \mathbf{z}+\mathbf{\tilde{v}}1_{\mathcal{R}_{\hat{\varepsilon
}}}\right\Vert <\varepsilon\sqrt{N}\right\}  }\text{,}%
\end{array}
\right.  \label{A.YZ}%
\end{equation}
where
\begin{equation}
\hat{\varepsilon}=\varepsilon/2\text{,} \label{A.Eps'v}%
\end{equation}
$1_{\left(  \mathcal{\cdot}\right)  }$ denotes an indicator random variable,
\begin{equation}
1_{\mathcal{R}_{\hat{\varepsilon}}}=\left\{
\begin{array}
[c]{l}%
{1}\text{{\textrm{,\thinspace}}}{\mathrm{\;}\mathbf{r}\in\mathcal{R}%
_{\hat{\varepsilon}}(\mathbf{x})}\\
{0}\text{{\textrm{,\thinspace}}}{\mathrm{\;}\mathbf{r}\notin\mathcal{R}%
_{\hat{\varepsilon}}(\mathbf{x})}%
\end{array}
\right.  \text{\textrm{,\thinspace}}\mathrm{\;} \quad1_{\bar{\mathcal{R}%
}_{\hat{\varepsilon}}}=\left\{
\begin{array}
[c]{l}%
{1}\text{{\textrm{,\thinspace}}}{\mathrm{\;}\mathbf{r}\in\bar{\mathcal{R}%
}_{\hat{\varepsilon}}(\mathbf{x})}\\
{0}\text{{\textrm{{,}\thinspace}}}{\mathrm{\;}\mathbf{r}\notin\bar
{\mathcal{R}}_{\hat{\varepsilon}}(\mathbf{x})}%
\end{array}
\right.  \text{,} \label{A.IR}%
\end{equation}
and
\begin{equation}
\left\{
\begin{array}
[c]{l}%
\mathcal{R}{_{\hat{\varepsilon}}(\mathbf{x})=\left\{  \mathbf{r}\in
\mathcal{R}:\left\Vert \mathbf{\tilde{v}}\right\Vert <\hat{\varepsilon}%
\sqrt{N}\right\}  }\text{,}\\
{\bar{\mathcal{R}}_{\hat{\varepsilon}}(\mathbf{x})=\left\{  \mathbf{r}%
\in\mathcal{R}:\left\Vert \mathbf{\tilde{v}}\right\Vert \geq\hat{\varepsilon
}\sqrt{N}\right\}  }\text{.}%
\end{array}
\right.  \label{A.R}%
\end{equation}
For all $\mathbf{z}\in\mathcal{Z}_{\hat{\varepsilon}}$, we have $\left\Vert
\mathbf{z}+\mathbf{\tilde{v}}1_{\mathcal{R}_{\hat{\varepsilon}}}\right\Vert
_{2}\leq\left\Vert \mathbf{z}\right\Vert _{2}+\left\Vert \mathbf{v}%
1_{\mathcal{R}_{\hat{\varepsilon}}}\right\Vert _{2}<\varepsilon\sqrt{N}$,
then
\begin{equation}
\mathcal{Z}_{\hat{\varepsilon}}\subseteq{\mathcal{\tilde{Z}}_{\varepsilon}%
}\text{.} \label{A.Z<}%
\end{equation}

It follows from (\ref{A.vy}) and (\ref{A.BC.1}) that%
\begin{equation}
\left\langle \mathbf{v}\right\rangle _{\mathbf{r}|\mathbf{x}}=0\text{,
}{\left\langle \mathbf{B}_{0}\right\rangle _{\mathbf{r}|\mathbf{x}}=0}\text{,}
\label{A.vB=}%
\end{equation}
and
\begin{align}
{\left\langle \left\langle \tilde{\mathbf{v}}^{T}\tilde{\mathbf{v}%
}\right\rangle _{\mathbf{r}|\mathbf{x}}\right\rangle _{\mathbf{x}}}  &
={\left\langle \left\langle L^{\prime}(\mathbf{r}|\mathbf{x})^{T}%
\mathbf{G}^{-1}\left(  \mathbf{x}\right)  L^{\prime}(\mathbf{r}|\mathbf{x}%
)\right\rangle _{\mathbf{r}|\mathbf{x}}\right\rangle _{\mathbf{x}}}\nonumber\\
&  =\left\langle {\mathrm{Tr}}\left(  \left\langle L^{\prime}(\mathbf{r}%
|\mathbf{x})L^{\prime}(\mathbf{r}|\mathbf{x})^{T}\right\rangle _{\mathbf{r}%
|\mathbf{x}}\mathbf{G}^{-1}\left(  \mathbf{x}\right)  \right)  \right\rangle
_{\mathbf{x}}\nonumber\\
&  =K+\zeta\nonumber\\
&  =K+O\left(  N^{-1}\right)  \text{,} \label{A.vv=}%
\end{align}
and it follows from condition \textbf{C1} that%
\begin{align}
\zeta &  =\left\langle {\mathrm{Tr}}\left(  \dfrac{1}{p(\mathbf{x})}%
\dfrac{\partial^{2}p(\mathbf{x})}{\partial\mathbf{x}\partial\mathbf{x}^{T}%
}\mathbf{G}^{-1}\left(  \mathbf{x}\right)  \right)  \right\rangle
_{\mathbf{x}}\nonumber\\
&  =\left\langle {\mathrm{Tr}}\left(  \left(  {{q}^{\prime}}\left(
{\mathbf{x}}\right)  ^{T}{{q}^{\prime}}\left(  {\mathbf{x}}\right)
+{{q}^{\prime\prime}}\left(  {\mathbf{x}}\right)  \right)  \mathbf{G}%
^{-1}\left(  \mathbf{x}\right)  \right)  \right\rangle _{\mathbf{x}%
}\nonumber\\
&  \leq\left\langle N^{-1}\left(  \left\Vert {{q}^{\prime}}\left(
{\mathbf{x}}\right)  ^{T}{{q}^{\prime}}\left(  {\mathbf{x}}\right)
\right\Vert +\left\Vert {{q}^{\prime\prime}}\left(  {\mathbf{x}}\right)
\right\Vert \right)  \left\Vert N\mathbf{G}^{-1}\left(  \mathbf{x}\right)
\right\Vert \right\rangle _{\mathbf{x}}\nonumber\\
&  =O\left(  N^{-1}\right)  \text{.} \label{A.Xi}%
\end{align}
Combining conditions \textbf{C1} and \textbf{C2}, (\ref{A.vy}), (\ref{A.x.t})
and (\ref{A.BC.1}), we find
\begin{equation}
\left\{
\begin{array}
[c]{l}%
{\left\langle \left\Vert \mathbf{B}_{0}\right\Vert ^{2m}\right\rangle
_{\mathbf{r}|\mathbf{x}}}\leq\left\langle {\left\langle \left\Vert
N^{-1}\mathbf{C}_{0}\right\Vert ^{2m}\left\Vert N\mathbf{G}^{-1}\left(
\mathbf{x}\right)  \right\Vert ^{2m}\right\rangle _{\mathbf{r}|\mathbf{x}}%
}\right\rangle _{\mathbf{x}}=O\left(  N^{-1}\right)  \text{,}\\
{\left\langle \left\Vert \mathbf{B}_{0}\right\Vert ^{2m+1}\right\rangle
_{\mathbf{r}|\mathbf{x}}}\leq\left\langle \left\Vert N\mathbf{G}^{-1}\left(
\mathbf{x}\right)  \right\Vert ^{2m+1}{\left\langle \left\Vert N^{-1}%
\mathbf{C}_{0}\right\Vert ^{2}\right\rangle _{\mathbf{r}|\mathbf{x}}%
^{1/2}\left\langle \left\Vert N^{-1}\mathbf{C}_{0}\right\Vert ^{4m}%
\right\rangle _{\mathbf{r}|\mathbf{x}}^{1/2}}\right\rangle _{\mathbf{x}}\\
\ \ \ \ \ \ \ \ \ \ \ \ \ \ \ \ \ \ \ \ \ \ \ \ \ \ \ \ =O\left(
N^{-1}\right)  \text{,}\\
{\left\langle \left\Vert \mathbf{v}\right\Vert ^{2m_{0}}\right\rangle
_{\mathbf{r}|\mathbf{x}}}\leq{\left\langle \left\vert N^{-1}l^{\prime
}(\mathbf{r}|\mathbf{x})^{T}l^{\prime}(\mathbf{r}|\mathbf{x})\right\vert
^{m_{0}}\right\rangle _{\mathbf{r}|\mathbf{x}}}\left\Vert N\mathbf{G}%
^{-1}\left(  \mathbf{x}\right)  \right\Vert ^{m_{0}}=O\left(  1\right)
\text{,}\\
{\left\Vert \mathbf{v}_{1}\right\Vert ^{2m_{0}}}\leq{\left\vert N^{-1}%
{q}{^{\prime}(\mathbf{x})}^{T}{q}{^{\prime}(\mathbf{x})}\right\vert ^{m_{0}%
}\left\Vert N\mathbf{G}^{-1}\left(  \mathbf{x}\right)  \right\Vert ^{m_{0}}%
}=O\left(  N^{-m_{0}}\right)  \text{,}%
\end{array}
\right.  \label{A.vv1<}%
\end{equation}
together with the power mean inequality,%
\begin{align}
{\left\langle \left(  {\tilde{\mathbf{v}}^{T}\tilde{\mathbf{v}}}\right)
^{m_{0}}\right\rangle _{\mathbf{r}|\mathbf{x}}}  &  \leq{\left\langle \left(
\left\Vert {\mathbf{v}}\right\Vert +\left\Vert {\mathbf{v}}_{1}\right\Vert
\right)  ^{2m_{0}}\right\rangle _{\mathbf{r}|\mathbf{x}}}\nonumber\\
&  \leq2^{2m_{0}-1}{\left\langle \left\Vert {\mathbf{v}}\right\Vert ^{2m_{0}%
}+\left\Vert {\mathbf{v}}_{1}\right\Vert ^{2m_{0}}\right\rangle _{\mathbf{r}%
|\mathbf{x}}}\nonumber\\
&  =O\left(  1\right)  \text{,} \label{A.vvm<}%
\end{align}
where $m\in\mathfrak{%
\mathbb{N}
}$, $m_{0}\in\left\{  1\text{,\thinspace}2\right\}  $. Notice that $\left\Vert
\mathbf{G}^{-1}\left(  \mathbf{x}\right)  \right\Vert =O\left(  N^{-1}\right)
$. Here we note that for all conformable matrices $\mathbf{A}$ and
$\mathbf{B}$,
\begin{equation}
\left\{
\begin{array}
[c]{l}%
{\left\vert \mathrm{Tr}\left(  \mathbf{AB}\right)  \right\vert \leq\left\Vert
\mathbf{A}\right\Vert \left\Vert \mathbf{B}\right\Vert }\text{,}\\
{\left\Vert \mathbf{AB}\right\Vert \leq\left\Vert \mathbf{A}\right\Vert
\left\Vert \mathbf{B}\right\Vert }\text{.}%
\end{array}
\right.  \label{A.AB<}%
\end{equation}

By (\ref{C1.b1}) we have
\begin{align}
\mathrm{Tr}\left(  N^{-1}\mathbf{J}\left(  \mathbf{x}\right)  \right)  ^{2}
&  ={\left\langle N^{-1}l^{\prime}(\mathbf{r}|\mathbf{x})^{T}l^{\prime
}(\mathbf{r}|\mathbf{x})\right\rangle _{\mathbf{r}|\mathbf{x}}^{2}}\nonumber\\
&  {\leq}{\left\langle \left(  N^{-1}l^{\prime}(\mathbf{r}|\mathbf{x}%
)^{T}l^{\prime}(\mathbf{r}|\mathbf{x})\right)  ^{2}\right\rangle
_{\mathbf{r}|\mathbf{x}}}=O\left(  1\right)  \text{.} \label{A.TrJ}%
\end{align}
Then it follows from (\ref{C1.a2}) and (\ref{A.TrJ}) that%
\begin{equation}
\det\left(  \mathbf{G}\left(  \mathbf{x}\right)  \right)  =O\left(
N^{K}\right)  \text{.} \label{A.detG}%
\end{equation}

\subsection{Proof of Lemma \ref{Lemma 1}}

\label{A.Lma1P}It follows from (\ref{A.Lrx}) that
\begin{align}
{\Gamma_{\omega}}  &  ={\left\langle \left\langle \ln\int_{{{\mathcal{X}}%
}_{\omega}(\mathbf{x})}\exp\left(  L(\mathbf{r}|\mathbf{\hat{x}}%
)-L(\mathbf{r}|\mathbf{x})\right)  d\mathbf{\hat{x}}\right\rangle
_{\mathbf{r}|\mathbf{x}}\right\rangle _{\mathbf{x}}}\nonumber\\
&  =-\left\langle \dfrac{1}{2}\ln\left(  \det\left(  \mathbf{G}(\mathbf{x}%
)\right)  \right)  \right\rangle _{\mathbf{x}}\nonumber\\
&  +{\underset{{\hat{\Gamma}}_{\omega}}{\underbrace{\left\langle \left\langle
\ln\left(  \int_{{\mathcal{Y}}_{\omega}}\exp\left(  \mathbf{y}^{T}%
\tilde{\mathbf{v}}-\dfrac{1}{2}\mathbf{y}^{T}\mathbf{y}+\dfrac{1}{2}%
\mathbf{y}^{T}\mathbf{By}\right)  d\mathbf{y}\right)  \right\rangle
_{\mathbf{r}|\mathbf{x}}\right\rangle _{\mathbf{x}}}}}\text{.}
\label{A.GammaEps}%
\end{align}
For $\mathbf{y}\in{\mathcal{Y}}_{\varepsilon}$, according to the definitions
in (\ref{A.IR}) and (\ref{A.R}), we have
\begin{align}
{\left\vert \mathbf{y}^{T}{\tilde{\mathbf{v}}}1_{\bar{\mathcal{R}}%
_{\hat{\varepsilon}}}\right\vert }  &  \leq{\left\Vert \mathbf{y}\right\Vert
\left\Vert {\tilde{\mathbf{v}}}1_{\bar{\mathcal{R}}_{\hat{\varepsilon}}%
}\right\Vert }\nonumber\\
&  \leq\left(  N\varepsilon^{2}\right)  ^{1/2}{\left\Vert {\tilde{\mathbf{v}}%
}1_{\bar{\mathcal{R}}_{\hat{\varepsilon}}}\right\Vert }\nonumber\\
&  \leq2{\tilde{\mathbf{v}}}^{T}{\tilde{\mathbf{v}}}1_{\bar{\mathcal{R}}%
_{\hat{\varepsilon}}}\text{.} \label{A.Lma1P.1}%
\end{align}
Then by condition \textbf{C1}, we get%
\begin{align}
{\left\langle {\tilde{\mathbf{v}}}^{T}{\tilde{\mathbf{v}}1_{\bar{\mathcal{R}%
}_{\hat{\varepsilon}}}}\right\rangle _{\mathbf{r}|\mathbf{x}}}  &
\leq{\left\langle \frac{\left\Vert {\tilde{\mathbf{v}}}\right\Vert ^{4}%
}{\left(  \hat{\varepsilon}\sqrt{N}\right)  ^{2}}\right\rangle _{\mathbf{r}%
|\mathbf{x}}}\nonumber\\
&  \leq N^{-1}{\left(  \hat{\varepsilon}_{0}\right)  ^{-2}\left\langle
\left\Vert {\tilde{\mathbf{v}}}\right\Vert ^{4}\right\rangle _{\mathbf{r}%
|\mathbf{x}}}=O\left(  N^{-1}\right)  \text{,} \label{A.Lma1P.vv<}%
\end{align}
where $\hat{\varepsilon}_{0}$ is a positive constant and $\hat{\varepsilon
}_{0}\in\left[  \min\hat{\varepsilon}\left(  \mathbf{r}\right)
\text{,\thinspace}\max\hat{\varepsilon}\left(  \mathbf{r}\right)  \right]  $.
By (\ref{A.yy}), (\ref{A.vv=}) and (\ref{A.GammaEps}), we get%
\begin{align}
{\hat{\Gamma}}_{\omega}  &  \geq\left\langle \left\langle \ln\left(
\int_{{\mathcal{Y}}_{\varepsilon}}\exp\left(  \mathbf{y}^{T}\mathbf{\tilde{v}%
}-\dfrac{1}{2}\left(  1+\epsilon\right)  \mathbf{y}^{T}\mathbf{y}+\dfrac{1}%
{2}\mathbf{y}^{T}\mathbf{B}_{0}\mathbf{y}\right)  d\mathbf{y}\right)
\right\rangle _{\mathbf{r}|\mathbf{x}}\right\rangle _{\mathbf{x}}\nonumber\\
&  \geq\left\langle \left\langle \ln\left(  \int_{{\mathcal{Z}_{\hat
{\varepsilon}}}}\exp\left(  {\dfrac{1}{2}}\left(  \mathbf{z}+\frac
{{\tilde{\mathbf{v}}1_{\mathcal{R}_{\hat{\varepsilon}}}}}{1+\epsilon}\right)
^{T}\mathbf{B}_{0}\left(  \mathbf{z}+\frac{{\tilde{\mathbf{v}}1_{\mathcal{R}%
_{\hat{\varepsilon}}}}}{1+\epsilon}\right)  \right)  \phi_{\hat{\varepsilon}%
}{(\mathbf{z})}d\mathbf{z}\right)  \right\rangle _{\mathbf{r}|\mathbf{x}%
}\right\rangle _{\mathbf{x}}\nonumber\\
&  +\left\langle \left\langle \ln\left(  \Psi_{\hat{\varepsilon}}\right)
+{\dfrac{{\tilde{\mathbf{v}}^{T}\tilde{\mathbf{v}}}}{2\left(  1+\epsilon
\right)  ^{2}}-{\dfrac{5{\tilde{\mathbf{v}}^{T}\tilde{\mathbf{v}}%
1_{\bar{\mathcal{R}}_{\hat{\varepsilon}}}}}{2\left(  1+\epsilon\right)  ^{2}}%
}}\right\rangle _{\mathbf{r}|\mathbf{x}}\right\rangle _{\mathbf{x}}\nonumber\\
&  \geq{\dfrac{1}{2}}\left\langle \left\langle \left(  \int_{{\mathcal{Z}%
_{\hat{\varepsilon}}}}\left(  \mathbf{z}+\frac{{\tilde{\mathbf{v}%
}1_{\mathcal{R}_{\hat{\varepsilon}}}}}{1+\epsilon}\right)  ^{T}\mathbf{B}%
_{0}\left(  \mathbf{z}+\frac{{\tilde{\mathbf{v}}1_{\mathcal{R}_{\hat
{\varepsilon}}}}}{1+\epsilon}\right)  \phi_{\hat{\varepsilon}}{(\mathbf{z}%
)}d\mathbf{z}\right)  \right\rangle _{\mathbf{r}|\mathbf{x}}\right\rangle
_{\mathbf{x}}\nonumber\\
&  +\left\langle \left\langle \ln\left(  \Psi_{\hat{\varepsilon}}\right)
\right\rangle _{\mathbf{r}|\mathbf{x}}\right\rangle _{\mathbf{x}}%
+{\dfrac{K+\zeta}{2\left(  1+\epsilon\right)  ^{2}}}+O\left(  N^{-1}\right)
\text{,} \label{A.Lma1P.GamEps'}%
\end{align}
where $\mathbf{z}={\mathbf{y}-\tilde{\mathbf{v}}}${$1_{\mathcal{R}%
_{\hat{\varepsilon}}(\mathbf{x})}$, }the last step in (\ref{A.Lma1P.GamEps'})
follows from Jensen's inequality, and
\begin{equation}
\left\{
\begin{array}
[c]{l}%
\phi_{\hat{\varepsilon}}{(\mathbf{z})=\Psi}_{\hat{\varepsilon}}^{-1}%
{\exp\left(  -\dfrac{1+\epsilon}{2}\mathbf{z}^{T}\mathbf{z}\right)  }%
\text{,}\\
\Psi_{\hat{\varepsilon}}={\int_{{\mathcal{Z}_{\hat{\varepsilon}}}}\exp\left(
-\dfrac{1+\epsilon}{2}\mathbf{z}^{T}\mathbf{z}\right)  d}\mathbf{z}\text{.}%
\end{array}
\right.  \label{A.Lma1P.Def}%
\end{equation}
Integrating by parts yields%
\begin{equation}
{\left\langle {1_{{\mathcal{\bar{Z}}_{\hat{\varepsilon}}}}}\right\rangle
_{\mathbf{z}}}={\int_{\bar{\mathcal{Z}}_{\hat{\varepsilon}}}}\left(
\frac{1+\epsilon}{{2\pi}}\right)  ^{K/2}{\exp\left(  -\dfrac{1+\epsilon}%
{2}\mathbf{z}^{T}\mathbf{z}\right)  d\mathbf{z}}={O\left(  N^{-K/2}%
e^{-N\delta}\right)  } \label{A.Lma1P.GauTail}%
\end{equation}
and
\begin{equation}
\left(  \frac{{2\pi}}{1+\epsilon}\right)  ^{K/2}\geq\Psi_{\hat{\varepsilon}%
}\geq\left(  \frac{{2\pi}}{1+\epsilon}\right)  ^{K/2}{\left(  1-{O\left(
N^{-K/2}e^{-N\delta}\right)  }\right)  } \label{A.Lma1P.GauTail.1}%
\end{equation}
for some $\delta>0$.

Then from (\ref{A.Lma1P.GamEps'}), we get
\begin{align}
&  \left\langle \left\langle \left(  \int_{{\mathcal{Z}_{\hat{\varepsilon}}}%
}\left(  \mathbf{z}+\frac{{\tilde{\mathbf{v}}1_{\mathcal{R}_{\hat{\varepsilon
}}}}}{1+\epsilon}\right)  ^{T}\mathbf{B}_{0}\left(  \mathbf{z}+\frac
{{\tilde{\mathbf{v}}1_{\mathcal{R}_{\hat{\varepsilon}}}}}{1+\epsilon}\right)
\phi_{\hat{\varepsilon}}{(\mathbf{z})}d\mathbf{z}\right)  \right\rangle
_{\mathbf{r}|\mathbf{x}}\right\rangle _{\mathbf{x}}\nonumber\\
&  ={\left(  \dfrac{{2\pi}}{1+\epsilon}\right)  }^{K/2}\Psi_{\hat{\varepsilon
}}^{-1}\left\langle \left\langle \left\langle \mathbf{z}{^{T}}\mathbf{B}%
_{0}\mathbf{z}{1_{{\mathcal{Z}_{\hat{\varepsilon}}}}}\right\rangle
_{\mathbf{z}}+\frac{{\tilde{\mathbf{v}}^{T}}\mathbf{B}_{0}^{2}{\tilde
{\mathbf{v}}1_{{\mathcal{Z}_{\hat{\varepsilon}}}}1_{\mathcal{R}_{\hat
{\varepsilon}}}}}{\left(  1+\epsilon\right)  ^{2}}\right\rangle _{\mathbf{r}%
|\mathbf{x}}\right\rangle _{\mathbf{x}}\nonumber\\
&  \geq{\left(  \dfrac{{2\pi}}{1+\epsilon}\right)  }^{K/2}\Psi_{\hat
{\varepsilon}}^{-1}\left\langle \left\langle \left\langle \mathbf{z}{^{T}%
}\mathbf{B}_{0}\mathbf{z}{1_{{\mathcal{Z}_{\hat{\varepsilon}}}}}\right\rangle
_{\mathbf{z}}\right\rangle _{\mathbf{r}|\mathbf{x}}\right\rangle _{\mathbf{x}%
}\geq O\left(  N^{-1}\right)  \text{,} \label{A.Lma1P.2}%
\end{align}
where%
\begin{equation}
\left\{
\begin{array}
[c]{l}%
{\left\langle \cdot\right\rangle _{\mathbf{z}}=\int_{%
\mathbb{R}
^{K}}}\left(  {\cdot}\right)  {\phi}_{0}\left(  \mathbf{z}\right)
{d}\mathbf{z}\text{,}\\
{\phi}_{0}{(\mathbf{z})=\left(  \dfrac{1+\epsilon}{{2\pi}}\right)  }%
^{K/2}{\exp\left(  -\dfrac{1+\epsilon}{2}\mathbf{z}^{T}\mathbf{z}\right)
}\text{.}%
\end{array}
\right.  \label{A.LmaP.NormalDef}%
\end{equation}
Here notice that%
\begin{equation}
{\left(  \dfrac{{2\pi}}{1+\epsilon}\right)  }^{K/2}\Psi_{\hat{\varepsilon}%
}^{-1}=1+{O\left(  N^{-K/2}e^{-N\alpha}\right)  } \label{A.Lma1P.3}%
\end{equation}
and%
\begin{align}
\left\langle \left\langle \left\langle \mathbf{z}{^{T}}\mathbf{B}%
_{0}\mathbf{z}{1_{{\mathcal{Z}_{\hat{\varepsilon}}}}}\right\rangle
_{\mathbf{z}}\right\rangle _{\mathbf{r}|\mathbf{x}}\right\rangle
_{\mathbf{x}}  &  =-\left\langle \left\langle \left\langle \mathbf{z}{^{T}%
}\mathbf{B}_{0}\mathbf{z}{1_{{\mathcal{\bar{Z}}_{\hat{\varepsilon}}}}%
}\right\rangle _{\mathbf{z}}\right\rangle _{\mathbf{r}|\mathbf{x}%
}\right\rangle _{\mathbf{x}}\nonumber\\
&  \geq-\left\langle \left\langle {\left\Vert \mathbf{B}_{0}\right\Vert ^{2}%
}\right\rangle _{\mathbf{r}|\mathbf{x}}^{1/2}\left\langle \left\langle
\left\Vert \mathbf{z}\right\Vert ^{4}{1_{{\mathcal{\bar{Z}}_{\hat{\varepsilon
}}}}}\right\rangle _{\mathbf{z}}\right\rangle _{\mathbf{r}|\mathbf{x}}%
^{1/2}\right\rangle _{\mathbf{x}}\nonumber\\
&  =O\left(  N^{-1}\right)  \text{.} \label{A.Lma1P.3a}%
\end{align}
Hence, from the consideration above, we find%
\begin{equation}
{{\hat{\Gamma}}_{\omega}}\geq\frac{K}{2}\ln\left(  \frac{{2\pi}}{1+\epsilon
}\right)  +{\dfrac{K}{2\left(  1+\epsilon\right)  ^{2}}}+O\left(
N^{-1}\right)  \text{.} \label{A.Lma1P.GamEps'>}%
\end{equation}

Since $\epsilon$ is arbitrary, let it go to zero. Thus, combining
(\ref{A.GammaEps}) and (\ref{A.Lma1P.GamEps'>}) yields
\begin{equation}
{\Gamma_{\omega}}=-\left\langle \dfrac{1}{2}\ln\left(  \det\left(
\frac{\mathbf{G}(\mathbf{x})}{{2\pi e}}\right)  \right)  \right\rangle
_{\mathbf{x}}+O\left(  N^{-1}\right)  \text{.} \label{A.Lma1P.GamEps=O}%
\end{equation}

Considering
\begin{equation}
{\left\langle \left\langle \ln\dfrac{p(\mathbf{r})}{p\left(  \mathbf{r}%
|\mathbf{x}\right)  p(\mathbf{x})}\right\rangle _{\mathbf{r}|\mathbf{x}%
}\right\rangle _{\mathbf{x}}}\geq\Gamma_{\omega}\text{,} \label{A.Lma1P.Last}%
\end{equation}
and combining (\ref{MI1}) and (\ref{A.Lma1P.GamEps=O}), we immediately get Eq.
(\ref{Lma1}).

\label{A.Lma1Pa}On the other hand, by conditions (\ref{Lma1.1a}) and
(\ref{Lma1.1b}), we have%
\begin{equation}
\left\{
\begin{array}
[c]{l}%
{\left\langle {\tilde{\mathbf{v}}}^{T}{\tilde{\mathbf{v}}1_{\bar{\mathcal{R}%
}_{\hat{\varepsilon}}}}\right\rangle _{\mathbf{r}|\mathbf{x}}}\leq
{\left\langle \dfrac{\left\Vert {\tilde{\mathbf{v}}}\right\Vert _{2}^{2+2\tau
}}{\left(  \hat{\varepsilon}\sqrt{N}\right)  ^{2\tau}}\right\rangle
_{\mathbf{r}|\mathbf{x}}}\leq N^{-\tau}\left(  \hat{\varepsilon}_{0}\right)
^{-2\tau}{\left\langle \left\Vert {\tilde{\mathbf{v}}}\right\Vert ^{2+2\tau
}\right\rangle _{\mathbf{r}|\mathbf{x}}}=o\left(  1\right) \\
\left\langle \left\langle \left\langle \mathbf{z}{^{T}}\mathbf{B}%
_{0}\mathbf{z}{1_{{\mathcal{Z}_{\hat{\varepsilon}}}}}\right\rangle
_{\mathbf{z}}\right\rangle _{\mathbf{r}|\mathbf{x}}\right\rangle _{\mathbf{x}%
}\geq-\left\langle \left\langle {\left\Vert \mathbf{B}_{0}\right\Vert ^{2}%
}\right\rangle _{\mathbf{r}|\mathbf{x}}^{1/2}\left\langle \left\langle
\left\Vert \mathbf{z}\right\Vert ^{4}{1_{{\mathcal{\bar{Z}}_{\hat{\varepsilon
}}}}}\right\rangle _{\mathbf{z}}\right\rangle _{\mathbf{r}|\mathbf{x}}%
^{1/2}\right\rangle _{\mathbf{x}}=o\left(  1\right)
\end{array}
\right.  \text{.} \label{A.Lma1Pa.1}%
\end{equation}
Similarly we can get (\ref{Lma1.2}). This completes the proof of \textbf{Lemma
\ref{Lemma 1}}.%
\qed

\subsection{Proof of Lemma \ref{Lemma 2}}

\label{A.Lma2P} Define the sets%
\begin{equation}
\Omega_{\epsilon}(\mathbf{x})=\left\{  \mathbf{r}\in{\mathcal{R}}:{\mathbf{y}%
}^{T}\mathbf{B}_{0}{\mathbf{y}}<\epsilon\left\Vert {\mathbf{y}}\right\Vert
^{2}\text{, }\forall{\mathbf{y}}\in%
\mathbb{R}
^{K}\right\}  \label{A.Lma2P.1A}%
\end{equation}
and
\begin{equation}
\Theta_{\epsilon}(\mathbf{x})=\left\{  \mathbf{r}\in{\mathcal{R}}%
:\int_{{{\mathcal{\bar{X}}}}_{\varepsilon}(\mathbf{x})}\dfrac{p(\mathbf{r}%
|\mathbf{\hat{x}})p(\mathbf{\hat{x}})}{p(\mathbf{r}|\mathbf{x})p(\mathbf{x}%
)}dx^{\prime}<\epsilon\det\left(  \mathbf{G}(\mathbf{x})\right)
^{-1/2}\right\}  \text{,} \label{A.Lma2P.1B}%
\end{equation}
where ${{\mathcal{\bar{X}}}}_{\varepsilon}(\mathbf{x})={{\mathcal{X}}%
}-{{\mathcal{X}}}_{\varepsilon}(\mathbf{x})$, assuming $\epsilon\in\left(
0\text{,\thinspace}1/2\right)  $ and $p(\mathbf{x})>0$.

Then by Markov's inequality, we have
\begin{equation}
\left\langle 1_{\bar{\Omega}_{\epsilon}}\right\rangle _{\mathbf{r}|\mathbf{x}%
}\leq\mathbb{P}_{\mathbf{r}|\mathbf{x}}\left\{  \left\Vert \mathbf{B}%
_{0}\right\Vert ^{2}\geq\epsilon^{2}\right\}  \leq\epsilon^{-2}\left\langle
\left\Vert \mathbf{B}_{0}\right\Vert ^{2}\right\rangle _{\mathbf{r}%
|\mathbf{x}}=O\left(  N^{-1}\right)  \text{,} \label{A.Lma2P.1C}%
\end{equation}
and by (\ref{C2.b}),%
\begin{align}
\left\langle 1_{\bar{\Theta}_{\epsilon}}\right\rangle _{\mathbf{r}%
|\mathbf{x}}  &  =\mathbb{P}_{\mathbf{r}|\mathbf{x}}\left\{  \int%
_{{{\mathcal{\bar{X}}}}_{\varepsilon}(\mathbf{x})}\dfrac{p(\mathbf{r}%
|\mathbf{\hat{x}})p(\mathbf{\hat{x}})}{p(\mathbf{r}|\mathbf{x})p(\mathbf{x}%
)}d\mathbf{\hat{x}}\geq\epsilon\det\left(  \mathbf{G}(\mathbf{x})\right)
^{-1/2}\right\} \nonumber\\
&  =\mathbb{P}_{\mathbf{r}|\mathbf{x}}\left\{  {\det}\left(  \mathbf{G}\left(
\mathbf{x}\right)  \right)  ^{1/2}\int_{{{\mathcal{\bar{X}}}}_{\hat{\omega}%
}(\mathbf{x})}p(\mathbf{\hat{x}}|\mathbf{r})d\hat{\mathbf{x}}>\epsilon
p(\mathbf{x}|\mathbf{r})\right\} \nonumber\\
&  =O\left(  N^{-\eta}\right)  \text{.} \label{A.Lma2P.1D}%
\end{align}

Consider the following equality,%
\begin{equation}
\left\langle {\ln\dfrac{p(\mathbf{r})}{p\left(  \mathbf{r}|\mathbf{x}\right)
p(\mathbf{x})}}\right\rangle _{\mathbf{r}|\mathbf{x}}=\left\langle
1_{\Theta_{\epsilon}}{\ln\dfrac{p(\mathbf{r})}{p\left(  \mathbf{r}%
|\mathbf{x}\right)  p(\mathbf{x})}}\right\rangle _{\mathbf{r}|\mathbf{x}%
}+\left\langle 1_{\bar{\Theta}_{\epsilon}}{\ln\dfrac{p(\mathbf{r})}{p\left(
\mathbf{r}|\mathbf{x}\right)  p(\mathbf{x})}}\right\rangle _{\mathbf{r}%
|\mathbf{x}}\text{.} \label{A.Lma2P.2}%
\end{equation}
For the last term in (\ref{A.Lma2P.2}), Jensen's inequality implies that%
\begin{equation}
\left\langle \left\langle 1_{\bar{\Theta}_{\epsilon}}{\ln\dfrac{p(\mathbf{r}%
)}{p\left(  \mathbf{r}|\mathbf{x}\right)  p(\mathbf{x})}}\right\rangle
_{\mathbf{r}|\mathbf{x}}\right\rangle _{\mathbf{x}}\leq\left\langle
\left\langle 1_{\bar{\Theta}_{\epsilon}}\right\rangle _{\mathbf{r}|\mathbf{x}%
}\right\rangle _{\mathbf{x}}\ln\frac{1}{\left\langle \left\langle
1_{\bar{\Theta}_{\epsilon}}\right\rangle _{\mathbf{r}|\mathbf{x}}\right\rangle
_{\mathbf{x}}}=o\left(  N^{-1}\right)  \text{.} \label{A.Lma2P.2A}%
\end{equation}
For the first term in (\ref{A.Lma2P.2}), it follows from (\ref{A.Lma2P.1B})
and (\ref{A.yy}) that
\begin{align}
&  \left\langle 1_{\Theta_{\epsilon}}{\ln\dfrac{p(\mathbf{r})}{p\left(
\mathbf{r}|\mathbf{x}\right)  p(\mathbf{x})}}\right\rangle _{\mathbf{r}%
|\mathbf{x}}\nonumber\\
&  \leq\left\langle 1_{\Theta_{\epsilon}}{\ln}\left(  \int_{{{\mathcal{X}}%
}_{\varepsilon}(\mathbf{x})}\exp\left(  L(\mathbf{r}|\hat{\mathbf{x}%
})-L(\mathbf{r}|\mathbf{x})\right)  d\mathbf{\hat{x}}+{\epsilon}\det\left(
\mathbf{G}(\mathbf{x})\right)  ^{-1/2}\right)  \right\rangle _{\mathbf{r}%
|\mathbf{x}}\nonumber\\
&  \leq-\frac{K}{2}\ln\left(  \det\left(  \mathbf{G}(\mathbf{x})\right)
\right) \nonumber\\
&  +\left\langle 1_{\Theta_{\epsilon}}{\ln}\left(  \int_{{\mathcal{Y}%
_{\varepsilon}}}\exp\left(  \mathbf{y}^{T}\tilde{\mathbf{v}}-\dfrac{1}%
{2}\left(  1-\epsilon\right)  \mathbf{y}^{T}\mathbf{y}+\dfrac{1}{2}%
\mathbf{y}^{T}\mathbf{B}_{0}\mathbf{y}\right)  d\mathbf{y}+{\epsilon}\right)
\right\rangle _{\mathbf{r}|\mathbf{x}}\text{.} \label{A.Lma2P.2B}%
\end{align}
The last term (\ref{A.Lma2P.2B}) is upper-bounded by
\begin{align}
&  \left\langle 1_{\Theta_{\epsilon}\cap\Omega_{\epsilon}}{\ln}\left(  \int_{%
\mathbb{R}
^{K}}\exp\left(  \mathbf{y}^{T}\tilde{\mathbf{v}}-\dfrac{1}{2}\left(
1-2\epsilon\right)  \mathbf{y}^{T}\mathbf{y}\right)  d\mathbf{y}+{\epsilon
}\right)  \right\rangle _{\mathbf{r}|\mathbf{x}}\label{A.Lma2P.3a}\\
&  +\left\langle 1_{\Theta_{\epsilon}\cap\bar{\Omega}_{\epsilon}}{\ln}\left(
\int_{%
\mathbb{R}
^{K}}\exp\left(  \mathbf{y}^{T}\tilde{\mathbf{v}}-\dfrac{1}{2}\left(
1-\epsilon\right)  \mathbf{y}^{T}\mathbf{y}+\dfrac{1}{2}\mathbf{y}%
^{T}\mathbf{B}_{0}\mathbf{y}\right)  d\mathbf{y}+{\epsilon}\right)
\right\rangle _{\mathbf{r}|\mathbf{x}}\text{.} \label{A.Lma2P.3b}%
\end{align}
The term (\ref{A.Lma2P.3a}) is equal to
\begin{align}
&  \left\langle 1_{\Theta_{\epsilon}\cap\Omega_{\epsilon}}\ln\left(  \left(
\frac{{2\pi}}{1-2\epsilon}\right)  ^{K/2}\exp\left(  \frac{\mathbf{\tilde{v}%
}^{T}\tilde{\mathbf{v}}}{2\left(  1-2\epsilon\right)  }\right)  +\epsilon
\right)  \right\rangle _{\mathbf{r}|\mathbf{x}}\nonumber\\
&  \leq\left\langle 1_{\Theta_{\epsilon}\cap\Omega_{\epsilon}}\left(
\frac{\tilde{\mathbf{v}}^{T}\tilde{\mathbf{v}}}{2\left(  1-2\epsilon\right)
}+\ln\left(  \left(  \frac{{2\pi}}{1-2\epsilon}\right)  ^{K/2}+\epsilon
\right)  \right)  \right\rangle _{\mathbf{r}|\mathbf{x}}\text{,}
\label{A.Lma2P.4A}%
\end{align}
The term (\ref{A.Lma2P.3b}) is equal to
\begin{subequations}
\begin{align}
&  \left\langle 1_{\Theta_{\epsilon}\cap\bar{\Omega}_{\epsilon}}{\ln}\left(
\left\langle \left(  \frac{{2\pi}}{1-\epsilon}\right)  ^{K/2}\exp\left(
\dfrac{1}{2}\left(  \mathbf{z+}\frac{\tilde{\mathbf{v}}}{1-\epsilon}\right)
^{T}\mathbf{B}_{0}\left(  \mathbf{z+}\frac{\tilde{\mathbf{v}}}{1-\epsilon
}\right)  +\frac{\tilde{\mathbf{v}}^{T}\tilde{\mathbf{v}}}{2\left(
1-\epsilon\right)  }\right)  \right\rangle _{\mathbf{z}}+{\epsilon}\right)
\right\rangle _{\mathbf{r}|\mathbf{x}}\nonumber\\
&  \leq\left\langle 1_{\Theta_{\epsilon}\cap\bar{\Omega}_{\epsilon}}\left(
\frac{K}{2}\ln\left(  \frac{{2\pi}}{1-\epsilon}\right)  +\frac{\mathbf{\tilde
{v}}^{T}\tilde{\mathbf{v}}}{2\left(  1-\epsilon\right)  }+\frac{\mathbf{\tilde
{v}}^{T}\mathbf{B}_{0}^{2}\tilde{\mathbf{v}}}{2\left(  1-\epsilon\right)
^{2}}+\frac{\tilde{\mathbf{v}}^{T}\mathbf{B}_{0}^{2}\tilde{\mathbf{v}}%
}{\left(  1-\epsilon\right)  ^{3}}\right)  \right\rangle _{\mathbf{r}%
|\mathbf{x}}\label{A.Lma2P.4B1}\\
&  +\left\langle 1_{\Theta_{\epsilon}\cap\bar{\Omega}_{\epsilon}}{\ln}\left(
\left\langle \exp\left(  \dfrac{1}{2}\mathbf{z}^{T}\mathbf{B}_{0}%
\mathbf{z+}\frac{\mathbf{z}^{T}\mathbf{B}_{0}\tilde{\mathbf{v}}}{1-\epsilon
}\mathbf{-}\frac{\tilde{\mathbf{v}}^{T}\mathbf{B}_{0}^{2}\tilde{\mathbf{v}}%
}{\left(  1-\epsilon\right)  ^{3}}\right)  \right\rangle _{\mathbf{z}%
}+{\epsilon}\left(  \frac{1-\epsilon}{{2\pi}}\right)  ^{K/2}\right)
\right\rangle _{\mathbf{r}|\mathbf{x}}\text{,} \label{A.Lma2P.4B2}%
\end{align}
where%
\end{subequations}
\begin{equation}
\left\{
\begin{array}
[c]{l}%
{\left\langle \cdot\right\rangle _{\mathbf{z}}=\int_{%
\mathbb{R}
^{K}}}\left(  {\cdot}\right)  {\phi}_{1}\left(  \mathbf{z}\right)
{d}\mathbf{z}\\
{\phi}_{1}{(\mathbf{z})=\left(  \dfrac{1-\epsilon}{{2\pi}}\right)  }%
^{K/2}{\exp\left(  -\dfrac{1-\epsilon}{2}\mathbf{z}^{T}\mathbf{z}\right)  }%
\end{array}
\right.  \text{.} \label{A.Lma2P.4C}%
\end{equation}

Notice that
\begin{equation}
\left\langle 1_{\Theta_{\epsilon}\cap\bar{\Omega}_{\epsilon}}\right\rangle
_{\mathbf{r}|\mathbf{x}}\leq\left\langle 1_{\bar{\Omega}_{\epsilon}%
}\right\rangle _{\mathbf{r}|\mathbf{x}}=O\left(  N^{-1}\right)
\label{A.Lma2P.5A}%
\end{equation}
and%
\begin{equation}
\left\langle 1_{\Theta_{\epsilon}\cap\Omega_{\epsilon}}\right\rangle
_{\mathbf{r}|\mathbf{x}}=1-\left\langle 1_{\bar{\Theta}_{\epsilon}\cup
\bar{\Omega}_{\epsilon}}\right\rangle _{\mathbf{r}|\mathbf{x}}=1+O\left(
N^{-1}\right)  \text{.} \label{A.Lma2P.5A1}%
\end{equation}
Then by (\ref{A.vv1<}), we get%
\begin{align}
&  \left\langle 1_{\Theta_{\epsilon}\cap\bar{\Omega}_{\epsilon}}\left(
\left\langle \exp\left(  \mathbf{z}^{T}\mathbf{B}_{0}\mathbf{z}\right)
\right\rangle _{\mathbf{z}}^{1/2}-1\right)  \right\rangle _{\mathbf{r}%
|\mathbf{x}}\nonumber\\
&  \leq{\left\langle 1_{\Theta_{\epsilon}\cap\bar{\Omega}_{\epsilon}}%
\sum_{m=0}^{\infty}\dfrac{1}{m!}\left\langle \left(  \mathbf{z}{^{T}%
}\mathbf{B}_{0}\mathbf{z}\right)  ^{m}\right\rangle _{\mathbf{z}}\right\rangle
_{\mathbf{r}|\mathbf{x}}^{1/2}}-\left\langle 1_{\Theta_{\epsilon}\cap
\bar{\Omega}_{\epsilon}}\right\rangle _{\mathbf{r}|\mathbf{x}}=O\left(
N^{-1}\right)  \text{,} \label{A.Lma2P.5B}%
\end{align}%
\begin{equation}
\left\langle \tilde{\mathbf{v}}^{T}\tilde{\mathbf{v}}1_{\bar{\Theta}%
_{\epsilon}}\right\rangle _{\mathbf{r}|\mathbf{x}}\leq\left\langle \left\Vert
\tilde{\mathbf{v}}\right\Vert ^{4}\right\rangle _{\mathbf{r}|\mathbf{x}}%
^{1/2}\left\langle 1_{\bar{\Theta}_{\epsilon}}\right\rangle _{\mathbf{r}%
|\mathbf{x}}^{1/2}=O\left(  N^{-1}\right)  \text{,} \label{A.Lma2P.5C}%
\end{equation}
and by (\ref{XiDef}),%
\begin{align}
0  &  \leq\left\langle \tilde{\mathbf{v}}^{T}\mathbf{B}_{0}^{2}{\tilde
{\mathbf{v}}}1_{\Theta_{\epsilon}\cap\bar{\Omega}_{\epsilon}}\right\rangle
_{\mathbf{r}|\mathbf{x}}\leq\left\langle \mathbf{v}^{T}\mathbf{B}_{0}%
^{2}\mathbf{v}\right\rangle _{\mathbf{r}|\mathbf{x}}+O\left(  N^{-1}\right)
\nonumber\\
&  \leq\xi\left\Vert N\mathbf{G}^{-1}\left(  \mathbf{x}\right)  \right\Vert
+O\left(  N^{-1}\right)  =O\left(  N^{-1}\right)  \text{.} \label{A.Lma2P.5D}%
\end{align}
Hence, we have
\begin{align}
&  \left\langle 1_{\Theta_{\epsilon}}\left(  \frac{K}{2}\ln\left(  \frac
{{2\pi}}{1-\epsilon}\right)  +\frac{\tilde{\mathbf{v}}^{T}\tilde{\mathbf{v}}%
}{2\left(  1-\epsilon\right)  }\right)  \right\rangle _{\mathbf{r}|\mathbf{x}%
}\nonumber\\
&  =\left(  \frac{K}{2}\ln\left(  \frac{{2\pi}}{1-\epsilon}\right)
+\frac{K+\zeta}{2\left(  1-\epsilon\right)  }\right)  +O\left(  N^{-1}\right)
\text{,} \label{A.Lma2P.6}%
\end{align}
and by Cauchy--Schwarz inequality and (\ref{A.Lma2P.5B}), the term
(\ref{A.Lma2P.4B2}) is upper bounded by%
\begin{align}
&  \left\langle 1_{\Theta_{\epsilon}\cap\bar{\Omega}_{\epsilon}}{\ln}\left(
\left\langle \exp\left(  \mathbf{z}^{T}\mathbf{B}_{0}\mathbf{z}\right)
\right\rangle _{\mathbf{z}}^{1/2}\left\langle \exp\left(  \frac{2\mathbf{z}%
^{T}\mathbf{B}_{0}\tilde{\mathbf{v}}}{1-\epsilon}\mathbf{-}\frac
{2\tilde{\mathbf{v}}^{T}\mathbf{B}_{0}^{2}\tilde{\mathbf{v}}}{\left(
1-\epsilon\right)  ^{3}}\right)  \right\rangle _{\mathbf{z}}^{1/2}+{\epsilon
}\left(  \frac{1-\epsilon}{{2\pi}}\right)  ^{K/2}\right)  \right\rangle
_{\mathbf{r}|\mathbf{x}}\nonumber\\
&  =\left\langle 1_{\Theta_{\epsilon}\cap\bar{\Omega}_{\epsilon}}{\ln}\left(
\left\langle \exp\left(  \mathbf{z}^{T}\mathbf{B}_{0}\mathbf{z}\right)
\right\rangle _{\mathbf{z}}^{1/2}+{\epsilon}\left(  \frac{1-\epsilon}{{2\pi}%
}\right)  ^{K/2}\right)  \right\rangle _{\mathbf{r}|\mathbf{x}}\nonumber\\
&  \leq\left\langle 1_{\Theta_{\epsilon}\cap\bar{\Omega}_{\epsilon}}\left(
\left\langle \exp\left(  \mathbf{z}^{T}\mathbf{B}_{0}\mathbf{z}\right)
\right\rangle _{\mathbf{z}}^{1/2}+{\epsilon}\left(  \frac{1-\epsilon}{{2\pi}%
}\right)  ^{K/2}-1\right)  \right\rangle _{\mathbf{r}|\mathbf{x}}=O\left(
N^{-1}\right)  \text{.} \label{A.Lma2P.7}%
\end{align}

Since $\epsilon$ is arbitrary, we can let it go to zero. Then taking
everything together, we get%
\begin{equation}
{\left\langle \left\langle \ln\dfrac{p(\mathbf{r})}{p\left(  \mathbf{r}%
|\mathbf{x}\right)  p(\mathbf{x})}\right\rangle _{\mathbf{r}|\mathbf{x}%
}\right\rangle _{\mathbf{x}}\leq}-\left\langle \dfrac{1}{2}\ln\left(
\det\left(  \frac{\mathbf{G}(\mathbf{x})}{{2\pi e}}\right)  \right)
\right\rangle _{\mathbf{x}}+O\left(  N^{-1}\right)  \text{.}
\label{A.Lma2P.last}%
\end{equation}
Putting (\ref{A.Lma2P.last}) into (\ref{MI1}) yields (\ref{Lma2}).

On the other hand, we have%
\begin{align}
&  \left\langle \left\langle {\ln\dfrac{p(\mathbf{r})}{p\left(  \mathbf{r}%
|\mathbf{x}\right)  p(\mathbf{x})}}\right\rangle _{\mathbf{r}|\mathbf{x}%
}\right\rangle _{\mathbf{x}}\nonumber\\
&  =\left\langle \left\langle 1_{\Theta_{\epsilon}\cap\Omega_{\epsilon}}%
{\ln\dfrac{p(\mathbf{r})}{p\left(  \mathbf{r}|\mathbf{x}\right)
p(\mathbf{x})}}\right\rangle _{\mathbf{r}|\mathbf{x}}\right\rangle
_{\mathbf{x}}\label{A.Lma2aP.4a}\\
&  +\left\langle \left\langle 1_{\Theta_{\epsilon}\cap\bar{\Omega}_{\epsilon}%
}{\ln\dfrac{p(\mathbf{r})}{p\left(  \mathbf{r}|\mathbf{x}\right)
p(\mathbf{x})}}\right\rangle _{\mathbf{r}|\mathbf{x}}\right\rangle
_{\mathbf{x}}+\left\langle \left\langle 1_{\bar{\Theta}_{\epsilon}}{\ln
\dfrac{p(\mathbf{r})}{p\left(  \mathbf{r}|\mathbf{x}\right)  p(\mathbf{x})}%
}\right\rangle _{\mathbf{r}|\mathbf{x}}\right\rangle _{\mathbf{x}}\text{.}
\label{A.Lma2aP.4b}%
\end{align}
For term (\ref{A.Lma2aP.4b}), it follows from Jensen's inequality that%
\begin{equation}
\left\langle \left\langle 1_{\bar{\Theta}_{\epsilon}}{\ln\dfrac{p(\mathbf{r}%
)}{p\left(  \mathbf{r}|\mathbf{x}\right)  p(\mathbf{x})}}\right\rangle
_{\mathbf{r}|\mathbf{x}}\right\rangle _{\mathbf{x}}\leq\left\langle
\left\langle 1_{\bar{\Theta}_{\epsilon}}\right\rangle _{\mathbf{r}|\mathbf{x}%
}\right\rangle _{\mathbf{x}}\ln\frac{1}{\left\langle \left\langle
1_{\bar{\Theta}_{\epsilon}}\right\rangle _{\mathbf{r}|\mathbf{x}}\right\rangle
_{\mathbf{x}}}=o\left(  1\right)  \label{A.Lma2aP.5}%
\end{equation}
and%
\begin{equation}
\left\langle \left\langle 1_{\Theta_{\epsilon}\cap\bar{\Omega}_{\epsilon}}%
{\ln\dfrac{p(\mathbf{r})}{p\left(  \mathbf{r}|\mathbf{x}\right)
p(\mathbf{x})}}\right\rangle _{\mathbf{r}|\mathbf{x}}\right\rangle
_{\mathbf{x}}\leq\left\langle \left\langle 1_{\Theta_{\epsilon}\cap\bar
{\Omega}_{\epsilon}}\right\rangle _{\mathbf{r}|\mathbf{x}}\right\rangle
_{\mathbf{x}}\ln\frac{1}{\left\langle \left\langle 1_{\Theta_{\epsilon}%
\cap\bar{\Omega}_{\epsilon}}\right\rangle _{\mathbf{r}|\mathbf{x}%
}\right\rangle _{\mathbf{x}}}=o\left(  1\right)  \text{,} \label{A.Lma2aP.6}%
\end{equation}
where
\begin{equation}
\left\{
\begin{array}
[c]{l}%
\left\langle 1_{\bar{\Omega}_{\epsilon}}\right\rangle _{\mathbf{r}|\mathbf{x}%
}\leq P\left(  \left\Vert \mathbf{B}_{0}\right\Vert ^{2}\geq\epsilon
^{2}\right)  \leq\epsilon^{-2}\left\langle \left\Vert \mathbf{B}%
_{0}\right\Vert ^{2}\right\rangle _{\mathbf{r}|\mathbf{x}}=o\left(  1\right)
\text{,}\\
\left\langle 1_{\Theta_{\epsilon}\cap\bar{\Omega}_{\epsilon}}\right\rangle
_{\mathbf{r}|\mathbf{x}}\leq\left\langle 1_{\bar{\Omega}_{\epsilon}%
}\right\rangle _{\mathbf{r}|\mathbf{x}}=o\left(  1\right)  \text{.}%
\end{array}
\right.  \label{A.Lma2aP.6a}%
\end{equation}

Similarly we can get (\ref{Lma2a}). This completes the proof of \textbf{Lemma
\ref{Lemma 2}}\textit{.}%
\qed

\subsection{Proof of Theorem \ref{Theorem 1}}

\label{A.Thm1P}By \textbf{Lemma \ref{Lemma 1}}\textit{\ }and\textit{\ }%
\textbf{Lemma \ref{Lemma 2}}, we immediately get (\ref{Thm1}). The proof of
(\ref{Thm1.1}) is similar.%
\qed

\subsection{Proof of Theorem \ref{Theorem 1a}}

First, we have
\begin{equation}
\mathbf{G}(\mathbf{x})=\mathbf{J}{^{1/2}}(\mathbf{x})\left(  \mathbf{I}%
_{K}+\boldsymbol{\Psi}(\mathbf{x})\right)  \mathbf{J}{^{1/2}}(\mathbf{x}%
)\text{.} \label{A.Thm1P.0}%
\end{equation}
Since $\mathbf{J}(\mathbf{x})$ and $\mathbf{G}(\mathbf{x})$\ are symmetric and
positive-definite, $\mathbf{I}_{K}+\boldsymbol{\Psi}(\mathbf{x})$ is also
symmetric and positive-definite. The eigendecompositon of $\boldsymbol{\Psi
}(\mathbf{x})$ is given by
\begin{equation}
\boldsymbol{\Psi}(\mathbf{x})=\mathbf{U_{\mathbf{x}}}\boldsymbol{\Lambda
}\mathbf{_{\mathbf{x}}U}_{\mathbf{x}}^{T}\text{,} \label{A.Thm1P.2}%
\end{equation}
where $\mathbf{U}_{\mathbf{x}}\in%
\mathbb{R}
^{K\times K}$ is an orthogonal matrix, and the matrix $\boldsymbol{\Lambda
}_{\mathbf{x}}\in%
\mathbb{R}
^{K\times K}$ is a $K\times K$ diagonal matrix with $K$ nonnegative real
numbers on the diagonal, $\lambda_{1}\geq\lambda_{2}\geq$\textrm{,\thinspace
}$\cdots$\textrm{,\thinspace}$\geq\lambda_{K}>-1$. Then we have
\begin{equation}
\left\langle {\mathrm{Tr}}\left(  \boldsymbol{\Lambda}_{\mathbf{x}}\right)
\right\rangle _{\mathbf{x}}=\left\langle {\mathrm{Tr}}\left(  \boldsymbol{\Psi
}(\mathbf{x})\right)  \right\rangle _{\mathbf{x}}=\left\langle {\mathrm{Tr}%
}\left(  \mathbf{P}(\mathbf{x})\mathbf{J}^{-1}(\mathbf{x})\right)
\right\rangle _{\mathbf{x}}=\varsigma\label{A.Thm1P.3}%
\end{equation}
and
\begin{equation}
{\left\langle \ln\left(  \det\left(  \mathbf{I}_{K}+\boldsymbol{\Psi
}(\mathbf{x})\right)  \right)  \right\rangle _{\mathbf{x}}}=\left\langle
\mathrm{Tr}\left(  \ln\left(  \mathbf{I}_{K}+\boldsymbol{\Lambda}_{\mathbf{x}%
}\right)  \right)  \right\rangle _{\mathbf{x}}\leq\left\langle \mathrm{Tr}%
\left(  \boldsymbol{\Lambda}_{\mathbf{x}}\right)  \right\rangle _{\mathbf{x}%
}=\varsigma\text{.} \label{A.Thm1P.4}%
\end{equation}
Notice that $\ln(1+x)\leq x$ for $\forall x\in(-1$\textrm{,\thinspace}%
$\infty)$. It follows from (\ref{A.Thm1P.0}) and (\ref{A.Thm1P.4}) that
\begin{equation}
\left\langle \ln\left(  \det\left(  \mathbf{G}(\mathbf{x})\right)  \right)
\right\rangle _{\mathbf{x}}-\left\langle \ln\left(  \det\left(  \mathbf{J}%
(\mathbf{x})\right)  \right)  \right\rangle _{\mathbf{x}}=\left\langle
\ln\left(  \det\left(  \mathbf{I}_{K}+\boldsymbol{\Psi}(\mathbf{x})\right)
\right)  \right\rangle _{\mathbf{x}}\leq\varsigma\text{.} \label{A.Thm1P.5}%
\end{equation}
From (\ref{IG}), (\ref{IF}) and (\ref{A.Thm1P.5}), we obtain (\ref{Thm1a.2}).

If $\mathbf{P}(\mathbf{x})$ is positive-semidefinite, then $\lambda_{1}%
\geq\lambda_{2}\geq$\textrm{,\thinspace}$\cdots$\textrm{,\thinspace}%
$\geq\lambda_{K}\geq0$,\ $\varsigma\geq0$ and\ ${\left\langle \ln\left(
\det\left(  \mathbf{I}_{K}+\boldsymbol{\Psi}(\mathbf{x})\right)  \right)
\right\rangle _{\mathbf{x}}\geq0}$. Hence we can get (\ref{Thm1a.3}).

On the other hand, it follows from (\ref{Thm1a.3a}), (\ref{A.Thm1P.4}) and the
power mean inequality that%
\begin{equation}
\left\vert \varsigma\right\vert \leq\left\langle \sum\nolimits_{k=1}%
^{K}\left\vert \lambda_{k}\right\vert \right\rangle _{\mathbf{x}}\leq\sqrt
{K}\left\langle \left(  \sum\nolimits_{k=1}^{K}\lambda_{k}^{2}\right)
^{1/2}\right\rangle _{\mathbf{x}}=\sqrt{K}\left\langle \left\Vert
\boldsymbol{\Psi}(\mathbf{x})\right\Vert \right\rangle _{\mathbf{x}}=\sqrt
{K}\varsigma_{1}=O(N^{-\beta})\text{.} \label{A.Thm1P.5a}%
\end{equation}
Let $\lambda_{k}^{-}=\min\left(  0\mathrm{,\,}\lambda_{k}\right)  $ for
$\forall k\in\left\{  1\mathrm{,\,}2\mathrm{,\,}\cdots\mathrm{,\,}K\right\}
$, then%
\begin{equation}
{\left\langle \sum\nolimits_{k=1}^{K}\ln\left(  1+\lambda_{k}^{-}\right)
\right\rangle _{\mathbf{x}}\leq\left\langle \ln\left(  \det\left(
\mathbf{I}_{K}+\boldsymbol{\Psi}(\mathbf{x})\right)  \right)  \right\rangle
_{\mathbf{x}}}\text{.} \label{A.Thm1P.5b}%
\end{equation}
Notice that $-1<\lambda_{k}^{-}\leq0$, then by (\ref{A.Thm1P.5a}), we have
\begin{equation}
{\left\langle \sum\nolimits_{k=1}^{K}\ln\left(  1+\lambda_{k}^{-}\right)
\right\rangle _{\mathbf{x}}}=\left\langle \sum\nolimits_{m=1}^{\infty}%
\dfrac{-1}{m}\sum\nolimits_{k=1}^{K}\left(  -\lambda_{k}^{-}\right)
^{m}\right\rangle _{\mathbf{x}}=O(N^{-\beta})\text{.} \label{A.Thm1P.6}%
\end{equation}
From (\ref{IG}), (\ref{IF}), (\ref{A.Thm1P.5}), (\ref{A.Thm1P.5b}) and
(\ref{A.Thm1P.6}), we immediately get (\ref{Thm1a.4}). This completes the
proof of \textbf{Theorem \ref{Theorem 1a}}.%
\qed

\subsection{Proof of Theorem \ref{Theorem 2}}

\label{A.Thm2P}Considering the change of variables theorem, for any
real-valued function $f$ and invertible transformation $\mathbf{T}$, we have
\begin{equation}
\int_{{{\mathcal{\tilde{X}}}}}f(\mathbf{\tilde{x}})d\mathbf{\tilde{x}}%
=\int_{{{\mathcal{X}}}}f\left(  \mathbf{T}(\mathbf{x})\right)  \left\vert
\det\left(  D\mathbf{T}(\mathbf{x})\right)  \right\vert d\mathbf{x}\text{,}
\label{A.Thm2P.1}%
\end{equation}
and for $p(\mathbf{x})$ and $p(\mathbf{\tilde{x}})$,
\begin{equation}
\left.  p(\mathbf{\tilde{x}})\right\vert _{\mathbf{\tilde{x}}=T(\mathbf{x}%
)}=\left\vert \det\left(  D\mathbf{T}(\mathbf{x})\right)  \right\vert
^{-1}p(\mathbf{x})\text{.} \label{A.Thm2P.2}%
\end{equation}
Then, it follows from (\ref{prxT}), (\ref{A.Thm2P.1}) and (\ref{A.Thm2P.2})
that%
\begin{equation}
\left\{
\begin{split}
{p(\mathbf{r})}  &  ={\int_{{{\mathcal{X}}}}p(\mathbf{r}|\mathbf{x}%
)p(\mathbf{x})d}\mathbf{x}={\int_{\mathcal{\tilde{X}}}p(\mathbf{r}%
|\mathbf{\tilde{x}})p(\mathbf{\tilde{x}})d}\mathbf{\tilde{x}}\text{,}\\
{H(\tilde{X})}  &  ={-\int_{{\mathcal{\tilde{X}}}}p(\mathbf{\tilde{x}})\ln
p(\mathbf{\tilde{x}})d}\mathbf{\tilde{x}}\\
&  ={-\int_{{{\mathcal{X}}}}p(\mathbf{x})\ln\left(  p(\mathbf{x})\left\vert
\det\left(  D\mathbf{T}(\mathbf{x})\right)  \right\vert ^{-1}\right)
d}\mathbf{x}\\
&  ={H(X)+\int_{{{\mathcal{X}}}}p(\mathbf{x})\ln\left\vert \det\left(
D\mathbf{T}(\mathbf{x})\right)  \right\vert d}\mathbf{x}\text{,}\\
\mathbf{G}(\mathbf{x})  &  =D\mathbf{T}(\mathbf{x})^{T}\mathbf{G}%
(\mathbf{\tilde{x}})D\mathbf{T}(\mathbf{x})\text{.}%
\end{split}
\right.  \label{A.Thm2P.3}%
\end{equation}
Substituting (\ref{A.Thm2P.2}) and (\ref{A.Thm2P.3}) into (\ref{MI}), we can
directly obtain (\ref{IT}). Moreover, if $p(\mathbf{\tilde{x}})$ and
$p(\mathbf{r}|\mathbf{\tilde{x}})$ fulfill conditions \textbf{C1,\ C2}
and\textbf{ }$\xi=O\left(  N^{-1}\right)  $, then by \textbf{Theorem
\ref{Theorem 1}}, we immediately obtain Eq. (\ref{ITP}). This completes the
proof of \textbf{Theorem \ref{Theorem 2}}.
\qed

\subsection{Proof of Corollary \ref{Corollary 2}}

\label{Cly2P}It follows from (\ref{I_Gau}) and\textbf{\textit{\ }Theorem
\ref{Theorem 2}} that
\begin{equation}
I_{G}=I_{G_{+}}=I(X;R)=I(Y;R)=\frac{1}{2}\ln\left(  \det\left(  \frac{1}{2\pi
e}\left(  \mathbf{AA}^{T}+\boldsymbol{\Sigma}_{\mathbf{f}}^{-1}\right)
\right)  \right)  +H(Y) \label{A.Cly2P.6}%
\end{equation}
and
\begin{equation}
H(Y)=\dfrac{1}{2}\ln\left(  \det\left(  2\pi e\boldsymbol{\Sigma}_{\mathbf{f}%
}\right)  \right)  =H(X)+\left\langle \ln\left\vert \det\left(  \mathbf{D}%
(\mathbf{x})\right)  \right\vert \right\rangle _{\mathbf{x}}\text{.}
\label{A.Cly2P.7}%
\end{equation}
Here notice that
\begin{align}
\mathbf{J}{(\mathbf{x})}  &  ={\left\langle \dfrac{\partial\ln p(\mathbf{r}%
|\mathbf{x})}{\partial\mathbf{x}}\dfrac{\partial\ln p(\mathbf{r}|\mathbf{x}%
)}{\partial\mathbf{x}^{T}}\right\rangle _{\mathbf{r}|\mathbf{x}}}\nonumber\\
&  ={\left\langle \dfrac{\partial\mathbf{y}^{T}}{\partial\mathbf{x}}%
\dfrac{\partial\ln p(\mathbf{r}|\mathbf{y})}{\partial\mathbf{y}}%
\dfrac{\partial\ln p(\mathbf{r}|\mathbf{y})}{\partial\mathbf{y}^{T}}%
\dfrac{\partial\mathbf{y}}{\partial\mathbf{x}^{T}}\right\rangle _{\mathbf{r}%
|\mathbf{y}}}\nonumber\\
&  ={\mathbf{D}(\mathbf{x})^{T}\mathbf{AA}^{T}\mathbf{D}(\mathbf{x})}
\label{A.Cly2P.8}%
\end{align}
and
\begin{equation}
\mathbf{P}(\mathbf{x})=-\frac{\partial^{2}\ln p(\mathbf{x})}{\partial
\mathbf{x}\partial\mathbf{x}^{T}}=-\frac{\partial\mathbf{y}^{T}}%
{\partial\mathbf{x}}\frac{\partial^{2}\ln p(\mathbf{y})}{\partial
\mathbf{y}\partial\mathbf{y}^{T}}\frac{\partial\mathbf{y}}{\partial
\mathbf{x}^{T}}=\mathbf{D}(\mathbf{x})^{T}\boldsymbol{\Sigma}_{\mathbf{f}%
}^{-1}\mathbf{D}(\mathbf{x})\text{.} \label{A.Cly2P.9}%
\end{equation}
Hence combining (\ref{A.Cly2P.6})--(\ref{A.Cly2P.9}), we can immediately
obtain (\ref{Cly2.4}). This completes the proof of \textbf{Corollary
\ref{Corollary 2}}.
\qed

\subsection{Proof of Theorem \ref{Theorem 3}}

\label{A.Thm3P}First, we have
\begin{align}
&  {\left\langle \ln\left(  \det\left(  \dfrac{\mathbf{G}(\mathbf{x})}{2\pi
e}\right)  \right)  \right\rangle _{\mathbf{x}}}\nonumber\\
&  =\left\langle \ln\left(  \det\left(  \dfrac{\mathbf{G}_{1\text{,\thinspace
}1}\left(  \mathbf{x}\right)  }{2\pi e}\right)  \det\left(  \dfrac{1}{2\pi
e}(\mathbf{G}_{2\text{,\thinspace}2}\left(  \mathbf{x}\right)  -\mathbf{G}%
_{2\text{,\thinspace}1}(\mathbf{x})\mathbf{G}_{1\text{,\thinspace}1}%
^{-1}(\mathbf{x})\mathbf{G}_{1\text{,\thinspace}2}(\mathbf{x}))\right)
\right)  \right\rangle _{\mathbf{x}}\nonumber\\
&  =\left\langle \ln\left(  \det\left(  \dfrac{\mathbf{G}_{1\text{,\thinspace
}1}\left(  \mathbf{x}\right)  }{2\pi e}\right)  \right)  +\ln\left(
\det\left(  \dfrac{\mathbf{G}_{2\text{,\thinspace}2}\left(  \mathbf{x}\right)
}{2\pi e}\right)  \right)  +\ln\left(  \det(\mathbf{I}_{K_{2}}-\mathbf{A}%
_{\mathbf{x}})\right)  \right\rangle _{\mathbf{x}}\text{.} \label{A.Thm3P.2}%
\end{align}
Then by the eigendecompositon of $\mathbf{A}_{\mathbf{x}}$, we have
\begin{equation}
\mathbf{A}_{\mathbf{x}}=\mathbf{U}_{\mathbf{x}}\boldsymbol{\Lambda
}_{\mathbf{x}}\mathbf{U}_{\mathbf{x}}^{T}\text{,} \label{A.Thm3P.3}%
\end{equation}
where $\mathbf{U}_{\mathbf{x}}$ and $\boldsymbol{\Lambda}_{\mathbf{x}}$\ are
$K_{2}\times K_{2}$ eigenvector matrix and eigenvalue matrix, respectively.
Since $\mathbf{G}\left(  \mathbf{x}\right)  $, $\mathbf{G}_{1\text{,\thinspace
}1}\left(  {\mathbf{x}}\right)  $ and $\mathbf{G}_{2\text{,\thinspace}%
2}\left(  {\mathbf{x}}\right)  $ are positive-definite, then $\mathbf{I}%
_{K}-\mathbf{A}_{\mathbf{x}}$\ is also positive-definite and $\mathbf{A}%
_{\mathbf{x}}$ is positive-semidefinite, with $0\leq\left(
\boldsymbol{\Lambda}_{\mathbf{x}}\right)  _{k\text{,\thinspace}k}=\lambda
_{k}<1$ for $\forall k\in\left\{  1\mathrm{,\,}2\mathrm{,\,}\cdots
\mathrm{,\,}K_{2}\right\}  $. Moreover, it follows from (\ref{Thm3.1}) that
\begin{equation}
\left\{
\begin{array}
[c]{l}%
{0}\leq\left\langle {\mathrm{Tr}\left(  \boldsymbol{\Lambda}_{\mathbf{x}%
}\right)  }\right\rangle _{\mathbf{x}}=\left\langle {\mathrm{Tr}\left(
\mathbf{A}_{\mathbf{x}}\right)  }\right\rangle _{\mathbf{x}}\ll1\text{,}\\
{0}\leq\left\langle {\mathrm{Tr}\left(  \boldsymbol{\Lambda}_{\mathbf{x}}%
^{m}\right)  }\right\rangle _{\mathbf{x}}=\left\langle {\sum\nolimits_{k=1}%
^{K_{2}}\lambda_{k}^{m}}\right\rangle _{\mathbf{x}}\leq{\left\langle
{\mathrm{Tr}\left(  \boldsymbol{\Lambda}_{\mathbf{x}}\right)  }\right\rangle
_{\mathbf{x}}}\ll1\text{.}%
\end{array}
\right.  \label{A.Thm3P.4}%
\end{equation}
Then by (\ref{A.Thm3P.4}) we have
\begin{equation}
\left\langle \ln\left(  \det(\mathbf{I}_{K_{2}}-\mathbf{A}_{\mathbf{x}%
})\right)  \right\rangle _{\mathbf{x}}=\left\langle {\mathrm{Tr}}\left(
\ln\left(  \mathbf{I}_{K_{2}}-\boldsymbol{\Lambda}_{\mathbf{x}}\right)
\right)  \right\rangle _{\mathbf{x}}=\sum\nolimits_{m=1}^{\infty}\dfrac{-1}%
{m}\left\langle {\mathrm{Tr}\left(  \boldsymbol{\Lambda}_{\mathbf{x}}%
^{m}\right)  }\right\rangle _{\mathbf{x}}\simeq0\text{.} \label{A.Thm3P.5a}%
\end{equation}
Substituting (\ref{A.Thm3P.5a}) into (\ref{A.Thm3P.2}) and then combining with
(\ref{IG}), we get (\ref{Thm3.2}).

If Eq.~(\ref{Thm3.3}) holds, then $\mathbf{A}_{\mathbf{x}}=\mathbf{0}$ and
$I_{G}={I}_{G_{1}}$. Conversely, if $I_{G}={I}_{G_{1}}$, then%
\begin{equation}
0=\left\langle \ln\left(  \det(\mathbf{I}_{K_{2}}-\mathbf{A}_{\mathbf{x}%
})\right)  \right\rangle _{\mathbf{x}}\leq-\left\langle {\mathrm{Tr}}\left(
\mathbf{A}_{\mathbf{x}}\right)  \right\rangle _{\mathbf{x}}\leq0\text{,}
\label{A.Thm3P.6}%
\end{equation}
$\mathbf{A}_{\mathbf{x}}=\mathbf{0}$, and Eq. (\ref{Thm3.3}) holds. This
completes the proof of \textbf{Theorem \ref{Theorem 3}.}%
\qed

\subsection{Proof of Theorem \ref{Theorem 4}}

\label{A.Thm4P}Similar to (\ref{A.Thm3P.2}), we have
\begin{align}
&  {\left\langle \ln\left(  \det\left(  \dfrac{\mathbf{G}(\mathbf{x})}{2\pi
e}\right)  \right)  \right\rangle _{\mathbf{x}}}\nonumber\\
&  ={\left\langle \ln\left(  \det\left(  \dfrac{\mathbf{G}_{1\text{,\thinspace
}1}\left(  \mathbf{x}\right)  }{2\pi e}\right)  \right)  +\ln\left(
\det\left(  \dfrac{\mathbf{P}_{2\text{,\thinspace}2}\left(  \mathbf{x}\right)
}{2\pi e}\right)  \right)  +\ln\left(  \det(\mathbf{I}_{K_{2}}+\mathbf{B}%
_{\mathbf{x}})\right)  \right\rangle _{\mathbf{x}}}\text{.} \label{A.Thm4P.1}%
\end{align}
Similar to (\ref{A.Thm1P.2}), the eigendecompositon of $\mathbf{B}%
_{\mathbf{x}}$ is given by%
\begin{equation}
\mathbf{B}_{\mathbf{x}}=\mathbf{U}_{\mathbf{x}}\boldsymbol{\Lambda
}_{\mathbf{x}}\mathbf{U}_{\mathbf{x}}^{T}\text{,} \label{A.Thm4P.2}%
\end{equation}
where $\mathbf{U}_{\mathbf{x}}$ and $\boldsymbol{\Lambda}_{\mathbf{x}}$\ are
$K_{2}\times K_{2}$ eigenvector matrix and eigenvalue matrix, respectively. If
the matrix $\mathbf{B}_{\mathbf{x}}$ is positive-semidefinite and satisfies
(\ref{Thm4.2}), then $\left(  \boldsymbol{\Lambda}_{\mathbf{x}}\right)
_{k\text{,\thinspace}k}=\lambda_{k}\geq0$ for $\forall k\in\left\{
1\mathrm{,\,}2\mathrm{,\,}\cdots\mathrm{,\,}K_{2}\right\}  $ and
\begin{align}
0  &  \leq\left\langle \ln\left(  \det(\mathbf{I}_{K_{2}}+\mathbf{B}%
_{\mathbf{x}})\right)  \right\rangle _{\mathbf{x}}=\left\langle \sum
\nolimits_{k=1}^{K_{2}}\ln\left(  1+\lambda_{k}\right)  \right\rangle
_{\mathbf{x}}\nonumber\\
&  \leq\left\langle {\mathrm{Tr}}\left(  \boldsymbol{\Lambda}_{\mathbf{x}%
}\right)  \right\rangle _{\mathbf{x}}={\mathrm{Tr}}\left(  \left\langle
\mathbf{B}_{\mathbf{x}}\right\rangle _{\mathbf{x}}\right)  \ll1\text{.}
\label{A.Thm4P.5}%
\end{align}
Substituting (\ref{A.Thm4P.5}) into (\ref{A.Thm4P.1}), we immediately get
(\ref{Thm4.3}). If $\mathbf{C}_{\mathbf{x}}=\mathbf{0}$, then $\ln\left(
\det(\mathbf{I}_{K_{2}}+\mathbf{B}_{\mathbf{x}})\right)  =0$ and
$I_{G}=I_{G_{2}}$. And if $I_{G}=I_{G_{2}}$, then $\ln\left(  \det
(\mathbf{I}_{K_{2}}+\mathbf{B}_{\mathbf{x}})\right)  =0$, $\mathbf{B}%
_{\mathbf{x}}=$ $\mathbf{0}$ and $\mathbf{C}_{\mathbf{x}}=\mathbf{0}$.%
\qed

\subsection{Proof of Corollary \ref{Corollary 3}}

Notice that
\begin{equation}
\left\{
\begin{array}
[c]{l}%
{H(X)=H(X_{1})+H(X_{2})}\text{,}\\
{H(X_{2})=\dfrac{1}{2}\ln\left(  \det({2\pi e}\boldsymbol{\Sigma}%
_{\mathbf{x}_{2}})\right)  }\text{,}\\
{\mathbf{P}_{2\text{,\thinspace}1}\left(  \mathbf{x}\right)  =\mathbf{P}%
_{1\text{,\thinspace}2}\left(  \mathbf{x}\right)  }=\mathbf{0}\text{,}\\
{\mathbf{P}_{2\text{,\thinspace}2}\left(  \mathbf{x}\right)
=\boldsymbol{\Sigma}_{\mathbf{x}_{2}}^{-1}}\text{,}%
\end{array}
\right.  \label{A.Cly3P.1}%
\end{equation}
and the matrices%
\begin{align}
\mathbf{C}_{\mathbf{x}}  &  =\mathbf{J}_{2\text{,\thinspace}2}\left(
{\mathbf{x}}\right)  -\mathbf{J}_{2\text{,\thinspace}1}\left(  {\mathbf{x}%
}\right)  \mathbf{G}_{1\text{,\thinspace}1}^{-1}\left(  {\mathbf{x}}\right)
\mathbf{J}_{1\text{,\thinspace}2}\left(  {\mathbf{x}}\right)  \text{,}%
\label{A.Cly3P.2a}\\
\mathbf{B}_{\mathbf{x}}  &  =\mathbf{P}_{2\text{,\thinspace}2}^{-1/2}\left(
{\mathbf{x}}\right)  \mathbf{C}_{\mathbf{x}}\mathbf{P}_{2\text{,\thinspace}%
2}^{-1/2}\left(  {\mathbf{x}}\right)  \label{A.Cly3P.2b}%
\end{align}
are positive-semidefinite, and the proof is similar to (\ref{3.24}). Then by
\textbf{Theorem \ref{Theorem 4}} we immediately get (\ref{Thm4.3}).
Substituting (\ref{A.Cly3P.1}) into (\ref{Thm4.3}) yields (\ref{Cly3.1}) with
strict equality if and only if $\mathbf{C}_{\mathbf{x}}=\mathbf{0}$. This
completes the proof of \textbf{Corollary \ref{Corollary 3}}.%
\qed

\subsection{Proof of Proposition \ref{Proposition 2}}

By writing $p({\boldsymbol{\theta}})$ as a sum of two density functions
$p_{1}({\boldsymbol{\theta}})$ and $p_{2}({\boldsymbol{\theta}})$,
\begin{equation}
p({\boldsymbol{\theta}})=\alpha p_{1}({\boldsymbol{\theta}})+\left(
1-\alpha\right)  p_{2}({\boldsymbol{\theta}})\text{,\thinspace}
\label{A.Ppn2P.1}%
\end{equation}
we have
\begin{equation}
\mathbf{G}(\mathbf{x})=N\int_{{{\Theta}}}p({\boldsymbol{\theta}}%
)\mathbf{S}(\mathbf{x}\text{;\thinspace}{\boldsymbol{\theta}}%
)d{\boldsymbol{\theta}}+\mathbf{P}(\mathbf{x})=\alpha\mathbf{G}_{1}%
(\mathbf{x})+(1-\alpha)\mathbf{G}_{2}(\mathbf{x})\text{,\thinspace}
\label{A.Ppn2P.2}%
\end{equation}
where $0\leq\alpha\leq1$ and
\begin{align}
\mathbf{G}_{1}{(\mathbf{x})}  &  {=N\int_{{{\Theta}}}p_{1}({\boldsymbol{\theta
}})}\mathbf{S}(\mathbf{x}\text{;\thinspace}{\boldsymbol{\theta}}%
)d{\boldsymbol{\theta}}+{\mathbf{P}(\mathbf{x})}\text{,}\label{A.Ppn2P.3a}\\
\mathbf{G}_{2}{(\mathbf{x})}  &  {=N\int_{{{\Theta}}}p_{2}({\boldsymbol{\theta
}})}\mathbf{S}(\mathbf{x}\text{;\thinspace}{\boldsymbol{\theta}}%
)d{\boldsymbol{\theta}}+{\mathbf{P}(\mathbf{x})}\text{.} \label{A.Ppn2P.3b}%
\end{align}
Using the Minkowski determinant inequality and the inequality of weighted
arithmetic and geometric means, we find
\begin{align}
{\det\left(  \mathbf{G}(\mathbf{x})\right)  ^{1/K}}  &  ={\det\left(
\alpha\mathbf{G}_{1}(\mathbf{x})+(1-\alpha)\mathbf{G}_{2}(\mathbf{x})\right)
^{1/K}}\nonumber\\
&  \geq\alpha\det\left(  \mathbf{G}_{1}(\mathbf{x})\right)  ^{1/K}%
+(1-\alpha)\det\left(  \mathbf{G}_{2}(\mathbf{x})\right)  ^{1/K}\nonumber\\
&  \geq\left(  \det\left(  \mathbf{G}_{1}(\mathbf{x})\right)  ^{\alpha}%
\det\left(  \mathbf{G}_{2}(\mathbf{x})\right)  ^{(1-\alpha)}\right)
^{1/K}\text{.} \label{A.Ppn2P.4}%
\end{align}
It follows from (\ref{A.Ppn2P.2}) and (\ref{A.Ppn2P.4}) that
\begin{equation}
\ln\left(  \det\left(  \alpha\mathbf{G}_{1}(\mathbf{x})+(1-\alpha
)\mathbf{G}_{2}(\mathbf{x})\right)  \right)  \geq\alpha\ln\left(  \det\left(
\mathbf{G}_{1}(\mathbf{x})\right)  \right)  +(1-\alpha)\ln\left(  \det\left(
\mathbf{G}_{2}(\mathbf{x})\right)  \right)  \text{,} \label{A.Ppn2P.5}%
\end{equation}
where the equality holds if and only if $\mathbf{G}_{1}(\mathbf{x}%
)=\mathbf{G}_{2}(\mathbf{x})$. Thus $\ln\left(  \det\left(  \mathbf{G}%
(\mathbf{x})\right)  \right)  $ is concave about $p({\boldsymbol{\theta}})$.
Therefore $I_{G}[p({\boldsymbol{\theta}})]$ is a concave function about
$p({\boldsymbol{\theta}})$. Similarly we can prove that $I_{F}%
[p({\boldsymbol{\theta}})]$ is also a concave function about
$p({\boldsymbol{\theta}})$. This completes the proof of \textbf{Proposition
\ref{Proposition 2}}. \qed

\phantomsection\addcontentsline{toc}{section}{References}
\bibliographystyle{apalike2}
\bibliography{NECO-08-17-2952R1-Source-BIB}

\end{document}